%% file: main.tex
\documentclass[conference,onecolumn,11pt,letter]{IEEEtran} 
\usepackage{amsmath,amstext, amsfonts, amssymb,url,pstricks,psfrag,graphicx,enumerate,dsfont}
\newtheorem{lemma}{Lemma} 
  
\newtheorem{theorem}{Theorem} 
\newtheorem{corollary}{Corollary}   
\newtheorem{definition}{Definition} 
\newenvironment{proofs}[1]{\flushleft{{\bf Proof [{#1}]:}}\\ }{\flushright{{\bf QED\\}}}  

\newcommand{\so}[1]   {{\sf o}{({#1})}}

\newcommand{\PX}[1]   {{\bf P}\!\left[{#1}\right]} 
\newcommand{\EX}[1]   {{\bf E}\!\left[{#1}\right]} 
\newcommand{\PCX}[2]  {{\bf P}\!\left[\left.\! {#1} \right| {#2} \right]} 
\newcommand{\ECX}[2]  {{\bf E}\!\left[\left.\! {#1} \right| {#2} \right]} 
\newcommand{\IND}[1]  {{\mathds{1}}_{\{ {#1} \} }}  
\newcommand{\DEF}[0]  {{\triangleq}}  
\newcommand{\supp}[1] {{\sf  supp}{{#1}}}       

\newcommand{\CX}[0] {{\cal C}}      
\newcommand{\DX}[0] {{\cal D}}      
\newcommand{\ENERG}[0] {{\cal S}}   
\newcommand{\AV}[0] {{\cal P}}      
\newcommand{\ENC}[0] {{\it X}}       

\newcommand{\blx}[0] {{\sf n} }      
\newcommand{\SC}[0] {{\mathcal Q}}   
\newcommand{\CTM}[0]{{W}}  
\newcommand{\CT}[2] {{\CTM}({#2}|{#1})}

\newcommand{\mes}[0]   {{\sf M}}      
\newcommand{\est}[0]   {\hat{{\sf M}}}

\newcommand{\dmes}[0]  {{\it m}}     
\newcommand{\mesS}[0]  {{\cal M}}

\newcommand{\inp}[0]   {{{{\sf X}}}}
\newcommand{\dinp}[0]  {{{\it x}}}
\newcommand{\inpS}[0]  {{{\cal X}}}

\newcommand{\out}[0]   {{{\sf Y}}}
\newcommand{\dout}[0]  {{{\it y}}}
\newcommand{\outS}[0]  {{{\cal Y}}}

\newcommand{\fsy}[0]   {{{\sf Z}}}
\newcommand{\dfsy}[0]  {{{\it z}}}
\newcommand{\fsyS}[0]  {{{\cal Z}}}

\newcommand{\rsy}[0]   {{{\sf A}}}

\newcommand{\ers}[1] {{\mu}_{#1}} 
\newcommand{\erst}[2] {\tilde{{\mu}}_{#1}({#2})}

\newcommand{\Era}[0]  {{\bf x}}
\newcommand{\Ero}[0]  {{\bf e}}

\newcommand{\Pe}[0]   {{\it P_{\Ero}}}         
\newcommand{\Per}[0]  {{\it P_{\Era}}}         
\newcommand{\Pem}[1]  {{\it P}_{\Ero|{#1}}}         
\newcommand{\Perm}[1] {{\it P}_{\Era|{#1}}}         

\newcommand{\PE}[0]   {{{\cal P}_{\Ero}}}        
\newcommand{\PERA}[0] {{{\cal P}_{0,\Era}}}      

\newcommand{\XT}[0]    {{\mathsf P}}  
\newcommand{\XCT}[0]   {{\mathsf V}}  
\newcommand{\XCTB}[0]  {{\mathsf U}}  
\newcommand{\Ctype}[2] {\XCT_{{#1}| {#2}}} 
\newcommand{\Ttype}[1] {\XCTB_{{#1}}} 
\newcommand{\type}[1] {\XT_{{#1}}} 
\newcommand{\XSOD}[0]  {{U}}
\newcommand{\XSID}[0]  {{\Pi}}

\newcommand{\Vshell}[2]  {T_{{#1}}\left( {#2} \right)     }            
\newcommand{\Cshell}[0]  {T_{\XSOD}}               
\newcommand{\Vsint}[4]  {F^{({#1})} \left({#2},{#3},{#4} \right)  }               

\newcommand{\ACW}[0]  {{\dinp_{\blx_1+1}^{\blx}(a)}}
\newcommand{\RCW}[0]  {{\dinp_{\blx_1+1}^{\blx}(r)}}
\newcommand{\YA}[0]   {\dout^{\blx_1}}
\newcommand{\YB}[0]   {\dout_{\blx_1+1}^{\blx}}

\newcommand{\ZA}[0]   {\dfsy^{\blx_1}}
\newcommand{\ZB}[0]   {\dfsy_{\blx_1+1}^{\blx}}

\newcommand{\ts}[0]  {\alpha}
\newcommand{\tst}[0]  {\tilde{\alpha}}
\newcommand{\por}[0]  {\succ}
\newcommand{\npor}[0]  {\nsucc}
\newcommand{\Prand}[4]  {\eta\left({#1},{#2},{#3},{#4}\right)}

\newcommand{\apri}[0]  {{\it \varphi}}
\newcommand{\XLD}[2]{\Omega \left(  {#1}, {#2} \right)}

\newcommand{\ACH}[0]  {{\Psi}}
\newcommand{\ach}[0]  {{\psi}}
\newcommand{\DECS}[0] {\Upsilon}
\newcommand{\Ymar}[1] {({#1} )_{Y}}

\newcommand{\PSet}[3]    {{\cal P} \left({#1},{#2},{#3}\right)}
\newcommand{\Cset}[0]    {{\cal V}}
\newcommand{\Eraset}[0]  {{\cal V_{\Era}}}
\newcommand{\Eroset}[0]  {{\cal V_{\Ero}}}

\newcommand{\XQ}[3]  {{{\it U}_{{#1}}({#2}|{#3})}}
\newcommand{\Xq}[1]  {{{\it U}_{{#1}}}}
\newcommand{\WS}[2]  {{\bf {\it S}}_{{#1},{#2}}}
\newcommand{\XPT}[2]    {{\it P}_{{#1}}\left[{#2}\right]} 
\newcommand{\XPTC}[3]   {{\it P}_{#1}\left[\left.\! {#2} \right| {#3} \right]} 
\newcommand{\XET}[2]    {{\it E}_{{#1}}\left[{#2}\right]} 
\newcommand{\XETC}[3]   {{\it E}_{#1}\left[\left.\! {#2} \right| {#3} \right]}

\newcommand{\cesp}[0]     {{\bf \zeta}   }
\newcommand{\cespal}[2]   {{R}^{*}_{l}({#1},{#2})}
\newcommand{\cespah}[2]   {{R}^{*}_{h}({#1},{#2})}
\newcommand{\hexp}[0]     {{E}_{H}}
\newcommand{\hexpn}[0]    {\tilde{{E}}_{H}}
\newcommand{\rexp}[0]     {{E}_{r}}
\newcommand{\spexp}[0]    {{E}_{sp}}
\newcommand{\bhexp}[2]    {{\Gamma} \left({#1},{#2} \right)}
\newcommand{\obhexp}[1]   {{\Gamma} \left({#1}\right)}

\newcommand{\TEX}[0]   {{\cal E}_{\Ero}} 
\newcommand{\TEXn}[0]  {{\cal E}}        
\newcommand{\TEXz}[0]  {{\cal E}_{\Era}} 
\newcommand{\XEtwo}[0] {{\cal E}_{\Era,2}} 

\newcommand{\EXo}[0]  {E_{\Ero}}
\newcommand{\EXor}[0]  {\tilde{E}_{\Ero}}
\newcommand{\EXa}[0]  {{E_{\Era}}}

\newcommand{\MI}[2]   {{\sf I}\left({#1} , {#2} \right) }   

\newcommand{\CENT}[2] {{\sf H} \left( \left.\! {#1} \right |{#2} \right)}
\newcommand{\DIV}[2]  {{\sf D}\left(\left.{#1} \right\Vert {#2} \right) }   
\newcommand{\CDIV}[3] {{\sf D}\left(\left.{#1} \right\Vert {#2} \vert {#3} \right) }  

\begin{document} 
\title{Errors-and-Erasures Decoding for Block Codes with Feedback} 

\author{\authorblockN{Bar\i\c{s} Nakibo\u{g}lu}
\authorblockA{
 Electrical  Engineering and Computer Science\\ 
 Massachusetts Institute of Technology 
\\MA 02139 Cambridge}
\and
\authorblockN{Lizhong Zheng}
\authorblockA{
Electrical Engineering and Computer Science\\ 
 Massachusetts Institute of Technology 
 \\MA 02139 Cambridge 
}}

\maketitle 
\begin{abstract} 
Inner and outer bounds are derived on the optimal performance of fixed length block codes on discrete memoryless channels  with feedback  and errors-and-erasures decoding. First an inner bound is derived using a two phase encoding scheme with communication and control phases together with the optimal decoding rule for the given encoding scheme,  among decoding rules that can be represented in terms of pairwise comparisons between the messages.  Then an outer bound is derived  using a generalization of the  straight-line bound to errors-and-erasures decoders and the optimal error exponent trade off of a  feedback encoder with two messages.  In addition  upper and lower bounds are derived, for the optimal erasure exponent of error free block codes in terms of the rate. 
Finally we present  a proof of the fact that the optimal  trade off between error exponents of a two message code  does not increase with feedback  on  DMCs. 
\end{abstract}

\section{Introduction:}\label{sec:introduction}
Shannon showed in \cite{sha} that the capacity of discrete memoryless channels (DMCs) does not increase even when  a noiseless and delay free feedback link is available from the receiver to the   transmitter.  On symmetric DMCs the sphere packing exponent bounds the error exponent of fixed length block codes from above, as  shown by Dobrushin\footnote{Later Haroutunian, \cite{har}, established an upper bound on the error exponent of block codes with feedback. This upper bound is equal to sphere packing exponent for symmetric channels but it is  strictly larger than the sphere packing exponent for non-symmetric channels.}  in \cite{dob}.  Thus relaxations like errors-and-erasures decoding or variable length coding are needed for feedback to increase the error exponent of block codes at rates larger than the critical rate on symmetric DMCs. In this work we investigate one such relaxation, namely  errors-and-erasures decoding and find inner and outer bounds to the optimal error exponent erasure exponent trade off. 

Finding the optimal encoding and decoding schemes,  and hence finding optimal performance by characterizing the surface of achievable error exponent erasure exponent pairs is an important motivation for the investigation of errors-and-erasure decoding.  Note, however,  that  finding the optimal performance with erasures  will implicitly solve the problem of finding the optimal feedback encoder and determining the error exponent for the erasure free fixed length block codes with feedback which is a long standing open problem.  Finding the optimal performance, however, is far from being the only important aspect of the problem. Determining the performance of feedback encoding schemes that are easier to implement,  more robust to the degradations of the feedback link   and bounding the loss in the performance compared to the  more complicated encoding schemes  are  both important tasks practically  and interesting ones  intellectually. This will be our aim in this paper. We will first analyze the performance of a two phase encoding scheme inspired by the optimal encoding schemes for variable length block codes and derive inner bounds to the optimal performance. Then we will derive outer bounds to the performance of general feedback encoding schemes with erasures and quantify the loss of performance by restricting ourselves to the above mentioned two phase schemes. This analysis complements the research on two related block coding schemes: variable length block coding and errors-and-erasures decoding for block codes without feedback. We start with a very brief overview of the previous work on these problems to motivate our investigation further.

Burnashev \cite{bur,bur1,bur2} was the first one to consider  variable-length block codes with feedback, instead of fixed length ones. He  obtained the exact expression for the error exponent at all rates. Later Yamamoto and Itoh, \cite{ito}, suggested a coding scheme which achieves the best error exponent for variable-length block codes with feedback by using a fixed length block code with  an errors-and-erasures decoding and repeating the same codeword  until a non-erasure decoding occurs.\footnote{Including erasures will not  increase  the exponent for variable-length block codes with feedback.} In fact any fixed length block code with erasures can be used in this repetitive fashion, like it was done in \cite{ito},  to get a variable length block code with essentially the same error exponent as the original fixed length block code. Thus \cite{bur} can be reinterpreted to give an upper bound to the error exponent achievable by fixed length block codes with erasures. Furthermore  this upper bound is achieved by the fixed length block codes with erasures described  in \cite{ito}, when  erasure probability is decaying to zero subexponentially with block length. However the techniques used in this line of work are insufficient for deriving proper inner or outer bounds  for the situation when erasure probability is decaying exponentially with block length. As explained in the following paragraph  the case with strictly positive erasure exponent is important both for engineering applications and for a  better  understanding of soft decoding with feedback.  Our investigation provides proper tools for such an analysis, results in  inner and outer bounds to the trade off between error and erasure exponents, while recovering all previously known results for the zero erasure exponent case.

When considered together with higher layers, the codes in the physical layer are part of a variable length/delay communication scheme with feedback. However in the physical layer itself fixed length block codes are used instead of variable length ones because of their amenability to modular design and robustness  against the noise in the feedback link. In such an architecture retransmissions  affect the performance of higher layers. The average transmission time is only a first order measure of this effect: as long as the erasure probability is vanishing with increasing block length, average transmission time will essentially be equal to the block length of the fixed length block code. Thus with an analysis like the one in \cite{ito},  the  cost of retransmissions are ignored as long as the erasure probability goes to zero with increasing block length. In a communication system with multiple layers, however, retransmissions usually  have costs beyond their effect on average transmission time, which are described  by constraints on the probability distribution of the decoding time. Knowledge of error erasure exponent trade off is useful in coming up with designs to meet those constraints.  An example of this phenomena is  variable length block coding schemes with  hard deadlines for decoding time, which has already been investigated by Gopala {\it et. al. } \cite{GYK} for   block codes without feedback. They have used a block coding scheme with erasures and they  resend the message whenever an erasure occurs. But because of the hard deadline, they employ this scheme only for some fixed number of trials. If all those trials fail, i.e. lead to an erasure, they use a non-erasure block code. Using the error exponent erasure exponent trade off they were able to obtain the best over all error performance for the given architecture. 

 This brings us to the second line of research we complement with our investigation: errors-and-erasures decoding for block codes without feedback. Forney  \cite{for} was the first one to consider errors-and-erasures decoding without feedback. He obtained an achievable trade off  between the exponents of  error and erasure probabilities.  Then Csisz\'ar and K\"orner, \cite{CK} achieved the same performance using universal coding and decoding algorithms.  Later Telatar and Gallager, \cite{EG}, introduced a strict improvement on certain channels over the results presented in \cite{for} and \cite{CK}. Recently there has been a revived interest in the errors-and-erasures decoding  for universally achievable performances \cite{moulin,Feder}, for  alternative methods of analysis \cite{merhav}, for extensions to the channels with side information \cite{erez} and implementation with linear block codes \cite{hof}. The encoding schemes in these codes do not have access to any feedback. However  if the transmitter can learn whether the decoded message was an erasure or not, it can resend the message whenever it is erased. Because of this block retransmission variant, these problems are sometimes called   decision feedback problems.

 We complement the  results on the error exponent erasure exponent trade off without feedback and the results about error exponent of variable length block codes with feedback,  by finding inner and outer bounds to the error exponent erasure exponent trade off of fixed length block codes with feedback.  We  first introduce  our model and notation in Section \ref{sec:model}.  Then in Section \ref{sec:ach} we derive  a  lower  bound using  a two phase coding algorithm similar to the one described by Yamamoto and Ito in \cite{ito} and  decoding rule and analysis  techniques, inspired by Telatar's in \cite{emre} for the non-feedback case.  Note that the analysis and the decoding rule in \cite{emre} is tailored for a single phase scheme and without feedback and the two phase scheme of \cite{ito} is tuned specifically to zero-erasure exponent; coming up with framework in which both of the ideas can be used efficiently is the main technical challenge here.  In Section \ref{sec:converse} we first  extend the  straight line bound  idea introduced by Shannon, Gallager and Berlekamp in \cite{SGB} to block codes with erasures.  Then we use it together with the outer bound on the error exponent trade off between two codewords with feedback to establish an outer bound for the error exponent of fixed length block codes with feedback and erasures.   In Section \ref{sec:efree} we first introduce error free  block codes with erasures and discuss their relation to the fixed length block codes with errors-and-erasures-decoding, and then  we  present inner and outer bounds to the erasure exponent of error free block codes and point out its relation to the error exponent erasure exponent trade off.

Before presenting  our analysis, let us make a brief digression  and discuss two channel models in which the use of feedback had been investigated for block codes without erasures. First channel model is the well known  additive white Gaussian noise channel (AWGNC) model. In AWGNCs if the power constraint $\AV$ is on the expected value of the energy spent on a block $\EX{\ENERG_{\blx}}$ i.e. power constraint is of the form  $\EX{\ENERG_{\blx}} \leq \AV \blx$, the error probability can be made to decay faster than any exponential function with block length $\blx$.  Schalkwijk and Kailath suggested a coding algorithm in \cite{sch1} which achieves a doubly exponential decay in error probability for continuous time AWGNCs, i.e. infinite bandwidth case. Later Schalkwijk \cite{sch2} modified that scheme to achieve  the same performance in discrete time AWGNCs, i.e. finite bandwidth case.   Concatenating Schalkwijk and Kailath scheme with pulse amplitude modulation stages, gives  a multi-fold exponential decrease in the error probability \cite{pin,zig,RGnew}.  However this behavior relies on the absence of any amplitude limit,  the particular form of the power constraint and the noise free nature of the feedback link. First of all, as observed in \cite{bur2} and \cite{NG} when there is an amplitude limit,  error probability decays only exponentially with block length. More importantly if the power constraint restricts the energy spent in transmission of each message for all  noise realizations, i.e. if the power constraint is an almost sure power constraint\footnote{As Kim {\it et. al.} \cite{KLW} calls it.} of the form $\ENERG_{\blx} \leq \AV \blx$; then  sphere packing exponent is still an upper bound to the error exponent  for AWGNCs as shown by Pinsker, \cite{pin}. Furthermore if the feedback link is also an AWGNC and if there is a  power constraint\footnote{This constraint can be an expected or almost sure constraint.} on the feedback transmissions, then  even in the case when there are only two messages, error probability decays only exponentially as it has been recently shown  by Kim {\it et.al.} \cite{KLW}.

The second channel model is the DMC model. Although feedback can not increase the error exponent for rates over the critical rate, it can simplify the encoding scheme  \cite{zig,yak}. Furthermore, for rates below the critical rate it is possible to improve the error exponent using feedback. Zigangirov \cite{zig} has  established  lower bounds to the error exponent for BSCs using a simple encoding scheme.  Zigangirov's lower bound is equal to the sphere packing exponent for all rates in the interval $[R_{crit}^{'},\CX]$ where $R_{crit}^{'}<R_{crit}$ and  Zigangirov's lower bound is strictly larger than the corresponding non-feedback exponent for rates below $R_{crit}^{'}$. Later Burnashev \cite{burbsc}  introduced an improvement to Zigangirov's bound for all positive rates less than $R_{crit}^{'}$.  D'yachkov \cite{yak}  generalized Zigangirov's encoding scheme for general DMCs and established a lower bound to the error exponent for general  binary input channels and k-ary symmetric channels.  However it is still an open problem to find a constructive  technique that can be used for all DMCs which outperforms the random coding bound.  Like AWGNCs there has been a revived interest in the effect of a noisy feedback link and achievable performances with  noisy feedback  on  DMCs. Burnashev and Yamamoto recently showed that error exponent of BSC channel increases even with a noisy feedback link \cite{BYJ}, \cite{BY}.  Furthermore Draper and Sahai \cite{DS} investigated the use of noisy feedback link  in variable length schemes.

\section{Model and Notation:}\label{sec:model}
The input and output alphabets of the forward channel are $\inpS$ and $\outS$, respectively. The channel input and output symbols at time $t$ will be denoted by $\inp_t$ and $\out_t$ respectively. Furthermore, the sequences of input and output symbols from time $t_1$ to time $t_2$ are denoted by  $\inp_{t_1}^{t_2}$ and $\out_{t_1}^{t_2}$. When $t_1=1$ we omit $t_1$ and simply write $\inp^{t_2}$ and $\out^{t_2}$ instead of  $\inp_{1}^{t_2}$ and $\out_{1}^{t_2}$. The forward channel is  a stationary  memoryless channel characterized by an $|\inpS|$-by-$|\outS|$ transition probability matrix $\CTM$.   
\begin{equation}   
\label{eq:chans}
\PCX{\out_t}{\inp^t,\out^{t-1}}=\PCX{\out_t}{\inp_t}= \CT{\inp_t}{\out_t}  \qquad \forall t.
 \end{equation}
The feedback channel is  noiseless and  delay free i.e. the input of the feedback channel $\fsy_{t-1}$,  chosen at the receiver, is observed at the transmitter before transmission of  $\inp_t$. In addition we assume that feedback channel is of infinite capacity thus $\fsy_{t-1}$  includes all of the observation of the receiver at time $t-1$, i.e.\footnote{For $t=1$ we have  $\fsy_{0}=\rsy_{0}$}     $\fsy_{t-1}=(\out_{t-1},\rsy_{t-1})$.  The random variables $\rsy_0, \rsy_1,\ldots,\rsy_{\blx}$ are  there to  enable randomized encoding and decoding schemes as we will see shortly. It is assumed that the choice $\rsy$'s does not affect  the forward channels behavior, i.e. in addition to (\ref{eq:chans}) we have\footnote{We make a slight abuse of notation and denote $\fsy_{0}^{t}$ by $\fsy^{t}$.}
\begin{equation}   
\label{eq:chans2}
\PCX{\out_t}{\inp^t,\fsy^{t-1}}= \CT{\inp_t}{\out_t}  \qquad \forall t.
 \end{equation}
The message $\mes$ is drawn from the message set $\mesS$  with a uniform probability distribution and is given to the transmitter at time zero.   At each time $t \in [1,\blx]$ the input symbol ${\ENC}_t(\mes,\fsy^{t-1})$ is sent. The sequence of functions $\ENC_t(\cdot): \mesS \times \fsyS^{t-1}$ which assigns an input symbol for each $\dmes \in \mesS$ and $\dfsy^{t-1}\in\fsyS^{t-1}$ is called the encoding function. Note that the random variables $\rsy_0,\rsy_1,\ldots,\rsy_{\blx-1}$  enable randomized encoding schemes.  
 After receiving $\out^{\blx}$ the receiver draws the final $\rsy$, i.e. $\rsy_{\blx}$, and   decodes to the message $\est(\fsy^{\blx}) \in \{\Era\} \cup \mesS$ where $\Era$ is the erasure symbol. The  random variable $\rsy_{\blx}$ does not have any effect on the   encoding; it is used only   to enable randomized decoding schemes.

  The conditional  error and erasure probabilities $\Pem{\mes}$  and $\Perm{\mes}$ and unconditional  error and erasure  probabilities,  $\Pe$  and $\Per$ are defined as,
\begin{align*}
 \Pem{\mes} &\triangleq \PCX{\est\neq \mes}{\mes}-\Perm{\mes}
 &
 \Perm{\mes} &\triangleq \PCX{\est=\Era}{\mes}  
 \\ 
 \Pe & \triangleq  \PX{\est\neq \mes} -\Per
 &
 \Per&\triangleq \PX{\est=\Era}  
 \end{align*}
Since all the messages are equally likely we have,
\begin{align*}
\Pe&=\tfrac{1}{|\mesS|}\sum\nolimits_{\dmes}\Pem{\dmes}
&
\Per&=\tfrac{1}{|\mesS|} \sum\nolimits_{\dmes}    \Perm{\dmes}. 
\end{align*}
We use a somewhat abstract but rigorous approach in defining the rate and achievable exponent pairs. A reliable sequence $\SC$, is a sequence of codes indexed by their block lengths such that
\begin{equation*}
\lim_{\blx \rightarrow \infty} (\Pe^{(\blx)}+ \Per^{(\blx)}+\tfrac{1}{|\mesS^{(\blx)}|})=0.
\end{equation*}
In other words reliable sequences are sequences of codes whose overall error probability, detected and undetected,  vanishes and whose size of message set grows to infinity  with block length $\blx$.
\begin{definition}\label{def:sec}
 The rate,  erasure exponent,  and error exponent of a reliable sequence $\SC$ are given by
 \begin{equation*}
    {R}_{\SC}   \triangleq \liminf_{\blx \rightarrow \infty} \tfrac{\ln |\mesS^{(\blx)}|}{\blx}  \qquad
 {\EXa}_{\SC}  \triangleq \liminf_{\blx \rightarrow \infty} \tfrac{-\ln \Per^{(\blx)}}{\blx}      \qquad  
 {\EXo}_{\SC}  \triangleq \liminf_{\blx \rightarrow \infty} \tfrac{-\ln \Pe^{(\blx)}}{\blx}.    
 \end{equation*}
\end{definition}
Haroutunian, \cite[Theorem 2]{har}, has already established a strong converse for erasure free block codes with feedback which in our setting implies that $\lim_{\blx \rightarrow \infty} (\Pe^{(\blx)}+\Per^{(\blx)})=1$ for all codes whose rates are strictly above the capacity, i.e. $R>\CX$. Thus we  consider only rates that are less than or equal to the capacity, $R \leq \CX$. For all rates $R$ below capacity and for all non-negative erasure exponents $\EXa$, we  define the (true) error exponent  $\TEX(R,\EXa)$ of fixed length block codes with feedback to be the best error exponent of the reliable sequences\footnote{We  restrict ourselves to the reliable sequences in order to ensure  finite error exponent at zero erasure exponent. Note that a decoder which always declares erasures has zero erasure exponent and infinite error exponent.} whose rate is at least $R$ and whose erasure exponent is at least $\EXa$. 
\begin{definition}\label{def:exp} 
$\forall R \leq \CX$ and $ \forall \EXa \geq 0$ the error exponent, $  \TEX (R,\EXa)$ is,
\begin{equation} 
 \TEX (R,\EXa) \triangleq
 \sup_{\SC:  R_{\SC} \geq R,  {\EXa}_{\SC} \geq \EXa}  {\EXo}_{\SC}.
\end{equation}
\end{definition}
Note  that 
\begin{equation}
  \label{eq:highera}
\TEX(R,\EXa)=\TEXn(R)   \qquad \forall  \EXa > \TEXn(R)  
\end{equation}
 where $\TEXn(R)$ is the (true) error exponent of erasure-free block codes on DMCs with feedback.\footnote{In order to see this consider a reliable sequence with erasures $\SC$ and replace its decoding algorithm by an erasure free decoding algorithm such that  $\est'(\dfsy^{\blx}) =  \est(\dfsy^{\blx})$ if $\est(\dfsy^{\blx}) \neq \Era$, to obtain a new reliable sequence  $\SC'$. Then $\Pe^{(\blx)}_{\SC'} \leq \Per^{(\blx)}_{\SC} + \Pe^{(\blx)}_{\SC}$; thus ${\EXo}_{\SC'}=\min \{ {\EXa}_{\SC} , {\EXo}_{\SC}\}$ and $R_{\SC'}=R_{\SC}$. This together with the definition of $\TEXn(R)$ leads to  equation  (\ref{eq:highera}). }  Thus benefit of the  errors-and-erasures decoding is  the possible increase in the error exponent as the erasure exponent goes below $\TEXn(R)$.

Determining $\TEXn(R)$ for all $R$'s  and for all channels is still an open problem; only upper and lower bounds  to $\TEXn(R)$ are known. Our investigation  focuses on quantifying the gains of errors-and-erasures decoding instead of finding $\TEXn(R)$. Consequently, we  restrict ourselves to the region where the erasure exponent is lower than the error exponent for the encoding scheme. 

For future reference let us recall the expressions for the random coding exponent and the sphere packing exponent,
\begin{align}
\rexp(R,P)&=\min_{V} \CDIV{V}{\CTM}{P}+ |\MI{P}{V}-R|^{+} & \rexp(R)&=\max_{P} \rexp(R,P)
\label{eq:exp}\\ 
\spexp(R,P)&=\min_{V: \MI{P}{V}\leq R} \CDIV{V}{\CTM}{P}    & \spexp(R)&=\max_{P} \spexp(R,P)
\label{eq:spexp}
\end{align}
where $\CDIV{V}{\CTM}{P}$ stands for conditional Kullback Leibler divergence of $V$ and $\CTM$ under $P$,  and $\MI{P}{V}$ stands for mutual information for input distribution $P$ and channel $V$.

 We denote the $\dout$ marginal of a distribution like $P(\dinp)V(\dout|\dinp)$ by $\Ymar{PV}$.  The support of a probability distribution $P$ is denoted by $\supp{P}$.

\section{An Achievable Error Exponent - Erasure Exponent Trade Off}\label{sec:ach}
In this section we  establish a lower bound to the achievable error exponent as a function of  erasure exponent and rate. We  use a two phase encoding scheme similar to the one described by Yamamoto and Ito in \cite{ito} together with a decoding rule similar to the one described by Telatar in \cite{emre}.  In the first phase, the transmitter  uses a fixed-composition code of length $\ts \blx$ and rate $\tfrac{R}{\ts}$. At the end of the first phase, the receiver makes a maximum mutual information decoding to obtain a tentative decision $\tilde{\mes}$. The transmitter knows $\tilde{\mes}$ because of the feedback link. In the remaining $(\blx-\blx_1)$ time units, i.e. the second phase, the transmitter confirms the tentative decision by sending the accept codeword, if $\tilde{\mes}=\mes$, and  rejects it by sending the reject codeword  otherwise. At the end of the second phase the receiver either declares an erasure or declares the tentative decision as the decoded message. Receiver declares the tentative decision as the decoded message only when the tentative decision ``dominates'' all other messages. The word ``dominate'' will be made precise later in Section \ref{subsec:coding}. Our scheme is inspired by \cite{ito} and \cite{emre}. However, unlike \cite{ito} our decoding rule makes use of outputs of both of the phases instead of output of just second phase while deciding between declaring an erasure or declaring the tentative decision as the final one, and unlike \cite{emre} our encoding scheme is a feedback encoding scheme with two phases.

In the rest of this section, we analyze  the performance of this coding architecture and  derive an achievable error exponent expression in terms of a given rate $R$, erasure exponent $\EXa$, time sharing constant $\ts$, communication phase type $P$, control phase type (joint empirical type of the accept codeword and reject codeword) $\XSID$ and domination rule $\por$. Then we  optimize over $\por$,  $\XSID$, $P$ and  $\ts$, to obtain an achievable error exponent expression as a function of rate $R$ and erasure exponent $\EXa$.
\subsection{Fixed-Composition Codes and  The Packing Lemma}
We start with a very brief overview of certain properties of types.  Those readers who are not familiar with method types can use  \cite{CT} for  a concise introduction   or   \cite{CK} for   a thorough study. The empirical distribution of an $\dinp^{\blx} \in {\inpS}^{\blx}$ is called the type of $\dinp^{\blx}$ and the empirical distribution of transitions from a $\dinp^{\blx} \in {\inpS}^{\blx}$ to a  $\dout^{\blx} \in {\outS}^{\blx}$ is called the conditional type:\footnote{Note that $\type{\dinp^{\blx}}$  corresponds to  a distribution on $\inpS$ for all $\dinp^{\blx} \in {\inpS}^{\blx}$, where as  $ \Ctype{\dout^{\blx}}{\dinp^{\blx}}$ determines a channel from the support of $\type{\dinp^{\blx}}$ to $\outS$.}
\begin{align}
\label{eq:deftype}
\type{\dinp^{\blx}}(\tilde{\dinp})     
&\DEF 
\tfrac{1}{\blx}   \sum_{t=1}^{\blx} \IND{\dinp_t=\tilde{\dinp}}
&& \tilde{\dinp} \in \inpS.\\
\label{eq:defctype}
\Ctype{\dout^{\blx}}{\dinp^{\blx}} (\tilde{\dout}|\tilde{\dinp})    
&\DEF
\tfrac{1}{\blx \type{\dinp^{\blx}}(\tilde{\dinp})}   \sum_{t=1}^{\blx} \IND{\dinp_t=\tilde{\dinp}}\IND{\dout_t=\tilde{\dout}}
&& \forall  \tilde{\dout}\in \outS,~~\forall \tilde{\dinp}~\mbox{s.t.}~\type{\dinp^{\blx}}(\tilde{\dinp})>0.
 \end{align}
For any probability  transition matrix  $W:\supp{\type{\dinp^{\blx}}} \rightarrow {\outS}$ we have\footnote{Note that for any $W:{\inpS} \rightarrow {\outS} $ there is unique consistent    $W':\supp{\type{\dinp^{\blx}}} \rightarrow {\outS}$.}
\begin{align}
\label{eq:pb1}
\prod_{t=1}^{\blx} W(y_t|x_t)= e^{-\blx (\CDIV{\Ctype{y^{\blx}}{x^{\blx}}}{W}{\type{x^{\blx}}}+\CENT{\Ctype{y^{\blx}}{x^{\blx}}}{\type{x^{\blx}}})}
 \end{align}
The set of  all $\dout^{\blx}$'s with the same conditional type $V$ with respect to $\dinp^{\blx}$ is called the $V$-shell of $\dinp^{\blx}$ and  denoted by  $\Vshell{V}{\dinp^{\blx}}$:
 \begin{equation}
    \Vshell{V}{\dinp^{\blx}} =\{\dout^{\blx} : \Ctype{\dout^{\blx}}{\dinp^{\blx}} =V  \}.
 \end{equation}
Note that for any  transition probability matrix  from $\inpS$ to $\outS$  total probability of  $\Vshell{V}{\dinp^{\blx}}$ has to be less than one. Thus by assuming that transition probabilities are $V$ and  using equation (\ref{eq:pb1}) we can conclude that,
\begin{align}
\label{eq:pb2}
 | \Vshell{V}{\dinp^{\blx}} | \leq e^{\CENT{\Ctype{\dout^{\blx}}{\dinp^{\blx}}}{\type{\dinp^{\blx}}}}
 \end{align}
Codes  whose codewords all have the same empirical distribution, $\type{\dinp^{\blx}(\dmes)}=P$  $\forall \dmes\in \mesS$ are called fixed-composition codes.  In Section \ref{subsec:eran}  we will  describe the error and erasure events in terms of the intersections of $V-$shells of different codewords. For doing that let us define  $\Vsint{\blx}{V}{\hat{V}}{\dmes}$ as the intersection of $V$-shell of $\dinp^{\blx}(\dmes)$ and  the $\hat{V}$-shells of other codewords:
  \begin{equation}
\label{eq:Vsintdef}
 \Vsint{\blx}{V}{\hat{V}}{\dmes} \triangleq \Vshell{V}{\dinp^{\blx}(\dmes)}   \bigcap \cup_{\tilde{\dmes}\neq \dmes}  \Vshell{\hat{V}}{\dinp^{\blx}(\tilde{\dmes})}.
  \end{equation}
The following packing lemma, proved by Csisz\'ar and K\"orner \cite[Lemma 2.5.1]{CK},  claims the existence of a code with a guaranteed upper bound on the size of $\Vsint{\blx}{V}{\hat{V}}{\dmes}$.
\begin{lemma}
\label{lem:pac}
   For every block length $\blx\geq 1$, rate $R>0$  and type $P$  satisfying $H(P)>R$, there exist at least $\lfloor e^{\blx (R-\delta_{\blx})} \rfloor $ distinct  type $P$ sequences
in ${\inpS}^\blx$ such that for every pair of stochastic matrices $V: \supp{P} \rightarrow \outS$, $\hat{V}:  \supp{P} \rightarrow \outS $ and  $\forall \dmes \in \mesS$
   \begin{equation*}
\left\lvert  \Vsint{\blx}{V}{\hat{V}}{\dmes} \right\rvert \leq \lvert  \Vshell{V}{\dinp^{\blx}(\dmes)} \rvert e^{- \blx \lvert I(P, \hat{V})-R \rvert^+  } 
   \end{equation*}
where $\delta_{\blx}=\tfrac{\ln 4+ (4 |\inpS| +6 |\inpS| |\outS|) \ln(\blx+1)}{\blx}$.
\end{lemma}
Above lemma is stated in a slightly different way by the authors of  \cite{CK}, for a fixed $\delta$ and large enough $\blx$. However, this form follows  immediately  from  their proof.

  If we use Lemma \ref{lem:pac} together with equations (\ref{eq:pb1}) and (\ref{eq:pb2}) we can bound the conditional probability of observing a $\dout^{\blx} \in \Vsint{\blx}{V}{\hat{V}}{\dmes}$ when  $\mes=\dmes$  as follows.
\begin{corollary}
\label{cor:cor1}
In a code satisfying Lemma \ref{lem:pac}, when message $\dmes \in \mesS$ is sent, the probability of receiving a $\dout^{\blx} \in  \Vshell{V}{\dinp^{\blx}(\dmes)}$ which is also in $\Vshell{\hat{V}}{\dinp^{\blx}(\tilde{\dmes})}$, for some $\tilde{\dmes}\in \mesS$ such that  $\tilde{\dmes}\neq \dmes$ is bounded as follows,
  \begin{align}
 \PCX{\Vsint{\blx}{V}{\hat{V}}{\mes}}{\mes}
&\leq e^{-\blx \Prand{R}{P}{V}{\hat{V}}}   
\label{eq:boundonf}
  \end{align}
where
\begin{equation}
\label{eq:prand}
 \Prand{R}{P}{V}{\hat{V}} \DEF \CDIV{V}{\CTM}{P}+ \lvert \MI{P}{\hat{V}}-R \rvert^+
\end{equation}
\end{corollary}
\subsection{Coding Algorithm}\label{subsec:coding}
In the first phase, the communication phase, we use a length $\blx_1=\lceil \ts \blx \rceil$ type $P$ fixed-composition code with $\lfloor e^{\blx_1 ( \frac{R}{\ts}-\delta_{\blx_1})} \rfloor $ codewords which satisfies the property described in Lemma \ref{lem:pac}.  At the end of the first phase the receiver makes a tentative decision by choosing the codeword that has the maximum empirical mutual information with the output sequence $\out^{\blx_1}$.  If there is a tie, i.e. if there are  more than one codewords which have the maximum empirical mutual information, the receiver chooses the codeword which has the lowest index. 
\begin{equation}
\label{eq:encrule}
  \tilde{\mes}= \left\{\dmes  :
    \begin{array}{l}
\MI{P}{\Ctype{\out^{\blx_1}}{\dinp^{\blx}(\dmes)}} >     \MI{P}{\Ctype{\out^{\blx_1}}{\dinp^{\blx}(\tilde{\dmes})}} \quad \forall \tilde{\dmes}<\dmes \\
\MI{P}{\Ctype{\out^{\blx_1}}{\dinp^{\blx}(\dmes)}} \geq  \MI{P}{\Ctype{\out^{\blx_1}}{\dinp^{\blx}(\tilde{\dmes})}} \quad \forall \tilde{\dmes}>\dmes      
          \end{array}
\right\}
\end{equation}
In the remaining $(\blx-\blx_1)$ time units, the transmitter sends the accept codeword $\dinp_{\blx_1+1}^{\blx}(a)$ if $\tilde{\mes}=\mes$ and sends the reject codeword $\dinp_{\blx_1+1}^{\blx}(r)$ otherwise.

Note that our encoding scheme uses the feedback link actively  for the encoding neither within the first phase nor within the second phase. It does not even change the codewords it uses for accepting or rejecting the tentative decision depending on the observation in the first phase. Feedback is only used to reveal the tentative decision to the transmitter.

Accept and reject codewords have joint type $\XSID(\tilde{\dinp},\tilde{\tilde{\dinp}})$, i.e. the ratio of the number of time instances in which accept codeword has an $\tilde{\dinp} \in \inpS$ and reject codeword has a $\tilde{\tilde{\dinp}} \in \inpS$ to the length of the codewords, $(\blx-\blx_1)$, is $\XSID(\tilde{\dinp},\tilde{\tilde{\dinp}})$. The joint conditional type of the output sequence in the second phase, $\Ttype{\dout_{\blx_1+1}^{\blx}}$, is the empirical conditional distribution of $\dout_{\blx_1+1}^{\blx}$. We call set of all output sequences $\dout_{\blx_1+1}^{\blx}$ whose joint conditional type is $\XSOD$, the $\XSOD$-shell and denote it by  $\Cshell$.

Like we did in the  Corollary \ref{cor:cor1}, we can upper bound the probability of $\XSOD$-shells. Note that if $\out_{\blx_1+1}^{\blx}  \in \Cshell$ then,
\begin{align*}
\PCX{\out_{\blx_1+1}^{\blx}}{\inp_{\blx_1+1}^{\blx}=\ACW}  &=e^{-(\blx-\blx_1)(\CDIV{\XSOD}{W_a}{\XSID}+ \CENT{\XSOD}{\XSID})} \\
\PCX{\out_{\blx_1+1}^{\blx}}{\inp_{\blx_1+1}^{\blx}=\RCW}  &=e^{-(\blx-\blx_1)(\CDIV{\XSOD}{W_r}{\XSID}+ \CENT{\XSOD}{\XSID})}  
\end{align*}
where $\ACW$ is the accept codeword, $\RCW$ is the reject codeword,  $W_a(\dout|\tilde{x},\tilde{\tilde{x}})= \CT{\tilde{x}}{\dout}
$ and $W_r(\dout|\tilde{x},\tilde{\tilde{x}})= \CT{\tilde{\tilde{x}}}{\dout}$. Noting that $|T_{\XSOD}| \leq  e^{(\blx-\blx_1)\CENT{\XSOD}{\XSID}}$,  we get:
\begin{subequations}
\begin{align}
  \PCX{\Cshell}{{\inp_{\blx_1+1}^{\blx}}=\ACW} &\leq e^{-(\blx-\blx_1)\CDIV{\XSOD}{W_a}{\XSID}} 
\label{eq:boundonda}\\
  \PCX{\Cshell}{{\inp_{\blx_1+1}^{\blx}}=\RCW} &\leq e^{-(\blx-\blx_1)\CDIV{\XSOD}{W_r}{\XSID}}.
\label{eq:boundondr}
\end{align}
\end{subequations}
\subsection{Decoding Rule}\label{subsec:decoding}
For an  encoder like the one in Section \ref{subsec:coding}, a decoder  that  depends  only  on the conditional type of $\out^{\blx_1}$ for different codewords in the  communication  phase, i.e. $\Ctype{\out^{\blx_1}}{\dinp^{\blx_1}(\dmes)}$'s for $\dmes\in \mesS$, the conditional type of the channel output in the control phase, i.e. $\Ttype{\out_{\blx_1+1}^{\blx}}$, and the indices of the codewords can achieve the minimum error probability   for a given erasure probability. However finding that decoder becomes analytically intractable. Instead, we restrict ourselves to the decoders that can be written in terms of pair wise comparisons between messages given $\out^{\blx}$. Furthermore we assume that these pairwise comparisons depend only on the conditional type of $\out^{\blx_1}$ for  the messages compared,  the conditional output type in the control phase and the indices of the messages. Thus if the triplet corresponding to the tentative decision  $(\Ctype{\out^{\blx_1}}{\dinp^{\blx_1}(\tilde{\mes})}, \Ttype{\out_{\blx_1+1}^{\blx}},\tilde{\mes})$  dominates all other triplets of the form  $(\Ctype{\out^{\blx_1}}{\dinp^{\blx_1}(\dmes)}, \Ttype{\out_{\blx_1+1}^{\blx}},\dmes)$ for $\dmes\neq\tilde{\mes}$,  the tentative decision becomes final; else an erasure is declared.\footnote{Note that conditional probability, $\PCX{\out^{\blx}}{\mes=\dmes}$, is only a function of corresponding  $\Ctype{\out^{\blx_1}}{\dinp^{\blx}(\dmes)}$ and $ \Ttype{\out_{\blx_1+1}^{\blx}}$.  Thus all decoding rules, that accepts or rejects the tentative decision, $\tilde{\mes}$, based on a threshold test on likelihood ratios, $\frac{\PCX{\out^{\blx}}{\mes=\tilde{\mes}}}{\PCX{\out^{\blx}}{\mes=\dmes}}$, for $\dmes\neq \tilde{\mes}$ are in this family of decoding rules.} 
\begin{equation}
\label{eq:dec-rule}
\est=  \begin{Bmatrix} 
\tilde{\mes} & \mbox{if } \forall  \dmes \neq  \tilde{\mes}   \qquad
(\Ctype{\out^{\blx_1}}{\tilde{\mes}}, \Ttype{\out_{\blx_1+1}^{\blx}},\tilde{\mes})  \por (\Ctype{\out^{\blx_1}}{\dmes}, \Ttype{\out_{\blx_1+1}^{\blx}},\dmes) \\ 
\Era         & \mbox{if } \exists   \dmes \neq \tilde{\mes}    \mbox{ s.t. }   
(\Ctype{\out^{\blx_1}}{\tilde{\mes}}, \Ttype{\out_{\blx_1+1}^{\blx}},\tilde{\mes})  \npor (\Ctype{\out^{\blx_1}}{\dmes}, \Ttype{\out_{\blx_1+1}^{\blx}},\dmes) 
  \end{Bmatrix}
 \end{equation}
The binary relation $\por$ is such that  if  $(V,\XSOD,\dmes)$ dominates $(\hat{V},\XSOD,\tilde{\dmes})$ then  $(\hat{V},\XSOD,\tilde{\dmes})$ does not dominate $(V,\XSOD,\dmes)$:
\begin{equation*}
 (V,\XSOD,\dmes) \por(\hat{V},\XSOD,\tilde{\dmes}) \Rightarrow (\hat{V},\XSOD,\tilde{\dmes}) \npor (V,\XSOD,\dmes).
\end{equation*}
This property is a necessary and sufficient condition for a binary relation to be a domination rule. Decoder given by (\ref{eq:dec-rule}), however, either accepts or rejects the tentative decision $\tilde{\mes}$ given in (\ref{eq:encrule}). Consequently its domination rule also satisfies following two properties:
\begin{enumerate}[(a)]
\item If the empirical mutual information of the messages in  the communication phase are not  equal, only the message with larger mutual information  can dominate the other one.
\item If the empirical mutual information of the messages  in  the communication phase are      equal, only the message with lower  index               can dominate the other one.
\end{enumerate}
 For any such binary relation there is a corresponding decoder of the form given in equation (\ref{eq:dec-rule}). In our scheme we  either use the trivial domination rule leading to the trivial decoder $\est=\tilde{\mes}$ or the domination rule given in equation (\ref{eq:optord}), both of which satisfies these conditions.
{\small
\begin{equation}
\label{eq:optord}
(V,\XSOD,\dmes) \por(\hat{V},\XSOD,\tilde{\dmes})
\Leftrightarrow
\begin{cases}
\MI{P}{V}>\MI{P}{\hat{V}}     \mbox{~and~}\ts \Prand{\tfrac{R}{\ts}}{P}{V}{\hat{V}}+(1-\ts) \CDIV{\XSOD}{W_a}{\XSID}  \leq  \EXa  &\mbox{if~} \dmes\geq \tilde{\dmes}\\
\MI{P}{V}\geq \MI{P}{\hat{V}} \mbox{~and~}\ts \Prand{\tfrac{R}{\ts}}{P}{V}{\hat{V}}+(1-\ts) \CDIV{\XSOD}{W_a}{\XSID}  \leq  \EXa  &\mbox{if~} \dmes <   \tilde{\dmes}
\end{cases}
\end{equation}}
where $\Prand{R}{P}{V}{\hat{V}}$ is given by the equation (\ref{eq:prand}).

Among the family of decoders we are considering, i.e. among  the decoders that only depend on the pairwise comparisons between conditional types  and indices of the messages compared, the decoder given in (\ref{eq:dec-rule}) and (\ref{eq:optord}) is optimal in terms of error exponent erasure exponent trade off.  Furthermore,  in order to employ this  decoding rule, the receiver needs to determine only the two messages with the highest empirical mutual information in the first phase.  Then the receiver needs to check whether the triplet corresponding to the tentative decision dominates the triplet corresponding to the message with the second highest empirical mutual information. If it does then, for the rule given in (\ref{eq:optord}), it is guaranteed to dominate the rest of the triplets too.

\subsection{Error Analysis}\label{subsec:eran}
Using an encoder like the one described in Section \ref{subsec:coding} and a decoder like the one in  (\ref{eq:dec-rule}) we  achieve the performance given below. If $\EXa\leq \ts \rexp(\tfrac{R}{\ts},P)$ then the domination rule given in equation (\ref{eq:optord}) is used in the decoder; else a trivial domination rule that leads to a erasure-free decoding, $\est=\tilde{\mes}$, is used in the decoder.
\begin{theorem}\label{thm:ach1}
  For any block length $\blx\geq 1$, rate $R$, erasure exponent $\EXa$, time sharing constant $\ts$, communication phase type $P$ and control phase type  $\XSID$, there exists a length $\blx$ block code with feedback such that
\begin{equation*}
  \ln |{\mesS} | \geq e^{\blx(R-\delta_{\blx})} \qquad \Per\leq e^{-\blx(\EXa-\delta_\blx^{'})} \qquad \Pe\leq e^{-\blx(\EXo (R, \EXa, \ts, P, \XSID)-\delta_\blx^{'})}
\end{equation*}
where $\EXo(R,\EXa,\ts,P,\XSID)$ is given by,
\begin{subequations}
\label{eq:thm1}
\begin{align}   
\EXo
&\!=\!\left\{
\begin{array}{cl}
\ts \rexp(\tfrac{R}{\ts},P)  &\mbox{if~} \EXa>\ts \rexp(\tfrac{R}{\ts},P)  \\
\hspace{-1cm}\displaystyle{\min_{\substack{(V,\hat{V},\XSOD):(V,\hat{V},\XSOD) \in \Cset\\ \hspace{1cm}\ts \Prand{\tfrac{R}{\ts}}{P}{V}{\hat{V}}+(1-\ts) \CDIV{\XSOD}{W_a}{\XSID} \leq \EXa}}} \hspace{-1cm}
\ts \Prand{\tfrac{R}{\ts}}{P}{\hat{V}}{V}+(1-\ts) \CDIV{\XSOD}{W_r}{\XSID}
&\mbox{if~}\EXa\leq \ts \rexp(\tfrac{R}{\ts},P)
\end{array}\right\}\\
\Cset
&\!=\!\left\{(V_1,V_2,\XSOD) : \MI{P}{V_1} \geq \MI{P}{V_2} ~\mbox{and}~  \Ymar{P V_1}=\Ymar{P V_2} \right\}\\
\delta_{\blx}^{'}
&\!=\!\tfrac{(|{\inpS}|+1)^2|{\outS}|\log(\blx+1)}{\blx}
\end{align}
\end{subequations}
\end{theorem} 
The optimization problem given in (\ref{eq:thm1}) is a convex optimization problem: it is  minimization of a convex function over a convex set. Thus the value of the exponent, $\EXo (R, \EXa, \ts, P, \XSID)$ can numerically be  calculated relatively easily. Furthermore $\EXo (R, \EXa, \ts, P, \XSID)$ can  be written in terms of solutions of lower dimensional optimization problems (see equation (\ref{eq:altoptdef}).  However problem of finding the optimal  $(\ts,P,\XSID)$ triple for a given $(R,\EXa)$ pair is not that easy in general, as we will discuss in more detail in Section \ref{sec:how}. 

Note that for all control phase types $\XSID$ and control phase output types $\XSOD$, $\CDIV{\XSOD}{W_a}{\XSID}\geq 0$, $\CDIV{\XSOD}{W_r}{\XSID}\geq 0$. Using this fact together with the definitions of  $\rexp(R,P)$, $\Prand{R}{P}{\hat{V}}{V}$  and $\EXo(R,\EXa,\ts,P,\XSID)$  given in (\ref{eq:exp}), (\ref{eq:prand}) and (\ref{eq:thm1}) we get:
\begin{align}
\label{eq:trs}
\EXo(R,\EXa,\ts,P,\XSID)&\geq \ts \rexp(\tfrac{R}{\ts},P)  &&\forall (R,\EXa,\ts,P,\XSID) \mbox{~s.t.~} \EXa\leq  \ts \rexp(\tfrac{R}{\ts},P) 
\end{align}
Since we are interested in quantifying the gains of errors-and-erasures decoding over the decoding schemes without erasures   we are ultimately interested only  in the region where $\EXa \leq \ts \rexp(\tfrac{R}{\ts},P)$ holds. However equation (\ref{eq:thm1}) gives us the whole achievable region for the family of codes we are considering.

\begin{proof}
A decoder of the form given in (\ref{eq:dec-rule})  decodes correctly when $\tilde{\mes}=\mes$ and $(\out^{\blx}, \mes) \por (\out^{\blx},\dmes)$ for all\footnote{We use the short hand  $(\out^{\blx}, \mes) \por (\out^{\blx},\dmes)$ for  $(\Ctype{\out^{\blx_1}}{\mes}, \Ttype{\out_{\blx_1+1}^{\blx}},\mes)  \por (\Ctype{\out^{\blx_1}}{\dmes}, \Ttype{\out_{\blx_1+1}^{\blx}},\dmes)$ in the rest of this section.} $\dmes \neq  \mes$.  Thus  an error or an erasure  occur only when the correct message does not dominate all other messages, i.e. when $ \exists \dmes  \neq \mes $ such that $ (\out^{\blx}, \mes) \npor (\out^{\blx},\dmes)$. Consequently, we can write the sum of conditional error and erasure probabilities for a message $\dmes \in  {\mesS}$ as,   
\begin{equation}
\label{eq:ero-erap}
\Pem{\dmes} +\Perm{\dmes} = \PCX{ \left\{\dout^{\blx} : \exists \tilde{\dmes} \neq \dmes \mbox{ s.t.} (\dout^{\blx}, \dmes) \npor (\dout^{\blx},\tilde{\dmes}) \right\}}{\mes=\dmes} 
 \end{equation}
This can happen in two ways, either there is an error in the first phase, i.e. $\tilde{\mes}\neq \dmes$  or first phase tentative decision is correct, i.e. $\tilde{\mes} = \dmes$, but the second phase observation $\YB$ leads to an erasure i.e. $\est=\Era$.  For a decoder using a domination rule satisfying constraints described in Section \ref{subsec:decoding},
\begin{align*}
\Pem{\dmes}+\Perm{\dmes} 
&\leq \sum_{V} \sum_{\substack{\hat{V}:\MI{P}{\hat{V}}\geq  \MI{P}{V}}} \sum_{\substack{\YA \in \Vsint{\blx_1}{V}{\hat{V}}{\dmes}}} \PCX{\YA}{\dmes}\\
& \qquad +\sum_{V}\sum_{\hat{V}: \MI{P}{\hat{V}} \leq \MI{P}{V}}\sum_{\substack{\YA \in \Vsint{\blx_1}{V}{\hat{V}}{\dmes}}} \PCX{\YA}{\dmes}
\sum_{\XSOD: (V,\XSOD,\dmes) \npor  (\hat{V},\XSOD,\dmes+1)  }\sum_{\YB \in \Cshell} \PCX{\YB}{\ACW}.
\end{align*}
where\footnote{Note that for the case when $\dmes=|{\mesS}|$, we need to replace $(V,\XSOD,\dmes) \npor  (\hat{V},\XSOD,\dmes+1)$ with $(V,\XSOD,\dmes-1) \npor  (\hat{V},\XSOD,\dmes)$.} $\Vsint{\blx_1}{V}{\hat{V}}{\dmes}$ is the intersection of $V$-shell of message $\dmes \in \mesS$ with the $\hat{V}$-shells of other messages, defined in equation (\ref{eq:Vsintdef}). As a result of Corollary \ref{cor:cor1} we have
\begin{align*}
\sum_{\YA \in \Vsint{\blx_1}{V}{\hat{V}}{\dmes}} \PCX{\YA}{\dmes}
&= \PCX{\Vsint{\blx_1}{V}{\hat{V}}{\dmes}}{\mes=\dmes}  \\ 
&\leq e^{-\blx_1 \Prand{\frac{R}{\ts}}{P}{V}{\hat{V}}}.   
\end{align*}
Furthermore because of equation (\ref{eq:boundonda}) 
\begin{align*}
\sum_{\YB \in \Cshell} \PCX{\YB}{\ACW}     
&=\PCX{\Cshell}{{\inp_{\blx_1+1}^{\blx}}=\ACW}\\ 
&\leq e^{-(\blx-\blx_1)\CDIV{\XSOD}{W_a}{\XSID}}. 
\end{align*}
In addition the number of different non-empty $V$-shells  in the communication phase is less than $(\blx_1+1)^{|{\inpS}||{\outS}|}$ and the number of non-empty $\XSOD$-shells in the control phase is less than $(\blx-\blx_1+1)^{|{\inpS}|^2|{\outS}|}$. We denote the set of  $(V,\hat{V},\XSOD)$ triples that corresponds to erasures with  a correct tentative decision  by $\Eraset$:
\begin{equation}
\label{eq:deferas}
\Eraset \triangleq \left\{(V,\hat{V},\XSOD): \MI{P}{V} \geq  \MI{P}{\hat{V}} ~\mbox{and}~ \Ymar{P V}=\Ymar{P \hat{V}} ~\mbox{and}~ (V,\XSOD,\dmes)\npor (\hat{V},\XSOD,\dmes+1) \right\}.
\end{equation}
In the above definition $\dmes$ is a dummy variable and $\Eraset$ is the same set for all $\dmes\in \mesS$. Thus using (\ref{eq:deferas}) we get
\begin{align*}
\Pem{\dmes}+\Perm{\dmes}
&\leq (\blx_1+1)^{2|{\inpS}||{\outS}|}\max_{V,\hat{V}: \MI{P}{V}\leq \MI{P}{\hat{V}}} e^{-\blx_1 \Prand{R/\ts}{P}{V}{\hat{V}}}\\
&+(\blx_1+1)^{2|{\inpS}||{\outS}|}(\blx-\blx_1+1)^{|{\inpS}|^2|{\outS}|}
\max_{(V,\hat{V},\XSOD) \in \Eraset}
e^{-\blx_1 (\Prand{R/\ts}{P}{V}{\hat{V}})-(\blx-\blx_1) \CDIV{\XSOD}{W_a}{\XSID}}.
\end{align*}
Using the definition of $\rexp(\tfrac{R}{\ts},P)$ given in (\ref{eq:exp}) we get
\begin{equation}
\label{eq:erasurebound}
\Pem{\dmes}+\Perm{\dmes}
\leq e^{\blx\delta_{\blx}^{'}} \max\left \{ e^{-\blx \ts \rexp(R/\ts,P)}, e^{-\blx \min_{(V,\hat{V},\XSOD) \in \Eraset } \ts \Prand{R/\ts}{P}{V}{\hat{V}}+(1-\ts) \CDIV{\XSOD}{W_a}{\XSID}}\right\}. 
 \end{equation}
On the other hand an error occurs only when an incorrect message dominates all other messages, i.e.  when $\exists \tilde{\dmes} \neq \dmes $ such that  $(\out^{\blx}, \tilde{\dmes}) \por (\out^{\blx},\tilde{\tilde{\dmes}})$ for all $\tilde{\dmes} \neq \tilde{\tilde{\dmes}}$:
\begin{equation*}
\Pem{\dmes} = \PCX{ \left\{\dout^{\blx}: \exists \tilde{\dmes} \neq \dmes \mbox{~s.t.~} (\dout^{\blx}, \tilde{\dmes}) \por (\dout^{\blx},\tilde{\tilde{\dmes}}) \quad  \forall \tilde{\tilde{\dmes}} \neq \tilde{\dmes} \right\}}{\mes=\dmes}.
 \end{equation*}
Note that when a $\tilde{\dmes} \in {\mesS}$ dominates all other $\tilde{\tilde{\dmes}}\neq \tilde{\dmes}$, it also dominates $\dmes$, i.e.
\begin{equation*}
\left\{\dout^{\blx}: \exists \tilde{\dmes} \neq \dmes \mbox{ s.t.} (\dout^{\blx}, \tilde{\dmes}) \por (\dout^{\blx},\tilde{\tilde{\dmes}}) \quad  \forall \tilde{\tilde{\dmes}} \neq \tilde{\dmes} \right\}  \subset \left\{\dout^{\blx}: \exists \tilde{\dmes} \neq \dmes \mbox{ s.t.} (\dout^{\blx}, \tilde{\dmes}) \por (\dout^{\blx},\dmes) \right\}.
\end{equation*}
Thus,
\begin{align}
\Pem{\dmes} 
&\leq \PCX{ \left\{\dout^{\blx}: \exists \tilde{\dmes} \neq \dmes \mbox{ s.t.} (\dout^{\blx}, \tilde{\dmes}) \por (\dout^{\blx},\dmes) \right\}}{\mes=\dmes} \notag\\
&= \sum_{V} \sum_{\hat{V}:\MI{P}{\hat{V}} \geq   \MI{P}{V}}
\sum_{\YA \in \Vsint{\blx_1}{V}{\hat{V}}{\dmes}} \PCX{\YA}{\mes=\dmes}
\sum_{\XSOD: (\hat{V},\XSOD,\dmes-1) \por  (V,\XSOD,\dmes)  } \sum_{\YB \in \Cshell} \PCX{\YB}{\RCW}.
\label{eq:errorboundbir}
\end{align}
The tentative decision is not equal to $\dmes$ only if there is a message with a strictly higher empirical mutual information or if there is a messages which has equal mutual information but smaller index. This is the reason why we sum over $ (\hat{V},\XSOD,\dmes-1) \por  (V,\XSOD,\dmes)$. Using the inequality (\ref{eq:boundondr}) in the inner most two sums and  then applying inequality (\ref{eq:boundonf})  we get,
 \begin{align}
\Pem{\dmes} 
&\leq (\blx+1)^{(|{\inpS}|^2+2|{\inpS}|)|{\outS}|}  \max_{
(V,\hat{V},\XSOD): \substack{\MI{P}{\hat{V}} \geq  \MI{P}{V}\\  (\hat{V},\XSOD,\dmes-1) \por  (V,\XSOD,\dmes) }
}
e^{-\blx (\ts \Prand{R/\ts}{P}{V}{\hat{V}}+(1-\ts) \CDIV{\XSOD}{W_r}{\XSID})}
\notag\\
&\leq e^{\blx \delta_{\blx}^{'}}e^{-\blx \min_{(\hat{V},V,\XSOD) \in \Eroset}  (\ts \Prand{R/\ts}{P}{V}{\hat{V}}+(1-\ts) \CDIV{\XSOD}{W_r}{\XSID})}
\notag\\
&= e^{\blx \delta_{\blx}^{'}}e^{-\blx \min_{(V,\hat{V},\XSOD) \in \Eroset}  (\ts \Prand{R/\ts}{P}{\hat{V}}{{V}}+(1-\ts) \CDIV{\XSOD}{W_r}{\XSID})}
\label{eq:errorbound}
 \end{align}
where $\Eroset$ is the complement of $\Eraset$ in $\Cset$ given by
\begin{equation}
\label{eq:deferos}
  \Eroset \triangleq \left\{(V,\hat{V},\XSOD) : \MI{P}{V}\geq \MI{P}{\hat{V}} ~\mbox{and}~  \Ymar{P V}=\Ymar{P\hat{ V}} ~\mbox{and}~ (V,\XSOD,\dmes)\por (\hat{V},\XSOD,\dmes+1) \right\}.
\end{equation}
Note that $\dmes$ in the definition of $\Eroset$ is also a dummy variable. The domination rule $\por$ divides  the set $\Cset$  into two subsets: the erasure subset $\Eraset$ and the error subset $\Eroset$.  Choosing domination rule  is equivalent to choosing the  $\Eroset$. Depending on the value of $\ts \rexp(\tfrac{R}{\ts},P)$ and $\EXa$ we chose different $\Eroset$'s as follows:
\begin{enumerate}[(i)]
\item$\EXa\!>\!\ts\!\rexp(\tfrac{R}{\ts},P)$: $\Eroset=\Cset$.  Then $\Eraset=\emptyset$ and Theorem \ref{thm:ach1} follows from equation (\ref{eq:erasurebound}).
\item$\EXa\!\leq\!\ts\!\rexp(\tfrac{R}{\ts},P)$: $\Eroset=\left\{(V,\hat{V},\XSOD):\begin{array}{c}
  \MI{P}{V} \geq  \MI{P}{\hat{V}}\!\mbox{~and~}\!\Ymar{P V}=\Ymar{P \hat{V}}\!\mbox{~and}\\ \ts \Prand{\tfrac{R}{\ts}}{P}{V}{\hat{V}}+(1-\ts) \CDIV{\XSOD}{W_a}{\XSID} \leq  \EXa  \end{array}\right\}$. Then all the $(V,\hat{V},\XSOD)$ triples satisfying $\ts\!\Prand{\tfrac{R}{\ts}}{P}{V}{\hat{V}}\!+\!(1\!-\!\ts) \CDIV{\XSOD}{W_a}{\XSID} \leq  \EXa$ are in the the error subset. Thus as a result of equation (\ref{eq:erasurebound}) erasure probability is bounded as $\Per\leq e^{-\blx(\EXa-\delta_\blx^{'})}$ and  Theorem \ref{thm:ach1} follows from equation (\ref{eq:errorbound}). 
\end{enumerate}
\end{proof}

\subsection{Lower Bound to $\TEX(R,\EXa)$:}\label{sec:how}
In this section we use Theorem \ref{thm:ach1} to derive a lower bound to the optimal error exponent $\TEX(R,\EXa)$. We do that by optimizing the achievable performance $\EXo(R,\EXa,\ts,P,\XSID)$ over
 $\ts$, $P$  and $\XSID$.  

\subsubsection{High Erasure Exponent Region (i.e. $\EXa> \rexp(R)$)}
As a result of (\ref{eq:thm1}), $\forall R\geq 0$ and $\forall \EXa>\rexp(R)$
\begin{subequations}
\begin{align}
\EXo(R,\EXa,\ts,P,\XSID)
&= \ts \rexp(\tfrac{R}{\ts},P)\leq \rexp(R)
&&\forall \ts \in [0,1],&& \forall P, &&\forall \XSID\\
\EXo(R,\EXa,\tilde{\ts},\tilde{P},\XSID)
&=\rexp(R)
&&\tilde{\ts}=1, &&\tilde{P}=\arg\max_{P} \rexp(R,P), &&
\forall\XSID.
\end{align}
\end{subequations}
Thus for all $(R,\EXa)$ pairs such that  $\EXa> \rexp(R)$: optimal time sharing constant is 1, optimal input distribution is the optimal input distribution for random coding exponent at rate $R$, we use maximum mutual information decoding and never declare erasures. Furthermore since $\ts=1$ we have only a single phase in our scheme. 
\begin{equation}
\label{eq:pc00}
  \EXo(R,\EXa)=\EXo (R,\EXa,1, P_{r(R)}, \XSID)=\rexp(R)  \qquad  \forall R\geq 0 \qquad  \forall \EXa> \rexp(R)
\end{equation}
where $P_{r(R)}$ satisfies $\rexp(R,P_{r(R)})=\rexp(R)$ and $\XSID$ can be any control phase type. Evidently benefits of errors-and-erasures decoding is not observed in this region.

\subsubsection{Low Erasure Exponent Region (i.e. $\EXa \leq \rexp(R)$)}
We observe and quantify the benefits of errors-and-erasures decoding for $(R,\EXa)$ pairs such that $\EXa\leq \rexp(R)$.  Since $\rexp(R)$ is a non-negative non-increasing and convex function of $R$, we have  
\begin{equation*}
\ts \in [\ts^*(R, \EXa),1 ] \Leftrightarrow   \EXa \leq \ts \rexp(\tfrac{R}{\ts})  \qquad \forall R\geq 0\qquad \forall 0<\EXa \leq \rexp(R)
\end{equation*}
where $\ts^*(R, \EXa)$ is the unique solution of the equation $\ts E_{r} (\tfrac{R}{\ts}) = \EXa$.

 For the case $\EXa=0$, however, $\ts E_{r} (\tfrac{R}{\ts}) =0$ has multiple solutions and Theorem \ref{thm:ach1} holds but resulting  error exponent, $\EXo(R,0,\ts,P,\XSID)$, does not correspond to the error exponent of a reliable sequence. Convention  introduced below in equation (\ref{eq:tsop})  addresses both issues at once, by choosing the minimum of those solutions as $\ts^*(R,0)$.  In addition by this convention $\ts^*(R,\EXa)$ is also continuous at $\EXa=0$: $\lim_{\EXa \rightarrow 0} \ts^*(R, \EXa)=\ts^*(R, 0)$.
\begin{equation}
  \label{eq:tsop}
\ts^*(R, \EXa) \triangleq
\begin{cases}
\tfrac{R}{g^{-1}(\EXa/R)}  
& \EXa\in (0,\rexp(R)]\\
R/\CX
& \EXa=0
\end{cases}
\end{equation}
where $g^{-1}(\cdot)$ is the inverse of the function $g(r)=\tfrac{\rexp(r)}{r}$.

As a result equations (\ref{eq:thm1}) and (\ref{eq:tsop}), $\forall R\geq 0$ and $\forall 0<\EXa\leq \rexp(R)$ we have 
\begin{subequations}
\label{eq:pc0}
\begin{align}
\EXo(R,\EXa,\ts,P,\XSID)
&=\ts \rexp(\tfrac{R}{\ts},P)\leq \rexp(R)  
&&\forall \ts \in [0,\ts^*(R, \EXa)), &&\forall P, &&\forall \XSID \\
\EXo(R,\EXa,\tilde{\ts},\tilde{P},\XSID)
&=\rexp(R)
&&\tilde{\ts}=1, &&\tilde{P}=\arg\max_{P}\rexp(R,P), &&\forall \XSID.
\end{align}
\end{subequations}
Thus for all  $(R,\EXa)$ pairs such that  $\EXa\leq \rexp(R)$ optimal time sharing constant is in the interval $[\ts^*(R, \EXa),1]$. 

For an $(R,\EXa,\ts)$ triple such that $R\geq 0$, $\EXa\leq \rexp(R)$ and $\ts \in [\ts^*(R, \EXa),1]$   let $\PSet{R}{\EXa}{\ts}$ be
\begin{equation}
\label{eq:psetdef}
\PSet{R}{\EXa}{\ts} \triangleq \{ P :  \ts \rexp(\tfrac{R}{\ts},P)\geq \EXa~,~\MI{P}{W}\geq \tfrac{R}{\ts} \}. 
\end{equation}
The constraint on mutual information is there to ensure that $\EXo(R,0,\ts,P,\XSID)$'s are corresponding to error exponent of reliable sequences. The set $\PSet{R}{\EXa}{\ts}$ is convex because  $\rexp(R,P)$ and $\MI{P}{W}$ are  concave in $P$.

Note that  $\forall R \geq 0$ and $\forall \EXa\in (0,  \rexp(R)]$,
\begin{subequations}
\label{eq:pc1}
  \begin{align}
\EXo(R,\EXa,\ts,P,\XSID)
&=\ts \rexp(\tfrac{R}{\ts},P) 
&&\forall \ts \in [\ts^*(R, \EXa),1], &&\forall P \notin \PSet{R}{\EXa}{\ts}, &&\forall \XSID\\
\EXo(R,\EXa,\ts,\tilde{P},\XSID)
&\geq \ts \rexp(\tfrac{R}{\ts})
&&\forall \ts \in [\ts^*(R, \EXa),1],  &&\tilde{P}=\arg \max_{P} \rexp(\tfrac{R}{\ts},P), &&\forall\XSID.
 \end{align}
\end{subequations}
As a result of (\ref{eq:pc1})  we can restrict the optimization over $P$ to $\PSet{R}{\EXa}{\ts}$ when $ R \geq 0$ and $\EXa\in (0,  \rexp(R)]$. For $\EXa=0$ case  if we require the expression $\EXo(R,0,\ts,P,\XSID)$ to correspond to the error exponent of a reliable sequence, we get the restriction given in equation (\ref{eq:pc1}). Thus  using  Theorem \ref{thm:ach1} we conclude that   $\EXo(R,\EXa)$ given below is an achievable error exponent at rate $R$  and erasure exponent $\EXa$.
\begin{equation}
\label{eq:pc3}
  \EXo(R,\EXa)=\max_{\ts\in [\ts^*(R,\EXa) ,1]} \max_{P \in \PSet{R}{\EXa}{\ts}} \max_{\XSID}   \EXo (R,\EXa, \ts, P, \XSID) \qquad  \forall R\geq 0 \qquad  \forall \EXa\leq \rexp(R)
\end{equation}
where $\ts^*(R,\EXa)$, $\PSet{R}{\EXa}{\ts}$ and $\EXo (R,\EXa, \ts, P, \XSID)$ are given in equations (\ref{eq:tsop}), (\ref{eq:psetdef}) and (\ref{eq:thm1}). 

Note that  unlike $\EXo(R,\EXa,\ts,P,\XSID)$ itself, $\EXo(R,\EXa)$ as defined in (\ref{eq:pc3}) corresponds to error exponent of reliable code sequences even at $\EXa=0$.

If the  maximizing $P$ for the inner maximization in equation (\ref{eq:pc3}) is same for all $\ts \in [\ts^*(R,\EXa),1]$, the  optimal value of $\ts$ is $\ts^{*}(R,\EXa)$. In order to see that, we first observe that any  fixed $(R,\EXa, P,\XSID)$ such that $\rexp(R,P)\geq \EXa$, function $\EXo (R,\EXa, \ts, P, \XSID)$  is convex in $\ts$ for all $\ts\in [\ts^*(R,\EXa,P),1]$ where $\ts^*(R,\EXa,P)$ is  the unique solution of the equation\footnote{Evidently we need to make a minor modification for $\EXa=0$ case as before to ensure that we consider only the $\EXor (R,\EXa, \ts, P, \XSID)$'s that correspond to the reliable sequences: $\ts^*(R,0,P)=\tfrac{R}{\MI{P}{W}}$.} $\ts \rexp(\tfrac{R}{\ts},P)= \EXa$ as it is shown Lemma \ref{lem:conveityints} in Appendix \ref{app:convexityints}. Since the maximization preserves the convexity, $\max_{\XSID}\EXo (R,\EXa, \ts, P,\XSID)$ is also  convex in $\ts$ for all $\ts\in [\ts^*(R,\EXa,P),1]$.  Thus for any  $(R,\EXa, P)$ triple,   $\max_{\XSID}\EXo (R,\EXa, \ts, P,\XSID)$, takes its maximum value either at the minimum possible value of $\ts$, i.e. $\ts^*(R,\EXa,P)=\ts^*(R,\EXa)$,  or at the maximum possible value of $\ts$, i.e. $1$. It is shown in  Appendix \ref{app:specialcase} $\max_{\XSID}\EXo (R,\EXa, \ts, P,\XSID)$ takes its  maximum value  at  $\ts=\ts^*(R,\EXa)$. 

Furthermore if the maximizing $P$ is not only the same for all $\ts \in [\ts^*(R,\EXa),1]$ for a given  $(R,\EXa)$ pair but also for all   $(R,\EXa)$  pairs such that $\EXa \leq \rexp(R)$ then we can  find the optimal $\EXo (R,\EXa)$ by simply maximizing over $\XSID$'s. In symmetric channels,  for example, uniform distribution is the optimal distribution for all $(R,\EXa)$ pairs. Thus
\begin{equation}
\label{eq:symach1}
\EXo(R,\EXa)=\left\{
\begin{array}{cl}
\EXo (R,\EXa, 1, P^*, \XSID)  
&\mbox{if~} \EXa> \rexp(R,P^*)  \\
\max_{\XSID}   \EXo (R,\EXa, \ts^*(R,\EXa), P^*, \XSID)
&\mbox{if~} \EXa\leq \rexp(R,P^*)  
\end{array}
\right\}
\end{equation}
where $P^*$ is the uniform distribution. 

\subsection{Alternative Expression for Exponent:}
The  minimization given in (\ref{eq:thm1}) for $\EXo(R,\EXa,\ts,P,\XSID)$ is over transition probability matrices and control phase output types. In order to get a better grasp of the resulting expression, we simplify the  analytical expression in this section. We do that by expressing the minimization in (\ref{eq:thm1}) in terms of solutions of lower dimensional optimization problems.

Let $\cesp(R,P,Q)$ be the minimum Kullback-Leibler divergence under $P$ with respect to $W$ among the transition probability matrices whose mutual information under $P$ is less than $R$ and whose output distribution under $P$ is $Q$. It is shown in Appendix \ref{app:convexityints} that for a given $P$, $\cesp(R,P,Q)$ is convex in $(R,Q)$ pair. Evidently for a given $(P,Q)$ pair $\cesp(R,P,Q)$ is a non-increasing in $R$. Thus for a given $(P,Q)$ pair $\cesp(R,P,Q)$ is strictly decreasing on a closed interval and is an extended real valued function  of the form:
\begin{subequations}
\label{eq:cesp} 
\begin{align}
\cesp(R,P,Q) &=
\left\{
\begin{array}{cl}
\infty  
&R<\cespal{P}{Q}\\
\min_{V:\substack{ \MI{P}{V}\leq R \\ \Ymar{P V}=Q}} \CDIV{V}{W}{P}
&R \in [\cespal{P}{Q},\cespah{P}{Q} ]\\
\min_{V:\Ymar{P V}=Q} \CDIV{V}{W}{P}
&R >\cespah{P}{Q}
\end{array}\right\}\\
\cespal{P}{Q}
&=\min\nolimits_{V:\substack{ P V\gg P W \\ \Ymar{P V}=Q}}  \MI{P}{V}\\
\cespah{P}{Q}
&=\min\nolimits_{R} \left\{R:\min\nolimits_{V:\substack{ \MI{P}{V}\leq R \\ \Ymar{P V}=Q}} \CDIV{V}{W}{P}
=\min\nolimits_{V:\substack{\Ymar{P V}=Q}} \CDIV{V}{W}{P} \right\}
\end{align}
\end{subequations}
where  $P V\gg PW$ iff for all $(\dinp,\dout)$ pairs such that $P(\dinp)W(\dout|\dinp)$ is zero, $P(\dinp)V(\dout|\dinp)$ is also zero.

Let $\bhexp{T}{\XSID}$ be the minimum Kullback-Leibler divergence with respect to $W_r$ under $\XSID$, among the $\XSOD$'s whose Kullback-Leibler divergence  with respect to $W_a$ under $\XSID$ is less than or equal to $T$. 
\begin{equation}
\label{eq:bhtexpd}
\bhexp{T}{\XSID}\triangleq \min_{\XSOD: \CDIV{\XSOD}{W_a}{\XSID}\leq T  } \CDIV{\XSOD}{W_r}{\XSID} 
 \end{equation}
For a given $\XSID$, $\bhexp{T}{\XSID}$ is  non-increasing and convex in $T$, thus $\bhexp{T}{\XSID}$ is  strictly decreasing in $T$ on a closed interval.  An equivalent expressions for  $\bhexp{T}{\XSID}$ and boundaries of this closed interval is derived in Appendix \ref{app:E2=F2},
\begin{equation}
\label{eq:bhtexp}
\bhexp{T}{\XSID}=
\left\{\begin{array}{cll}
\infty
& \mbox{if~} &T< \CDIV{\Xq{0}}{W_a}{\XSID}\\
\CDIV{\Xq{s}}{W_r}{\XSID}
& \mbox{if~} &T=\CDIV{\Xq{s}}{W_a}{\XSID} \quad \mbox{for~some~}  s \in [0,1]\\
\CDIV{\Xq{1}}{W_r}{\XSID}
& \mbox{if~} &T> \CDIV{\Xq{1}}{W_a}{\XSID} 
\end{array}\right\}
\end{equation} 
where
\begin{equation*}
\XQ{s}{y}{x_1,x_2}=
\left\{
\begin{array}{lcl}
\tfrac{\IND{W(y|x_2)>0}}{\sum_{\tilde{y}:W(\tilde{y}|x_2)>0}W(\tilde{y}|x_1)} W(y|x_1)
&\mbox{ if }&s=0\\
\tfrac{W(y|x_1)^{1-s} W(y|x_2)^s }{\sum_{\tilde{y}}W(\tilde{y}|x_1)^{1-s} W(\tilde{y}|x_2)^s}
&\mbox{ if }&s\in (0,1)\\
\tfrac{\IND{W(y|x_1)>0}}{\sum_{\tilde{y}:W(\tilde{y}|x_1)>0} W(\tilde{y}|x_2)} W(y|x_2)
&\mbox{ if} &s=1\\
\end{array}\right\}
\end{equation*}

For a $(R,\EXa,\ts,P,\XSID)$ such that $\EXa\leq \ts\rexp(\tfrac{R}{\ts},P)$, using the definition  of $\EXo(R,\EXa,\ts,P,\XSID)$ in (\ref{eq:thm1}) together with the equations (\ref{eq:prand}), (\ref{eq:cesp}) and (\ref{eq:bhtexp}) we get 
\begin{align*}
\!\EXo\!(R,\EXa,\ts,P,\XSID)
&=\displaystyle{\min_{\substack{Q,T,R_1,R_2:\\R_1\geq R_2\geq 0,~T\geq 0 \\ \ts \cesp(\frac{R_1}{\ts},P,Q)+ |R_2- R|^{+}+T \leq \EXa}}}
\ts\cesp(\tfrac{R_2}{\ts},P,Q)+ |R_1-R|^{+} +(1-\ts) \bhexp{\tfrac{T}{1-\ts}}{\XSID}
\end{align*} 
For any $(R,\EXa,\ts,P,\XSID)$  above minimum is also achieved at a $(Q,R_1,R_2,T)$ such that  $R_1\geq R_2\geq R$. In order to see this take any minimizing  $(Q^*,R_1^*,R_2^*,T^*)$, then there are three possibilities:
\begin{enumerate}[(a)]
\item $R_1^*\geq R_2^*\geq R$ claim holds trivially.
\item $R_1^*\geq R> R_2^*$, since  $\cesp(\tfrac{R_2}{\ts},P,Q)$ is non-increasing function $(Q^*,R_1^*,R,T^*)$, is also minimizing, thus claim holds.
\item $R>R_1^*> R_2^*$, since  $\cesp(\tfrac{R}{\ts},P,Q)$ is non-increasing function $(Q^*,R,R,T^*)$, is also minimizing, thus claim holds.
\end{enumerate}
Thus we obtain the following expression for $\EXo(R,\EXa,\ts,P,\XSID)$,
\begin{equation}
\label{eq:altoptdef}
\!\EXo\!(R,\EXa,\ts,P,\XSID)
\!=\!\left\{
\begin{array}{cl}
\ts \rexp(\tfrac{R}{\ts},P)
&\mbox{if~}\EXa\!>\!\ts\!\rexp(\tfrac{R}{\ts},P)\\
\hspace{-2.4cm}
\displaystyle{\min_{\substack{\hspace{0.2cm}Q,T,R_1,R_2:\\\hspace{0.9cm}R_1\geq R_2\geq R,~T\geq 0 \\\hspace{2.4cm} \ts \cesp(\frac{R_1}{\ts},P,Q)+ R_2- R+T \leq \EXa}}}
\hspace{-1.6cm}
\ts\cesp(\tfrac{R_2}{\ts},P,Q)+ R_1-R +(1-\ts) \bhexp{\tfrac{T}{1-\ts}}{\XSID}
&\mbox{if~}\EXa\!\leq\!\ts\!\rexp(\tfrac{R}{\ts},P)
\end{array}\right\}
\end{equation}
Equation (\ref{eq:altoptdef})  is simplified further for symmetric channels. For symmetric channels,
\begin{equation}
  \label{eq:symachs}
\spexp(R)= \cesp(R,P^*,Q^*)=\min_{Q} \cesp(R,P^*,Q)
\end{equation}
where $P^*$ is the uniform input distribution and $Q^*$ is the corresponding output distribution under $W$. 

Using alternative expression for $\EXo(R,\EXa,\ts,P,\XSID)$  given in  (\ref{eq:altoptdef}) together with equations (\ref{eq:symach1}) and (\ref{eq:symachs}) for symmetric channels we get,
\begin{equation}
\label{eq:symach2}
\EXo(R,\EXa)=\left\{
\begin{array}{cl}
\rexp(R)
&\mbox{if~} \EXa> \rexp(R)  \\
 \displaystyle{\max_{\XSID}
\hspace{-.5cm}
\min_{ 
\substack{R'',R',T:\\
R''\geq R'\geq R~T\geq 0\\
\ts^* \spexp(\tfrac{R''}{\ts^*})+ R'-R+T\leq \EXa}}}
\hspace{-.5cm}
\ts^* \spexp\left(\tfrac{R'}{\ts^*}\right)+ R''-R+(1-\ts^*)\bhexp{\tfrac{T}{1-\ts^*}}{\XSID}
&\mbox{if~} \EXa\leq \rexp(R)  
\end{array}
\right\}
\end{equation}
where $\ts^*(R,\EXa)$ is given in equation (\ref{eq:tsop}).

Although (\ref{eq:symachs}) does not hold in general  using definition of  $\cesp(R,P,Q)$ and  $\spexp(R,P)$ we can assert that
\begin{equation}
  \label{eq:sxtr}
 \cesp(R,P,Q)\geq\min_{\tilde{Q}} \cesp(R,P,\tilde{Q})=\spexp(R,P)
\end{equation}
Note that (\ref{eq:sxtr}) can be used to bound the minimized expression  in (\ref{eq:altoptdef}) from below. In addition recall that  if the set that a minimization is done over is enlarged resulting minimum can not increase. We can use (\ref{eq:altoptdef}) also to enlarge the set that minimization is done over in  (\ref{eq:sxtr}). Thus we get  an exponent $\EXor(R,\EXa,\ts,P,\XSID)$ which is smaller than or equal to $\EXo(R,\EXa,\ts,P,\XSID)$ in all channels and for all $\EXor(R,\EXa,\ts,P,\XSID)$'s:
\begin{equation}
\label{eq:def2}
\EXor(R,\EXa,\ts,P,\XSID)
\!=\!\left\{
\begin{array}{cl}
\ts \rexp(\tfrac{R}{\ts},P)  &\mbox{if~} \EXa>\ts \rexp(\tfrac{R}{\ts},P) \\
\hspace{-2.4cm}
\displaystyle{\min_{ 
\substack{R'',R',T:\\
\hspace{0.8cm}R''\geq R'\geq R~T\geq 0\\
\hspace{2.4cm}\ts \spexp(\tfrac{R''}{\ts},P)+ R'-R+T\leq \EXa}}}
\hspace{-2cm}
\ts \spexp\left(\tfrac{R'}{\ts},P\right)+ R''-R+(1-\ts)\bhexp{\tfrac{T}{1-\ts}}{\XSID}
&\mbox{if~} \EXa \leq \ts \rexp(\tfrac{R}{\ts},P)
\end{array}\right\}
\end{equation}
After an investigation very similar to the one we have already done for $\EXo (R,\EXa, \ts, P, \XSID)$ in Section \ref{sec:how}, we obtain the below expression for the optimal error exponent for reliable sequences emerging from (\ref{eq:def2}):
\begin{equation}
\label{eq:optexp3}
\EXor(R,\EXa)=
\left\{\begin{array}{c l l}
\rexp(R) 
&\forall R\geq 0 &  \forall \EXa> \rexp(R)\\
\displaystyle{\max_{\ts\in [\ts^*(R,\EXa) ,1]} \max_{P \in \PSet{R}{\EXa}{\ts}} \max_{\XSID}   \EXor (R,\EXa, \ts, P, \XSID) }
&\forall R\geq 0 &  \forall \EXa\leq \rexp(R)
\end{array}
\right\}
\end{equation}
where $\ts^*(R,\EXa)$, $\PSet{R}{\EXa}{\ts}$ and $\EXor (R,\EXa, \ts, P, \XSID)$ are given in equations (\ref{eq:tsop}), (\ref{eq:psetdef}) and (\ref{eq:def2}), respectively.

\subsection{Special Cases}
\subsubsection{Zero Erasure Exponent Case, $\TEX(R,0)$}
 Using a simple repetition-at-erasures scheme, fixed length errors-and-erasures codes, can be converted into variable length block codes, with the same error exponent. Thus the error exponents of variable length block codes given by Burnashev  in \cite{bur}  is an upper bound to the error exponent of fixed length block codes with erasures:
 \begin{equation*}
  \TEX(R,\EXa) \leq  \left( 1-\tfrac{R}{\CX} \right) \DX \qquad \forall R\geq 0, \EXa\geq 0
 \end{equation*}
where $\DX= \max_{\dinp,\tilde{\dinp}} \sum_{\dout} \CT{\dinp}{\dout} \log \frac{\CT{\dinp}{\dout}}{\CT{\tilde{\dinp}}{\dout}}$.

We show below that, $\EXor (R,0)\geq ( 1-\tfrac{R}{\CX}) \DX$. This implies that our coding scheme is optimal for $\EXa=0$  for all rates i.e. $\EXor (R,0)=\TEX(R,0)=( 1-\tfrac{R}{\CX})\DX$.

 Recall that  for all $R$ less than capacity  $\ts^*(R,0)=\tfrac{R}{\CX}$. Furthermore for any $\ts\geq \tfrac{R}{\CX}$ 
\begin{equation*}
\PSet{R}{0}{\ts}=\{ P :   \MI{P}{W} \geq \tfrac{R}{\ts} \} 
\end{equation*}
Thus for any $(R,0,\ts,P)$ such that $P\in\PSet{R}{0}{\ts}$,  $R''\geq R'\geq R$, $T\geq 0$ and $\ts \spexp(\tfrac{R''}{\ts},P)+ R'-R+T\leq 0 $, imply that $R'=R$, $R''=\ts \MI{P}{W}$, $T=0$. Consequently
\begin{equation}
 \EXor (R,0,\ts,P,\XSID) = \ts \left[\spexp\left(\tfrac{R}{\ts},P \right)+\MI{P}{W}-\tfrac{R}{\ts} \right] +(1-\ts)   \CDIV{W_r}{W_a}{\XSID}
\end{equation}
When we maximize over $\XSID$ and $ P \in  \PSet{R}{0}{\ts}$ we get:
\begin{equation}
\EXor (R,0,\ts)= \max_{P\in \PSet{R}{0}{\ts}}  \ts \spexp \left(\tfrac{R}{\ts},P \right)+\ts \MI{P}{W}-R +( 1-\ts)\DX
\qquad \forall\ts \in [\tfrac{R}{\CX},1].
\end{equation}
Simply inserting the minimum possible value of $\ts$ i.e. $\ts^*(R,0)=\tfrac{R}{\CX}$:
\begin{align*} 
\EXor (R,0,\tfrac{R}{\CX})
&= \max_{P\in \PSet{R}{0}{\frac{R}{\CX}}}  \tfrac{R}{\CX} \spexp \left(\CX,P \right)+\tfrac{R}{\CX} \MI{P}{W}-R +( 1-\tfrac{R}{\CX})\DX\\
&=( 1-\tfrac{R}{\CX})\DX.
\end{align*}
Thus $\EXor (R,0)\geq ( 1-\tfrac{R}{\CX})\DX$.

Indeed one need not to rely on the converse on  variable length block codes in order  to establish the fact that $\EXor (R,0)=( 1-\frac{R}{\CX}) \DX$. The lower bound to probability of error presented in the next section, not only recovers this particular  optimality result but also upper bounds the optimal  error exponent, $\TEX(R,\EXa)$, as a function of rate $R$ and erasure exponents $\EXa$.

\subsubsection{Channels with non-zero Zero Error Capacity}
 For channels with a non-zero zero-error capacity, as a result of  equation (\ref{eq:thm1})   $\EXo(R,\EXa)=\infty$    for any $\EXa<\rexp(R)$. This implies that we can get error-free block codes with this two phase coding scheme for any rate $R<\CX$ and any erasure exponent $\EXa\leq \rexp(R)$. As we  discuss in Section \ref{sec:efree}  in more detail, this is the best erasure exponent for rates over the critical rate, at least for symmetric channels.

\section{An Outer Bound for Error Exponent Erasure Exponent Trade Off}\label{sec:converse}
  In this section we derive an upper bound on $\TEX(R,\EXa)$ using previously known results on erasure free block codes with feedback and a generalization of the straight line bound of Shannon, Gallager and Berlekamp \cite{SGB}. We first present a lower bound on the minimum error probability of block codes with feedback and erasures, in terms of that of shorter codes in Section \ref{sec:converse1}. Then in  Section \ref{sec:converse2} we make a brief overview of the outer bounds on the error exponents of erasure free block codes with feedback. Finally  in Section \ref{sec:converse3}, we use the relation we have derived in Section \ref{sec:converse1} to tie the previously known results we have summarized in Section \ref{sec:converse2}  to bound $\TEX(R,\EXa)$.

\subsection{A Trait of Minimum Error Probability of block codes with Erasures}\label{sec:converse1}
Shannon, Gallager and Berlekamp in \cite{SGB} considered fixed length block codes, with list decoding and established a family of lower bounds on the minimum error probability in terms of the product of minimum error probabilities of certain shorter codes. They have shown,  \cite[Theorem 1]{SGB}, that  for fixed length block codes  with list decoding and without feedback
 \begin{equation}
\label{eq:sgb}
\tilde{\PE} (M, \blx, L) \geq  \tilde{\PE} (M, \blx_1, L_1)  \tilde{\PE} (L_1+1, \blx-\blx_1, L)
 \end{equation}
 where  $\tilde{\PE} (M, \blx, L)$ denotes the minimum error probability of erasure free block  codes of length $\blx$  with $M$ equally probable messages and with decoding list size  $L$. As they have already pointed out in \cite{SGB} this theorem continues to hold in the case when a feedback link is available from receiver to the transmitter; although $\tilde{\PE}$'s are different when  feedback is available, the relation given in equation (\ref{eq:sgb}) still holds. They were interested in erasure free codes. We, on the other hand, are interested in  block codes which might have  non-zero erasure probability.  Accordingly we need to incorporate erasure probability as one of the parameters of the optimal error probability. This is what this section is dedicated to.

In a  size $L$ list decoder with  erasures, decoded set $\est$  is either a subset\footnote{Note that if $\est \subset \mesS$ then $\Era\notin \est$ because  $\Era \notin \mesS$.} of $\mesS$ whose size is at most $L$, like the erasure-free case,  or a set which only includes the erasure symbol, i.e. either $\est \subset \mesS$ such that $|\est|\leq L$ or $\est=\{\Era\}$. An erasure  occurs whenever $\est=\{\Era\}$ and an error  occurs whenever $\est \neq \{\Era\}$ and $\mes \notin \est$. We will denote the minimum error probability of  length $\blx$ block codes, with  $M$ equally probable messages, decoding list size $L$ and erasure probability $\Per$ by $\PE(M,\blx, L, \Per)$.

 Theorem \ref{thm:con1} below bounds the error probability of block  codes with erasures and list decoding using the error probabilities of shorter codes with erasures and list decoding, like \cite[Theorem 1]{SGB} does in the erasure free case. Like its counter part in erasure free case Theorem \ref{thm:con1} is later used to establish outer bounds to error exponents.
\begin{theorem}
\label{thm:con1}
      For any $\blx$, $M$, $L$, $\Per$, $\blx_1 \leq \blx$, $L_1$, and  $0\leq s \leq 1$ the  minimum error probability of fixed length block codes with feedback satisfy
      \begin{equation}  
\label{eq:thmcon1}
(1-s) \PE  (M, \blx, L, \Per)  \geq \PE (M, \blx_1, L_1, s)  \PE \left( L_1+1, \blx-\blx_1, L, \tfrac{(1-s)\Per}{\PE (M, \blx_1, L_1, s)}  \right) 
 \end{equation}  
\end{theorem}
Note that given a $(M,\blx, L)$ triple if the error probability erasure probability pairs  $({\Pe}_{1},{\Per}_{1})$  and $({\Pe}_{2},{\Per}_{2})$ are achievable,  then for any $\gamma \in [0,1]$  using the initial symbol  $\rsy_0$ of the feedback link we can construct a code that uses   the code achieving  $({\Pe}_{1},{\Per}_{1})$ with probability $\gamma$,  the code achieving   $({\Pe}_{2},{\Per}_{2})$  with probability $(1-\gamma$). This new code achieves  error probability erasure probability pair  $(\gamma{\Pe}_{a}+(1-\gamma){\Pe}_{b}, \gamma{\Per}_{a}+(1-\gamma) {\Per}_{b})$. As a result for any $(M,\blx, L)$ triple the  set of achievable error probability erasure probability pairs is convex. We  use this fact twice in order to prove Theorem \ref{thm:con1}.

Let us first consider the following lemma which bounds the achievable error probability erasure probability,  pairs for block codes with nonuniform a priori probability distribution, in terms of block  codes with a uniform  a priori probability distribution  but fewer  messages.
\begin{lemma}
\label{lem:ero-era}
For any length $\blx$  block code with message set $\mesS$,  a priori probability distribution $\apri(\cdot)$ on $\mesS$, erasure probability $\Per$,  decoding list size $L$,  and integer $K$
  \begin{align}
\label{eq:ero-era}
\Pe &\geq \XLD{\apri}{K} \PE\left(K+1, \blx, L, \tfrac{\Per}{\XLD{\apri}{K}} \right) &&\mbox{where}&&
\XLD{\apri}{K}= \min_{ \substack{{\cal S} \subset \mesS, |{\cal S}|\leq  K}} \apri({\cal S}^{{\bf c}})
,~~~{\cal S}^{{\bf c}}=\mesS/{\cal S}.
\end{align}
Recall that $\PE\left(K+1, \blx, L, \Per \right)$ is the minimum error probability of length $\blx$ codes with $(K+1)$ equally probable messages and decoding list size $L$, with feedback if the original code does have feedback and without feedback if the original code does not.
\end{lemma}
Note that  $\XLD{\apri}{K}$ is the error probability of a decoder which decodes to the set of $K$ most likely messages under $\apri$. In other words $\XLD{\apri}{K}$ is the minimum error probability for a size $K$ list decoder when the posterior probability distribution is $\apri$.

\begin{proof} 
If  $\XLD{\apri}{K}=0$, the lemma holds trivially. Thus we assume $\XLD{\apri}{K}>0$ henceforth. For any size $(K+1)$ subset $\mesS'$ of $\mesS$, one can use the encoding scheme  and the decoding rule of the original code for $\mesS$, to construct the following block code for $\mesS'$:
\begin{itemize}
\item {\bf Encoder:}$\forall \dmes \in \mesS'$ use the encoding scheme for message $\dmes$ in the original code, i.e.
\begin{equation*}
  \ENC'_t(\dmes,\dfsy^{t-1})=\ENC_t(\dmes,\dfsy^{t-1})  \qquad \forall \dmes \in \mesS',~~ t \in [1,\blx],~~\dfsy^{t-1}\in {\fsyS}^{t-1}
\end{equation*}
\item {\bf Decoder:}   For all $\dfsy^{\blx}\in {\fsyS}^{\blx}$ if the original decoding rule declares erasure, declare erasure, else the decode to the  intersection of the original decoded list and $\mesS'$.
\begin{equation*}
\est'=
\begin{cases}
\Era                                   &\mbox{if~}~\est=\Era\\
\est\cap \mesS'                &\mbox{else}
\end{cases}
\end{equation*} 
\end{itemize}
This is a length $\blx$ code with $(K+1)$ messages and  decoding list size  $L$. Furthermore for all $\dmes$ in $\mesS'$  the conditional error probability $\Pem{\dmes}'$ and  the conditional erasure probability $\Perm{\dmes}'$ are equal to the conditional error probability $\Pem{\dmes}$ and the conditional  erasure probability $\Perm{\dmes}$ in the original code, respectively. 

Note that  
\begin{equation}
\label{eq:peperreg}
\tfrac{1}{K+1}\sum\nolimits_{\dmes \in \mesS' } \left(\Pem{\dmes}, \Perm{\dmes} \right) \in \ACH(K+1, \blx, L) \qquad \forall \mesS' \subset \mesS \mbox{ such that~} |\mesS'|=K+1
\end{equation}
where $\ACH(K+1, \blx, L)$ is the set of achievable error probability, erasure probability pairs for length $\blx$  block codes with $(K+1)$ equally probable messages and with decoding list size $L$. 

Let the smallest non-zero element of $\{\apri(1), \apri(2), \ldots \apri(|\mesS|) \}$ be $\apri(\xi_1)$. For any size $(K+1)$ subset of $\mesS$ which includes $\xi_1$ and all  whose elements have non-zero probabilities, say $\mesS_1$, we have,
\begin{align*}
 \left( \Pe, \Per  \right) 
&= \sum\nolimits_{\dmes \in  \mesS} \apri(\dmes) (\Pem{\dmes}, \Perm{\dmes})\\
&= \sum\nolimits_{\dmes \in  \mesS} [\apri(\dmes)  -\apri(\xi_1)  \IND{\dmes \in \mesS_1}] 
(\Pem{\dmes}, \Perm{\dmes})+ \apri(\xi_1) \sum\nolimits_{\dmes \in \mesS_1}(\Pem{\dmes},\Perm{\dmes})
\end{align*}
Equation (\ref{eq:peperreg}) and the definition of $\ACH(K+1, \blx, L)$, implies that  $\exists\ach_{1} \in  \ACH(K+1, \blx, L)$ such that
\begin{align}
 \left( \Pe, \Per \right)
&= \sum\nolimits_{\dmes \in  \mesS }  \apri^{(1)}(\dmes) (\Pem{\dmes}, \Perm{\dmes})+ \apri(\ach_{1}) \ach_{1}\\
1
&=\apri(\ach_{1})+\sum\nolimits_{\dmes \in \mesS} \apri^{(1)}(\dmes)
\end{align}
where $\apri(\ach_{1})=(K+1) \apri(\xi_1)$ and $\apri^{(1)}(\dmes)= \apri(\dmes)-\apri(\xi_1)  \IND{\dmes \in \mesS_1}$.  Furthermore the number of non-zero  $\apri^{(1)}(\dmes)$'s is at least one less than that of non-zero $\apri(\dmes)$'s. The remaining probabilities, $\apri^{(1)}(\dmes)$, have a minimum, $\apri^{(1)}(\xi_2)$ among its non-zero elements. One can repeat the same argument once more using that element and reduce the number of non-zero elements at least one more.  After at most $|\mesS|-K$  such iterations  one reaches to a $\apri^{(\ell)}$ which is  non-zero for $K$ or fewer messages:
\begin{align}
\label{eq:yicon}
 \left( \Pe, \Per \right) = \sum\nolimits_{j=1}^{\ell} \apri (\ach_{j}) \ach_{j}+ \sum\nolimits_{\dmes \in \mesS}  \apri^{(\ell)}(\dmes)  (\Pem{\dmes},\Perm{\dmes}) 
\end{align}
where $\apri^{(\ell)}(\dmes) \leq \apri(\dmes)$ for all $\dmes$ in $\mesS$ and $\sum_{ \dmes \in \mesS} \IND{\apri^{(\ell)}(\dmes)>0} \leq K$.

In equation (\ref{eq:yicon}), the first sum  is equal to a convex combination of $\ach_{j}$'s multiplied by $ \sum_{j=1}^{\ell} \apri(\ach_{j})$;  the second sum is equal to a pair with non-negative entries. 
As a result of definition of $\XLD{\apri}{K}$ given in equation (\ref{eq:ero-era}),
\begin{equation}
  \label{eq:ero-erab}
 \XLD{\apri}{K}\leq \sum\nolimits_{j=1}^{\ell} \apri (\ach_{j}).
\end{equation}
Then as a result of convexity of $\ACH(K+1, \blx, L)$ we can conclude that there exists a $\ach  \in \ACH(K+1, \blx , L)$  such that $( \Pe,\Per) = a \XLD{\apri}{K} \tilde{\ach}    +(b_1,b_2)$ for some $a\geq 1$, $b_1\geq0$ and $b_2\geq 0$. Thus 
\begin{equation}
\label{eq:ero-eraf}
\exists \ach  \in \ACH(K+1, \blx , L) \mbox{~such that~} ( \tfrac{\Pe}{\XLD{\apri}{K}},\tfrac{\Per}{\XLD{\apri}{K}}) =    \ach+ (b_3,b_4) 
\mbox{~for some~}  b_3\geq0, b_4\geq 0.
\end{equation}
Then the lemma follows   from equation (\ref{eq:ero-eraf}),
the fact that $\PE(M, \blx, L, \Per)$ is decreasing in $\Per$
 and the fact that $\PE(M, \blx, L, \Per)$
 is  uniquely determined by   $\ACH(M, \blx, L)$  for $s_{\Era}\in [0,1]$ as follows
\begin{equation}
\label{eq:pedefz}
\PE\left(M, \blx, L, \Per  \right)=\min_{\ach_{\Era}:(\ach_{\Ero},\ach_{\Era}) \in \ACH(M, \blx, L)} \ach_{\Ero}   \qquad \forall (M, \blx, L,\ach_{\Era}).
\end{equation}
\end{proof}

For proving Theorem \ref{thm:con1}, we express the error and erasure probabilities, as a convex combination of error and erasure probabilities of  $(\blx-\blx_1)$ long block codes with a priori probability distribution $\apri_{\ZA}(\dmes)=\PCX{\dmes}{\ZA}$ over the messages and apply Lemma \ref{lem:ero-era} together with convexity arguments  similar to the ones above.

\begin{proofs}{Theorem \ref{thm:con1}}
For all $\dmes$ in $\mesS$, let $\DECS (\dmes)$  be the decoding region of $\dmes$, $\DECS(\Era)$ be the decoding region of the erasure symbol $\Era$ and  $\tilde{\DECS}(\dmes)$ the error region of $\dmes$:
\begin{align}
\DECS(\dmes)
&\DEF\{\dfsy^{\blx} : \dmes \in \est\}     &
\DECS(\Era) 
&\DEF\{\dfsy^{\blx} : \Era \in \est\}         &
\tilde{\DECS}(\dmes) \DEF
&\DECS(\dmes)^{c} \cap \DECS(\Era)^{c} &
&\mbox{where~} \DECS^{c}=\fsyS^{\blx}/\DECS  .
\end{align}
Then  for all $\dmes \in \mesS$,
\begin{equation}
  ( \Pem{\dmes}, \Perm{\dmes}) = \left(   \PCX{\tilde{\DECS}(\dmes)}{\dmes} , \PCX{\DECS(\Era)}{\dmes} \right). 
\end{equation}
Note that\footnote{There is a slight abuse of notation here, if $\rsy$'s include  real valued random variables with densities, we should integrate, rather than sum, over them. Since it is clear from the context what needs to be done we omit that subtlety in below calculations.}
\begin{align*}
\Perm{\dmes}
&= \sum\nolimits_{\dfsy^{\blx}:\dfsy^{\blx} \in \DECS(\Era)} \PCX{\dfsy^{\blx}}{\dmes} \\
&=  \sum\nolimits_{\ZA} \PCX{\ZA}{\dmes}   \sum\nolimits_{\ZB: (\ZA,\ZB) \in \DECS(\Era)} \PCX{\ZB}{\dmes,\ZA}.
\end{align*}
 Then the erasure probability is 
\begin{align*}
\Per
&= \sum\nolimits_{\dmes \in \mesS} \tfrac{1}{|\mesS|}  \sum\nolimits_{\ZA} \PCX{\ZA}{\dmes}   \sum\nolimits_{\ZB: (\ZA,\ZB) \in \DECS(\Era)} \PCX{\ZB}{\dmes,\ZA}\\
&=  \sum\nolimits_{\ZA} \PX{\ZA} \left(\sum\nolimits_{\dmes \in \mesS}  \PCX{\dmes}{\ZA}    \sum\nolimits_{\ZB: (\ZA,\ZB) \in \DECS(\Era)} \PCX{\ZB}{\dmes,\ZA}\right)\\
&=  \sum\nolimits_{\ZA} \PX{\ZA} \Per(\ZA).
\end{align*}
Note that for every  $\ZA$, $\Per(\ZA)$ is the erasure  probability of a code of length $(\blx-\blx_1)$  with a priori  probability distribution $\apri_{\ZA}(\dmes)=\PCX{\dmes}{\ZA}$. Furthermore one can write the error probability, $\Pe$ as 
\begin{align*}
\Pe
&=\sum\nolimits_{\ZA} \PX{\ZA} \left(\sum_{\dmes \in \mesS}  \PCX{\dmes}{\ZA}    \sum\nolimits_{\ZB: (\ZA,\ZB) \in \tilde{\DECS}(\dmes)} \PCX{\ZB}{\dmes,\ZA}\right)\\
&=  \sum\nolimits_{\ZA} \PX{\ZA} \Pe(\ZA)
\end{align*}
where $\Pe(\ZA)$ is the error probability of the very same  length $(\blx-\blx_1)$ code.  As a result of Lemma \ref{lem:ero-era}, the pair $(\Pe(\ZA), \Per(\ZA))$ satisfies  
\begin{equation}
\Pe(\ZA)\geq  \XLD{\apri_{\ZA}}{L_1} \PE\left(L_1+1, (\blx-\blx_1), L, \tfrac{\Per(\ZA)}{ \XLD{\apri_{\ZA}}{L_1} } \right). 
\end{equation}
Then for any  $s \in [0,1]$.
\begin{align}
(1-s)\Pe
\notag
&\!\geq\!  \sum\nolimits_{\ZA} \PX{\ZA} (1-s) \XLD{\apri_{\ZA}}{L_1} \PE\left(L_1+1, (\blx-\blx_1), L, \tfrac{\Per(\ZA)}{ \XLD{\apri_{\ZA}}{L_1} } \right)\\
\notag
&\!\geq\!  \left(\sum\nolimits_{\ZA} \PX{\ZA} (1-s) \XLD{\apri_{\ZA}}{L_1} \right) \PE\left(L_1+1, (\blx-\blx_1), L, \tfrac{\sum_{\ZA} \PX{\ZA} (1-s) \Per(\ZA)}{\sum_{\ZA} \PX{\ZA} (1-s) \XLD{\apri_{\ZA}}{L_1}} \right)\\
\label{eq:feb1}
&\!=\!  \left(\sum\nolimits_{\ZA}\!\!\! \PX{\ZA}
\! (1\!-\!s) \XLD{\apri_{\ZA}}{L_1} \right) \PE\!\left(\!L_1\!+\!1, (\!\blx\!-\!\blx_1\!), L, \tfrac{(1\!-\!s)\Per}{\sum_{\ZA} \!\PX{\ZA}  (1\!-\!s) \XLD{\apri_{\ZA}}{L_1}}\! \right)
\end{align}
where the second inequality follows from the  convexity of $\PE(M, \blx, L, \Per)$ in $\Per$. Note that  $\PE(M, \blx, L, \Per)$ is convex in $\Per$ because of the equation (\ref{eq:pedefz}) and the convexity of the region   $\ACH(M, \blx, L)$.

Now consider a code which uses the  first $\blx_1$ time units of the original encoding scheme as its encoding scheme. Decoder of this new code draws a real number from $[0,1]$ uniformly at random, independently of $\fsy^{\blx_1}$ of the original code  (and the message evidently). If this number is less than $s$ it declares erasure else it makes a maximum likelihood decoding with list of size $L_1$. Then the sum on the left hand side of the below expression (\ref{eq:feb2}) is its error probability. But that probability is lower bounded by $\PE \left(M, \blx_1, L_1, s \right)$ which is minimum error probability over all length $\blx_1$ block codes with $M$ messages and decoding  list size $L_1$, i.e.
\begin{equation}
\label{eq:feb2}
  \sum\nolimits_{\ZA} \PX{\ZA} (1-s) \XLD{\apri_{\ZA}}{L_1}\geq  \PE \left(M, \blx_1, L_1, s \right).
\vspace{-.1cm}
\end{equation}
Then the theorem   follows from the fact that $\Pe(M,\blx,L_1,\Per)$ is decreasing function of $\Per$ and 
 the equations (\ref{eq:feb1}) and (\ref{eq:feb2}).
\vspace{-.3cm}
\end{proofs}
Like the result of Shannon, Gallager and Berlekamp in \cite[Theorem 1]{SGB}, Theorem \ref{thm:con1} is correct both with and without feedback. Although $\PE$'s are different in each case, the relationship between them given in equation (\ref{eq:thmcon1}) holds in both cases.

\subsection{Classical  Results on Error Exponent of Erasure-free  Block Codes with Feedback:}\label{sec:converse2}
 In this section we give a very brief overview of the previously known results on the error probability of erasure free block codes with feedback. These result are  used in Section \ref{sec:converse3} together with Theorem \ref{thm:con1} to bound $\TEX(R,\EXa)$ from above. Note that Theorem \ref{thm:con1} only relates the error probability  of longer codes to that of the shorter ones. It does not in and of itself bound the error probability. It is in a sense a tool to glue together various bounds on the error probability.

First bound we consider is on the error exponent of erasure free block codes with feedback. Haroutunian proved in \cite{har} that, for any $(M_{\blx}, \blx, L_{\blx})$ sequence of triples, such that $\lim_{\blx \rightarrow \infty} \tfrac{\ln M_{\blx}-\ln L_{\blx}}{\blx}=R$,
\begin{equation}
\label{eq:har1}
  \lim_{\blx \rightarrow \infty} \tfrac{-\ln \PE  (M_{\blx}, \blx, L_{\blx}, 0)}{\blx} \leq  \hexp (R)
\end{equation}
where 
\begin{equation}
\label{eq:har2}
 \hexp (R)=\min_{V: \CX(V)\leq R} \max_{P} \CDIV{V}{W}{P} 
\qquad \mbox{and} \qquad   \CX(V)=\max_{P} \MI{P}{V}.
\end{equation}

Second bound we consider is on the trade off between the error exponents of two messages in a two message erasure free block code with feedback. Berlekamp mentions this result in passing in \cite{ber} and attributes it to  Gallager and  Shannon.
\begin{lemma}
\label{lem:E2=F2}
  For any feedback encoding scheme with two messages and erasure free decision rule and  for all  $T \geq T_0$: 
\begin{equation}
\label{eq:ber1}
 \mbox{either} \qquad \Pe_1\geq  \tfrac{1}{4}  e^{-\blx T+ \sqrt{\blx} 4  \ln P_{min}}
\qquad \mbox{or} \qquad 
 \Pe_2> \tfrac{1}{4} e^{-\blx \obhexp{T}+ \sqrt{\blx} 4  \ln P_{min}}
\end{equation}
where $P_{min}=\min_{\dinp,\dout:\CT{\dinp}{\dout}>0} \CT{\dinp}{\dout}$.
\begin{align}
 \label{eq:tzero}
 T_0
&\triangleq  \max\nolimits_{\dinp,\tilde{\dinp}}  -\ln \sum\nolimits_{\dout:\CT{\tilde{\dinp}}{\dout}>0}  \CT{\dinp}{\dout} \\
\label{eq:ber2}
 \obhexp{T} 
&\triangleq  \max\nolimits_{\XSID} \bhexp{T}{\XSID}.  
\end{align}
\end{lemma}
 Result is old and somewhat intuitive to those who are familiar with the calculations in the non-feedback case. Thus  probably it has been proven  a number of times. But we are not aware of a published proof, hence we have included one in Appendix \ref{app:E2=F2}.

Although Lemma  \ref{lem:E2=F2}  establishes only the converse part $(T,\obhexp{T})$ is indeed the optimal trade off for the error exponents of two messages in an erasure free block code, both with and without feedback.  Achievablity of this trade off has already been  established in  \cite[Theorem 5]{SGB} for the case without feedback; evidently this implies the achievablity with feedback. Furthermore  $T_0$ does have  an operational meaning, it is the maximum error exponent first message can have, while the second message has zero error probability. This fact is also proved in Appendix \ref{app:E2=F2}.

For some channels Lemma \ref{lem:E2=F2} gives us a bound on the error exponent of erasure free-codes at zero rate,  which is tighter than Haroutunian's bound at zero rate.  In order to see this  let us first define  $T^{*}$ to be 
\begin{equation}
\label{eq:defTS}
T^{*}=\max_{T} \min\{T,\obhexp{T}\}.
\end{equation}
Note that  $T^*$ is finite iff $\sum_{\dout}\CT{\dinp}{\dout}\CT{\tilde{\dinp}}{\dout}>0$ for all $\dinp$, $\tilde{\dinp}$ pairs. Recall that this is also the necessary and sufficient condition of zero-error capacity, $\CX_0$, to be zero.  $\hexp(R)$ on the other hand is infinite for all $R\leq R_{\infty}$ like $\spexp(R)$ where $R_{\infty}$ is given by,
\begin{equation}
\label{eq:Rinf}
  R_{\infty}=-\min\nolimits_{P(\cdot)} \max\nolimits_{\dout} \ln \sum\nolimits_{\dinp :\CT{\dinp}{\dout}>0} P(\dinp)
\end{equation}
Even in the cases where $\hexp(0)$ is finite, $\hexp (0)\geq T^*$. We can use this fact, Lemma \ref{lem:E2=F2}, and Theorem \ref{thm:con1}, or \cite[Theorem 1]{SGB} for that matter, to strengthen  Haroutunian bound at low rates, as follows.
\begin{lemma}
\label{lem:harimp}
For all channels with zero zero-error capacity, $\CX_{0}=0$ and any sequence of $M_{\blx}$, such that $\lim_{\blx \rightarrow \infty} \tfrac{\ln M_{\blx}}{\blx}=R$,
\begin{equation}
\label{eq:har1n}
  \lim_{\blx \rightarrow \infty} \tfrac{-\ln \PE  (M_{\blx}, \blx, 1, 0)}{\blx} \leq  \hexpn (R)
\end{equation}
where 
\begin{equation*}
 \hexpn (R)=\left\{ 
   \begin{array}{l c l}
\hexp(R) & \mbox{if} & R\geq R_{ht}\\     
T^*+\tfrac{\hexp(R_{ht})-T^*}{R_{ht}} R
 & \mbox{if} & R\in [0, R_{ht})     
   \end{array}\right\}
\end{equation*}
and $R_{ht}$ is the unique solution of the equation $T^*=\hexp(R)-R\hexp'(R)$ if it exists, $R_{ht}=\CX$ otherwise.  
\end{lemma}
Before going into the proof let us note that $\hexpn(R)$ is obtained simply by drawing the tangent line to the curve $(R,\hexp(R))$ from the point $(0,T^{*})$. The curve $(R,\hexpn(R))$ is same as the tangent line, for the rates between $0$ and $R_{ht}$, and it is same as the curve $(R,\hexp(R))$ from then on where $R_{ht}$ is the rate of the point at which the tangent from $(0, T^*)$ meets the curve $(R,\hexp(R))$.
 \begin{proof}
For $R\geq R_{ht}$ this Lemma immediately follows from Haroutunian's result in \cite{har} for $L_1=1$. If $R< R_{ht}$ then we apply Theorem \ref{thm:con1}.
\begin{equation}
\label{eq:e2f2a}
(1-s)  \PE  (M, \blx, L_1, \Per)  \geq \PE  (M, \tilde{\blx}, L_1, s)  \PE \left(L_1+1, \blx-\tilde{\blx}, \tilde{L}, \tfrac{(1-s) \Per}{\PE  (M, \blx, L_1, s)} \right)
\end{equation}
 with\footnote{Or \cite[Theorem 1]{SGB} with  $L_1=1$ and $\blx_1=\lfloor \tfrac{R}{R_{ht}} \rfloor$.} $s=0$, $\Per=0$, $L_1=1$ and $\tilde{\blx}=\lfloor \tfrac{R}{R_{ht}} \rfloor$.   Furthermore, by Lemma \ref{lem:E2=F2}  and the definition of $T^{*}$ given in (\ref{eq:defTS}) we have,
\begin{equation}
\label{eq:e2f2n}
  \PE  (2, \blx-\tilde{\blx}, L, 0)\geq \tfrac{e^{-(\blx-\tilde{\blx}) T^*+ \sqrt{\blx-\tilde{\blx}} \ln  P_{min}}}{8}
\end{equation}
Using equations (\ref{eq:e2f2a}) and  (\ref{eq:e2f2n}) we get,
 \begin{equation*}
   \tfrac{-\ln \PE  (M, \blx, 1, 0)}{\blx} \leq  \tfrac{-\ln \PE  (M, \tilde{\blx}, 1, 0)}{\tilde{\blx}} \tfrac{R}{R_{ht}} 
+\left[1-\tfrac{R}{R_{ht}}+\tfrac{1}{\tilde{\blx}}\right]T^*+  \left( \sqrt{\tfrac{1}{\tilde{\blx}}} \right)\left( \sqrt{\tfrac{R_{ht}-R}{R_{ht}}} \right)\ln  \tfrac{P_{min}}{8}
 \end{equation*}
where $\tfrac{\ln M_{\blx}}{\tilde{\blx}}=R_{ht}$. Lemma follows by simply applying Haroutunian's result to the first terms on the right hand side.
 \end{proof}

\subsection{Generalized Straight Line Bound for Error-Erasure Exponents}\label{sec:converse3}
Theorem \ref{thm:con1} bounds the minimum error probability length $\blx$ block codes  from below in terms of the minimum error probability of length $\blx_1$  and  length $(\blx-\blx_1)$ block codes. The rate and erasure probability of the longer code constraints the rates and erasure probabilities of the shorter ones, but  does not specify them completely.  We use this fact together with the improved  Haroutunian's bound on the error exponents of erasure free  block codes with feedback, i.e. Lemma \ref{lem:harimp}, and the error exponent trade off of the erasure free feedback block codes with two messages, i.e. Lemma \ref{lem:E2=F2}, to obtain a family of  upper bounds on the error exponents of feedback block codes with erasure.
\begin{theorem}
\label{thm:con2}
For any DMC with $\CX_0=0$ rate $R\in [0,\CX]$ and $\EXa \in [0,\hexp(R)]$  and for any $r \in  [r_h(R,\EXa),\CX]$
\begin{equation*}
\TEX \left(R,\EXa \right) \leq \tfrac{R}{r} \hexpn(r) + (1-\tfrac{R}{r}) \obhexp{\tfrac{\EXa- \frac{R}{r} \hexpn (r)}{1-\frac{R}{r}}}
\end{equation*}
where $r_h(R,\EXa),$ is the unique solution of $R\hexpn(r)-r\EXa=0$.
\end{theorem}
Theorem \ref{thm:con2} simply states that any line connecting any two points of the curves $(R,\EXa,\EXo)=(R,\hexpn(R),\hexpn(R))$  and  $(R,\EXa,\EXo)=(0,\EXa,  \obhexp{\EXa})$ lies above the surface $(R,\EXa,\EXo)=(R,\EXa,\TEX(R,\EXa))$. The condition $\CX_0=0$ is not merely a technical condition due to the proof technique; as we will see in Section \ref{sec:efree} for channels with $\CX_0>0$,  there are zero-error codes with erasure exponent  as high as $\spexp(R)$ for any rate $R\leq \CX$.

\begin{proof}
We will consider the cases $r \in  (r_h(R,\EXa),\CX]$ and $r =r_h(R,\EXa)$ separately.
\begin{itemize}
\item {$r \in  (r_h(R,\EXa),\CX]$:}  Apply Theorem \ref{thm:con1} with $s=0$, $L=1$, $L_1=1$, take the logarithm of both sides of equation (\ref{eq:thmcon1}) and  divide by $\blx$, 
  \begin{equation}
\label{eq:newconx1}
\tfrac{-\ln\PE  (M, \blx, 1, \Per)}{\blx} \leq \left(\tfrac{\blx_1}{\blx}\right) \tfrac{-\ln \PE (M, \blx_1, 1, 0)}{\blx_1}+\left(1-\tfrac{\blx_1}{\blx}\right)\tfrac{- \ln \PE \left(2, \blx-\blx_1, 1, \frac{\Per}{\PE (M, \blx_1, 1, 0)}\right)}{\blx-\blx_1}. 
      \end{equation}
For any  $(M,\blx,\Per)$ sequence such that  $\liminf_{\blx \rightarrow \infty} \tfrac{\ln M}{\blx}=R$, $\liminf_{\blx \rightarrow \infty} \tfrac{-\ln\Per}{\blx}=\EXa$, if we choose $\blx_1=\lfloor \tfrac{R}{r}\blx \rfloor$ since  $r> r_h(R,\EXa)$ we have,
 \begin{equation*}
\liminf_{\blx \rightarrow \infty}  \tfrac{-1}{\blx-\blx_1} \ln \tfrac{\Per}{\PE (M, \blx_1, 1, 0)}>0.
 \end{equation*}
Furthermore as a result of Lemma \ref{lem:harimp}  and the convexity of $\hexpn(R)$ we have
 \begin{equation*}
\liminf_{\blx \rightarrow \infty}  \tfrac{-1}{\blx-\blx_1} \ln \tfrac{\Per}{\PE (M, \blx_1, 1, 0)}   \leq T^*.
 \end{equation*}
Assume for the moment that for any $T\in (0, T^*]$ and for any sequence of $\Per^{(\blx)}$ such that $\liminf_{\blx \rightarrow \infty} \tfrac{- \ln \Per^{(\blx)}}{\blx}=T$ we have
\begin{equation}
  \label{eq:assumptionxt}
\liminf_{\blx \rightarrow \infty} \tfrac{- \ln \PE(2,\blx,1,\Per^{(\blx)})}{\blx} \leq  
\obhexp{T}.
\end{equation}
Using  equation (\ref{eq:newconx1})  and taking the limit as $\blx$ goes to infinity we get
\begin{equation*}
\TEX \left(R,\EXa \right) \leq \tfrac{R}{r}  \TEX (r) + (1-\tfrac{R}{r})\obhexp{\tfrac{r\EXa- R \TEX (r)} {r-R}}.
\end{equation*}
Then  Theorem \ref{thm:con2}  follows from Lemma \ref{lem:harimp} and the fact that $\obhexp{T}$ is nondecreasing function of $T$,   

In order to establish  equation  (\ref{eq:assumptionxt}); note that if $T_0>0$ and $T\leq T_0$ then $\obhexp{T}=\infty$. Thus equation  (\ref{eq:assumptionxt}) holds trivially. For $T>T_0$ case we  prove  equation (\ref{eq:assumptionxt})  by contradiction. Assume that (\ref{eq:assumptionxt}) is wrong.  Then there exists a block code with  erasures that satisfies 
\begin{align*}
\PCX{\tilde{\DECS}(\tilde{\tilde{\dmes}})}{\tilde{\dmes}}
&\leq e^{-\blx (\obhexp{T}+\so{1})}
& \PCX{\DECS(\Era)}{\tilde{\dmes}}
&\leq e^{-\blx (T+\so{1})}
\\
\PCX{\tilde{\DECS}(\tilde{\dmes})}{\tilde{\tilde{\dmes}}}
&\leq e^{-\blx (\obhexp{T}+\so{1})}
&\PCX{\DECS(\Era)}{\tilde{\tilde{\dmes}}}
&\leq  e^{-\blx (T+\so{1})}
  \end{align*}
Enlarge the decoding region of $\tilde{\dmes}$ by taking its union with the erasure region: 
\begin{equation*}
\DECS'(\tilde{\dmes})=\DECS(\tilde{\dmes}) \cup \DECS(\Era)
\qquad
\DECS'(\tilde{\tilde{\dmes}})=\DECS(\tilde{\tilde{\dmes}}) 
\qquad
\DECS'(\Era)= \emptyset.
\end{equation*}
The resulting code  is an erasure free code with 
\begin{align*}
\PCX{\DECS'(\tilde{\tilde{\dmes}})}{\tilde{\dmes}}
&\leq e^{-\blx (\obhexp{T}+\so{1})}
&\mbox{and}   &&
\PCX{\DECS'(\tilde{\dmes})}{\tilde{\tilde{\dmes}}}
&\leq e^{-\blx (\min\{\obhexp{T},T\}+\so{1})}
  \end{align*}
Since $T_0<T \leq T^*$, $\obhexp{T}\geq T$, this contradicts with Lemma \ref{lem:E2=F2} thus equation (\ref{eq:assumptionxt}) holds.

\item $r=r_h(R,\EXa)$: Apply  Theorem \ref{thm:con1} with $s=0$, $L=1$, $L_1=1$ and $\blx_1=\max\{\ell: \PE  (M, \ell, 1, 0) >\Per \ln \tfrac{1}{\Per} \}$,
\begin{align}
\notag
\PE  (M, \blx, 1, \Per)
&\geq \PE (M, \blx_1, 1, 0)  \PE (2, \blx-\blx_1, 1, \tfrac{\Per}{\PE (M, \blx_1, 1, 0)})\\
&\geq \Per \ln \tfrac{1}{\Per} \PE (2, \blx-\blx_1, 1, \tfrac{1}{-\ln \Per})
\label{eq:newconx1a}
\end{align}
Note that for $\blx_1=\max\{\ell: \PE  (M, \ell, 1, 0) >\Per \ln \tfrac{1}{\Per} \}$,
  \begin{equation*}
\liminf_{\blx \rightarrow \infty} \tfrac{\blx_1}{\blx} \TEXn(\tfrac{R \blx}{\blx_1})= \EXa
      \end{equation*}
Then as a result of Lemma \ref{lem:harimp} we have,
  \begin{equation*}
\liminf_{\blx \rightarrow \infty} \tfrac{\blx_1}{\blx} \hexpn(\tfrac{R \blx}{\blx_1})\geq \EXa
      \end{equation*}
Then
\begin{equation}
\label{eq:newconx1d}
\liminf_{\blx \rightarrow \infty} \tfrac{\blx_1}{\blx} \geq \tfrac{R}{r_h(R,\EXa)}
\end{equation}
Assume for the moment that for any $\epsilon_{\blx}$ such that $\liminf_{\blx \rightarrow \infty}  \epsilon_{\blx} =0$
\begin{equation}
\label{eq:newconx1e}
\liminf_{\blx \rightarrow \infty} \tfrac{- \ln\PE (2, \blx, 1, \epsilon_{\blx} )}{\blx} \leq \obhexp{0}
\end{equation}
Then taking the logarithm of both sides of the equation (\ref{eq:newconx1a}), dividing both sides by $\blx$, taking the limit as $\blx$ tends to infinity and substituting equations  (\ref{eq:newconx1d}) and (\ref{eq:newconx1e})  we get,
\begin{equation}
\label{eq:newconx1f}
\TEX \left(R,\EXa \right) \leq \EXa+ (1-\tfrac{\EXa}{\hexpn(r_h(R,\EXa))})  \obhexp{0}
\end{equation}
Note that, Theorem \ref{thm:con2} for  $r=r_h(R,\EXa)$  case is equivalent to (\ref{eq:newconx1f}). Identity given in (\ref{eq:newconx1e}) follows from an analysis similar to the one used for establishing 
(\ref{eq:assumptionxt}), in which  but instead of Lemma \ref{lem:E2=F2}, we use a simple typicality argument   like \cite[Corollary 1.2]{CK}.
\end{itemize}
\end{proof}
We have set $L_1=1$ in the proof.   If instead of $L_1=1$ we had chosen $L_1$ to be a subexponential function of $\blx$ which grew to infinity with $\blx$, the logic and  the mechanics of the proof would still work but we would have replaced $\obhexp{T}$  with $\TEX(0,\EXa)$, while keeping the term including $\hexpn(R)$ the same. Since the best known upper bound for $\TEX(0,\EXa)$ is $\obhexp{\EXa}$ for $\EXa\leq T^*$ final result is same for case with feedback.\footnote{In binary symmetric channels  these result can be strengthened  using the value of $\TEXn(0)$, \cite{zig-c}. However those changes will improve the upper bound on error exponent only at low rates and high erasure exponents.} On the other hand for the case without feedback, which is not the main focus of this paper, this  does make a difference. By choosing $L_1$ to be a function of block length that goes to infinity  subexponentially with block length  one can use  Telatar's converse result \cite[Theorem 4.4]{emre} on the error exponent at zero rate and zero erasure exponent without feedback.


In Figure \ref{fig:exp3},  the upper and lower bounds we have derived for error exponent are plotted as a function of  erasure exponent for a binary symmetric channel with cross over probability $\epsilon=0.25$ at  rate  $R=8.62 \times 10^{-2}$ nats per channel use. Solid lines are lower bounds to the error exponent for block codes  with feedback, which have been established in Section \ref{sec:ach},  and without feedback,  which was   established previously, \cite{for}, \cite{CK}, \cite{emre}.  Dashed lines are the upper bounds obtained using Theorem \ref{thm:con2}.  

Note that all four curves  meet at a point on bottom right,  this is the point that corresponds to the error exponent of block codes at rate $R=8.62 \times 10^{-2}$ nats per channel use and its values are the same with and without feedback since we are on a symmetric channel and our rate is over the critical rate.  Any point to the lower right of this point is achievable both with and without feedback.

\begin{figure}[ht]
\setlength{\unitlength}{1pt}
\hspace{-0.8cm}
\includegraphics[scale=0.50]{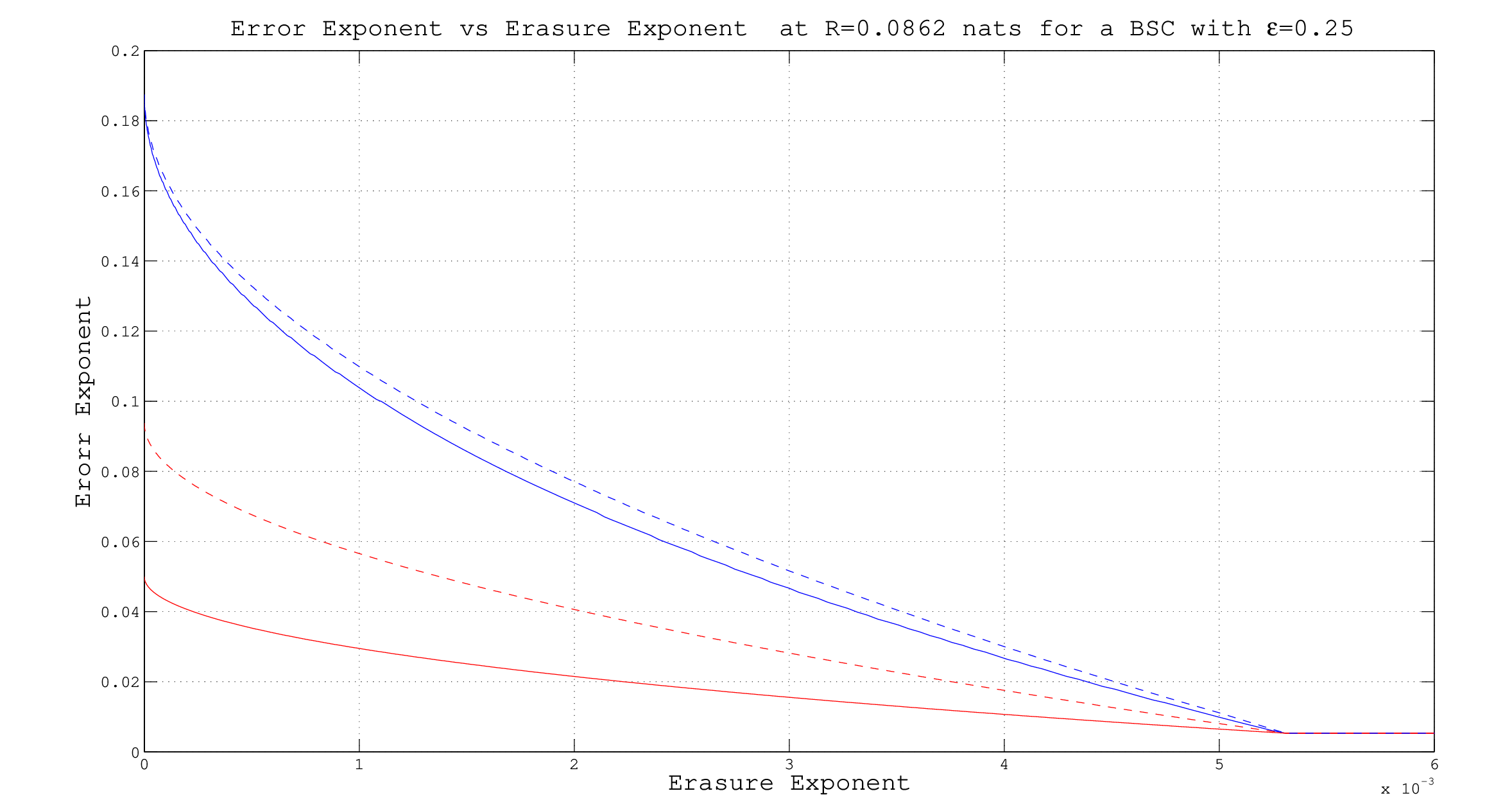}
\caption{Error Exponent vs Erasure Exponent}
\label{fig:exp3}
\end{figure}
The proximity of the inner and the outer bound demonstrated in Figure \ref{fig:exp3} is not  particular to the channel we have chosen. A discussion of the closeness of the  inner and outer bounds are given in Section \ref{sec:disc}.

\section{Erasure Exponent of Error-Free Codes:$ \TEXz  (R)$}\label{sec:efree}
For all DMCs which have one or more zero probability transitions, for all rates below capacity,  $R\leq  \CX$ and for small enough $\EXa$'s, $\EXo(R,\EXa)=\infty$. For such $(R,\EXa)$ pairs, coding scheme we have described in Section \ref{sec:ach}  gives us an error free code. The connection between the erasure exponent of error free  block codes, and error exponent of  block codes with erasures is not confined to this particular encoding scheme. In order to explain those connections in more detail let us first define the error-free codes more formally.
\begin{definition}
  A sequences $\SC_0$ of block codes with feedback is an  error-free reliable sequence  iff 
  \begin{align*}
  \Pe^{(\blx)}&=0  \quad\forall  \blx,  & &\mbox{and}  &
\limsup\nolimits_{\blx \rightarrow \infty} (\Per^{(\blx)}+\tfrac{1}{|\mesS^{(\blx)}|})&=0.
 \end{align*}
\end{definition}
The highest rate achievable for error-free reliable codes is the  zero-error capacity with feedback and erasures, $\CX_{x,0}$.

If all  the transition probabilities are positive i.e. $\min_{\dinp,\dout} \CT{\dinp}{\dout}=\delta>0$, then $\tfrac{\PCX{\dfsy^{\blx}}{\dmes}}{\PCX{\dfsy^{\blx}}{\tilde{\dmes}}}\geq (\tfrac{\delta}{1-\delta})^{n}$ for all $\dmes,\tilde{\dmes} \in {\mesS}$ and $\dfsy^{\blx} \in {\fsyS}^{\blx}$. Thus we have
\begin{equation}
\PCX{\dmes}{\dfsy^{\blx}} \geq (\tfrac{\delta}{1-\delta})^{n}  \PCX{\tilde{\dmes}}{\dfsy^{\blx}} \qquad \forall \dmes,\tilde{\dmes} \in \mesS , \forall \dfsy^{\blx} \in {\fsyS}^{\blx}
\end{equation}
Consequently we have 
$\Pe\geq \tfrac{(e^{\blx R}-1)\delta^{\blx}}{(e^{\blx R}-1)\delta^{\blx}+(1-\delta)^{\blx}}$
and $\CX_{x,0}$ is zero. On the other hand as an immediate consequence of the encoding scheme suggested by Yamamoto and Itoh in \cite{ito},  if there is one or more zero probability transitions,  $\CX_{x,0}$ is equal to channel capacity $\CX$.
\begin{definition}\label{def:zer}
For all DMCs with at least one $(\dinp,\dout)$ pair such that $\CT{\dinp}{\dout}=0$, $\forall R \leq \CX $  erasure exponent of error free block codes with feedback is defined as 
\begin{equation}
  \TEXz (R) \triangleq  \displaystyle{\sup_{ \SC_0 : \substack{ R (\SC_0) \geq R }} \EXa(\SC_0)}.
\end{equation} 
 \end{definition}
For any erasure exponent,  $\EXa$ less than $\TEXz(R)$, there is an error-free reliable sequence, i.e. there is a reliable sequence with infinite error exponent:
\begin{equation}
\label{eq:boundoncexq0}
\EXa \leq \TEXz(R) \Rightarrow  \TEX(R,\EXa)=\infty. 
\end{equation}
More interestingly if $\EXa> \TEXz(R)$ then $\TEX(R,\EXa)<\infty$. In order to see this let $\delta$ be the minimum non-zero transition probability. Then for any $\dmes, \tilde{\dmes} \in \mesS$ and $\dfsy\in \fsyS$ such that $\PCX{\dmes}{\dfsy^{\blx}} \PCX{\tilde{\dmes}}{\dfsy^{\blx}}>0$ we have $\PCX{\dfsy^{\blx}}{\dmes}\geq \delta^{n} \PCX{\dfsy^{\blx}}{\tilde{\dmes}}$. Thus if
 $\PCX{\est \notin \{\mes,\Era \}}{\fsy^{\blx}}\neq 0$ then  $\PCX{\est \notin \{\mes,\Era \}}{\fsy^{\blx}}\geq \tfrac{\delta^{\blx}}{1+\delta^{\blx}}$. Using this we get,
\begin{align}
\notag
\EX{\IND{\PCX{\est \notin \{\mes,\Era \}}{\fsy^{\blx}}\neq 0}}
&\leq  \tfrac{1+\delta^{\blx}}{\delta^{\blx}} \EX{\IND{\PCX{\est \notin \{\mes,\Era \}}{\fsy^{\blx}}\neq 0} \PCX{\est \notin \{\mes,\Era\}}{\fsy^{\blx}}}\\
\label{eq:boundoncexq}
&= (1+\delta^{-\blx}) \PX{\est \notin \{\mes,\Era \}}
\end{align}
Equation (\ref{eq:boundoncexq}) reveals that the total probability of $\dfsy^{\blx}$'s at which receiver chooses to decode to a message rather than declaring an erasure despite the fact that it is not certain about the message is upper bounded by $(1+\delta^{-\blx})$ times the undetected error probability. Thus if we replace the decoder with a new decoder which   declares an erasure unless it is sure about the transmitted message, i.e. unless there is a message with posterior probability one, resulting erasure probability $\Per'$ will be bounded in terms of original error and erasure probabilities as follows,
\begin{equation}
\label{eq:boundoncexq2}
\Per'\leq \Per+ (1+\delta^{-\blx})\Pe.
\end{equation}
Thus by changing the decoding rule, any length $\blx$  code with error probability $\Pe$ and erasure probability $\Per$ can be transformed into  error free code with erasure probability $\Per'$, where $\Per'$ satisfies equation (\ref{eq:boundoncexq2}).  Using this transformation we can change any code with errors-and-erasure decoding into a error free block code with erasures. Evidently we can use the very same  transformation to convert reliable sequences into error-free reliable sequences. Considering  error and erasure exponents of the original  reliable sequences and  erasure  exponents of resulting error free reliable sequences we get, 
\begin{equation}
\label{eq:boundoncexq3}
\TEXz(R)\geq \min \{\EXa,\TEX(R,\EXa)+\ln \delta\} \qquad \forall R,\EXa.
\end{equation}
Consequently,
\begin{equation}
\label{eq:boundoncexq4}
\EXa> \TEXz(R) \Rightarrow \TEX(R,\EXa)\leq \TEXz(R)-\ln \delta<\infty.  
\end{equation}
As a result of equations (\ref{eq:boundoncexq0}) and  (\ref{eq:boundoncexq4}) we can conclude that $\TEX(R,\EXa)=\infty$ if and only if $\EXa\leq \TEXz(R)$. In a sense like the error exponent of erasure free block codes,  $\TEXn(R)$, erasure exponent of the error free bock codes, $\TEXz(R)$,  gives a partial description of $\TEXn(R,\EXa)$.  $\TEXn(R)$ gives the value of error exponents below which erasure exponent can be pushed to infinity and $\TEXz(R)$ gives the value of erasure exponent below which error exponent can be pushed to infinity.

Below the erasure exponent of zero-error codes, $\TEXz(R)$, is investigated separately for two families of channels:   Channels which have a positive zero error capacity, i.e.   $\CX_0>0$ and  Channels which have zero zero-error capacity, i.e.  $\CX_0=0$. 

\subsection{Case  1:  $\CX_0>0$} 
\begin{theorem}
  \label{thm:nonzero}
  For a DMC if $\CX_0>0$ then,
  \begin{equation*}
\hexp(R)\geq \TEXz(R)   \geq   \spexp(R).
  \end{equation*}
\end{theorem}
\begin{proof}
 If zero-error capacity is strictly greater then zero, i.e.  $\CX_0>0$,  then one can achieve the sphere packing exponent, with zero error probability using  a  two phase scheme. In the first phase  transmitter uses a length $\blx_1=\lceil e^{\blx_1 R}\rceil$ block code without feedback  with a list decoder of size $L=\left\lceil \tfrac{\partial}{\partial R} \spexp(R,P_{R}^{*}) \right\rceil$   where $P_{R}^{*}$ is the input distribution satisfying $\spexp(R)=\spexp(R,P_{R}^{*})$. Note that with this list size the  sphere packing exponent\footnote{Indeed this upper bound on error probability is tight exponentially for block codes without feedback.}  is achievable at rate $R$. Thus correct message is in the list with at least probability  $(1-e^{-\blx_1 \spexp(R)})$, see \cite[Page 196]{CK}.  In the second phase transmitter uses a zero error code, of length\footnote{For some DMCs with $\CX_0>0$ and for some $L$ one may  need more than $ \lceil\tfrac{\ln (L+1)}{\CX_0}\rceil$ time units to convey one of the $(L+1)$ messages without any errors, because $\CX_0$ itself is defined as a limit. But even in those cases we are guaranteed to have a fixed amount of time for that transmissions, which does not change with $\blx_1$. Thus  above argument holds as is even in those cases.} $\blx_2= \lceil\tfrac{\ln (L+1)}{\CX_0}\rceil$  with $L+1$ messages, to tell the receiver whether the correct message is in that list or not, and the correct message itself if it is in the list.   Clearly such a feedback code with two phases is error free, and it has erasures only when there exists an error in the first phase. Thus the erasure probability of the over all code is upper bounded by $e^{-\blx_1 \spexp(R)}$. Note that $\blx_2$ is fixed for a given $R$. Consequently as the length of the first phase, $\blx_1$, grows to infinity the rate and erasure exponent of $(\blx_1+\blx_2)$ long block code converges to the rate and error exponent of $\blx_1$ long code of the first phase, i.e. to $R$ and  $\spexp(R)$. Thus
 \begin{equation*}
   \TEXz(R)   \geq   \spexp(R).
 \end{equation*}
Any  error free  block code with erasures can be forced to decode, at erasures. The resulting fixed length code has  an error probability no larger than the erasure probability of the original code. However  we know that, \cite{har}, error probability of the erasure free block codes with feedback decreases with an exponent no larger than $\hexp(R)$. Thus,
\begin{equation*}
\TEXz(R)   \leq \hexp(R).
  \end{equation*}
This upper bound on the erasure exponent also follows from the converse result we present in the next section, Theorem \ref{thm:con4}.
  \end{proof}
For symmetric channels $\hexp(R)=\spexp(R)$ and Theorem \ref{thm:nonzero} determines the erasure exponent of error-free codes on symmetric channels with non-zero zero-error-capacity completely. 

\subsection{Case 2: $\CX_0=0$}
This case is more involved than the previous one. We first establish an upper bound on $\TEXz(R)$ in terms of the improved version of Haroutunian's bound, i.e. Lemma  \ref{lem:harimp},  and the erasure exponent of error-free codes at zero rate, $\TEXz(0)$.  Then we show that $\TEXz(0)$ is equal to the erasure exponent error-free block codes with two messages, $\XEtwo$,  and bound  $\XEtwo$ from below.

For any $M$, $\blx$ and $L$, $\PE(M,\blx, L,\Per)=0$ for large enough $\Per$. We denote the minimum of such $\Per$'s by $\PERA(M,\blx, L)$. Thus we can write $\XEtwo$ as 
\begin{equation*}
\XEtwo=\liminf_{\blx \rightarrow \infty} \PERA \left( 2, \blx, 1  \right). 
\end{equation*}

\begin{theorem}
\label{thm:con3}
      For any $\blx$, $M$, $L$,  $\blx_1 \leq \blx$ and $L_1$,  minimum erasure probability of fixed length error-free block codes  with feedback, $\PERA(M,\blx, L)$,  satisfies
      \begin{equation}
      \label{eq:thmcon3}
\PERA  (M, \blx, L)  \geq \PE (M, \blx_1, L_1, 0) \PERA \left( L_1+1, \blx-\blx_1, L  \right). 
      \end{equation}
\end{theorem}
Like  Theorem \ref{thm:con1},  Theorem \ref{thm:con3} is correct both with and without feedback. Although 
$\PERA$'s and $\PE$  will be different in each case, the relationship between them given in equation (\ref{eq:thmcon3}) holds in both cases.
\begin{proof}
If $\PE (M, \blx_1, L_1, 0)=0$ theorem holds trivially. Thus we assume henceforth that  $\PE (M, \blx_1, L_1, 0)>0$. Using Theorem \ref{thm:con1} with $\Per=\PERA  (M, \blx, L)$ we get
\begin{equation*}
\PE \left(M,\blx,L,\PERA  (M, \blx, L) \right)\geq \PE (M, \blx_1, L_1, 0) \PE \left(L_1+1, (\blx-\blx_1), L, \tfrac{\PERA  (M, \blx, L)}{\PE (M, \blx_1, L_1, 0) } \right). 
\end{equation*}
Since $\PE \left(M,\blx,L,\PERA  (M, \blx, L) \right)=0$ and $\PE (M, \blx_1, L_1, 0)>0$ we have,
\begin{equation*}
  \PE \left(L_1+1, (\blx-\blx_1), L, \tfrac{\PERA  (M, \blx, L)}{\PE (M, \blx_1, L_1, 0)} \right)=0.
\end{equation*}
Thus 
\begin{equation*}
   \tfrac{\PERA  (M, \blx, L)}{\PE (M, \blx_1, L_1, 0)} \geq  \PERA \left(L_1+1, (\blx-\blx_1), L \right).
\end{equation*}  
\end{proof}
As we have done in the errors-and-erasures case we can convert this into a bound on exponents. If we use
the improved version of Haroutunian's bound, i.e. Lemma  \ref{lem:harimp},  as an upper bound on the error exponent of erasure free block codes we get the following.
\begin{theorem}
\label{thm:con4}
For any rate $R \geq 0$ for any $\ts \in \left[  \tfrac{R}{\CX},1 \right]$
\begin{equation*}
\TEXz \left(R \right) \leq \ts \hexpn \left(\tfrac{R}{\ts} \right) + (1-\tst) \TEXz \left(0 \right)
\end{equation*}
\end{theorem}
Now let us focus on  the value of erasure exponent at zero rate:
\begin{lemma}
\label{lem:epfl}
For the channels which has zero zero-error capacity, i.e. $\CX_0=0$, erasure exponent of error free block codes at zero rate $\TEXz(0)$ is equal to the erasure exponent of error free block codes with two messages $\XEtwo$.
\end{lemma}
Note that unlike the two message case, $\XEtwo$, in the zero rate case $\TEXz(0)$ the number of messages are increasing with block length to infinity, thus we can not claim $\XEtwo=\TEXz(0)$  just as a result of their definitions.
\begin{proof}
If we write Theorem \ref{thm:con3} for $L=1$, $\blx_1=0$ and $L_1=1$
\begin{align}
\notag
\PERA  (M, \blx, 1)
&\geq  \PE  (M, 0, 1)  \PERA  (2, \blx, 1)\\
\notag
&=\tfrac{M-1}{M} \PERA  (2, \blx, 1) \qquad \forall M,\blx
\end{align}
Thus as an immediate result of the definitions of $\TEXz(0)$ and $\XEtwo$, we have $\TEXz(0)\leq \XEtwo$. 

In order to prove the equality one needs to prove $\TEXz(0) \geq \XEtwo$.  For doing that let  us assume that  it is possible to  send one bit with erasure probability $\epsilon$ with a block code of length  $\ell(\epsilon)$:
\begin{equation}
\epsilon \geq   \PERA  (2, \ell(\epsilon), 1)
\end{equation}
One can use this code to send $r$ bits, by repeating each bit whenever there exists an erasure. If the block length is $\blx=k \ell(\epsilon)$ then a message erasure occurs only when the number of bit erasures in $k$ trials is more then $k-r$. Let $\#e$ denote the number of erasures out of $k$ trials then
 \begin{align*}
\PX{\#e=l}&=\tfrac{k!}{(k-l)! l!} (1-\epsilon)^{k-l} \epsilon^l  &\mbox{and}&&
\Per&= \sum\nolimits_{l=k-r+1}^{k}  \PX{\#e=l}.
 \end{align*}
Thus
 \begin{align*}
\Per
&= \sum\nolimits_{l=k-r+1}^{k}  \tfrac{k!}{l! (k-l)!} (1-\epsilon)^{k-l}\epsilon^{l}\\
&= \sum\nolimits_{l=k-r+1}^{k}  \tfrac{k!}{l! (k-l)!} \left(\tfrac{l}{k}\right)^{l}  \left(1-\tfrac{l}{k}\right)^{k-l} e^{- [ l \ln \frac{l/k}{\epsilon} + (k-l) \ln\frac{1-l/k}{1-\epsilon}]}\\
&= \sum\nolimits_{l=k-r+1}^{k}  \tfrac{k!}{l! (k-l)!} \left(\tfrac{l}{k}\right)^{l}  \left(1-\tfrac{l}{k}\right)^{k-l}  e^{- k\DIV{\frac{l}{k}}{\epsilon}}.
 \end{align*}
Then for any $ \epsilon\leq 1-\tfrac{r}{k}$, we have 
\[\Per \leq  e^{- k \DIV{1-\frac{r}{k}}{\epsilon}}.\]
Evidently $\Per \geq \PERA(2^r, \blx, 1)$ for $\blx=k\ell(\epsilon)$. Thus, 
\begin{equation*}
\tfrac{-\ln \PERA(2^r, \blx, 1) }{\blx} \geq \tfrac{\DIV{1-\frac{r}{k}}{\epsilon}}{\ell(\epsilon)}.
\end{equation*}
Then $\tfrac{-\ln \epsilon}{\ell(\epsilon)}$ is an achievable erasure exponent for any sequence of $(r,k)$'s such that $\lim_{k \rightarrow \infty} \frac{r}{k}=0$, i.e. $\TEXz(0)   \geq \tfrac{-\ln \epsilon}{\ell(\epsilon)}$. Thus any exponent achievable for two message case is achievable for zero rate case: $\TEXz(0) \geq \XEtwo$.
\end{proof}

As a result of Lemma \ref{lem:zerobound}  which is presented in the next section we know that
\begin{equation*}
\PERA(2, \blx, 1)  \geq   (\sup_{s \in (0,.5)}  \beta(s))^{\blx}  \quad \mbox{where}  \quad \beta(s)=\min\nolimits_{\dinp, \tilde{\dinp}}  \sum\nolimits_{\dout}  \CT{\dinp}{\dout}^{(1-s)} \CT{\tilde{\dinp}}{\dout}^{s}.
\end{equation*}
Thus as a result of Lemma \ref{lem:epfl} we have
\begin{equation*}
  \TEXz(0) =\XEtwo \leq  -\ln \sup_{s \in (0,0.5)}  \beta(s).
\end{equation*}

\subsection{Lower  Bounds on $\PERA(2, \blx, 1)$}\label{subsec:xetwo}
Suppose  at time $t$ the correct message, $\mes$, is assigned to the input letter $\dinp$ and  the other message is assigned to  the input letter  $\tilde{\dinp}$, then the receiver can not to rule out the  incorrect message at time $t$ with probability $\sum_{y: \CT{\tilde{\dinp}}{\dout}>0} \CT{\dinp}{\dout}$. Using this fact one can prove that,
\begin{equation}
\label{eq:boundone}
\PERA(2, \blx, 1)  \geq \left(\min\nolimits_{\dinp,\tilde{\dinp}}  \sum\nolimits_{\dout: \CT{\tilde{\dinp}}{\dout}>0} \CT{\dinp}{\dout} \right)^{\blx}.
\end{equation}
Now let us consider channels whose transition probability matrix  $\CTM$ is of the form
\begin{equation}
\label{eq:exzchan}
  \CTM=\left[
 \begin{matrix}
    1-q         &   q  \\
    0           &   1           
   \end{matrix} \right].
\end{equation}
We denote the output letter  that can be reached from both of the input letters by $\tilde{\dout}$.  For the moment we consider only the  deterministic encoding schemes, i.e. $\fsy_{t}=\out_{t}$. Note that in the optimal encoding scheme,
\begin{equation*}
 \ENC_t(1,\dout^{t-1}) \neq  \ENC_t(2,\dout^{t-1}) 
 \qquad \forall t, ~~ \forall \dout^{t-1} \in \outS^{t-1}
\end{equation*}
Then 
\begin{equation}
\label{eq:optenccon}
 \PCX{\out_t=\tilde{\dout}}{\mes = 1, \dout^{t-1}}
 \PCX{\out_t=\tilde{\dout}}{\mes = 2, \dout^{t-1}}= q   
 \qquad \forall t, ~~ \forall \dout^{t-1} \in \outS^{t-1}
 \end{equation}
Furthermore if $\out^{\blx}=\tilde{\dout}^{\blx}$ then the receiver can not decode without errors, i.e. it has to declare an erasure.  Then,
\begin{align}
\PERA(2, \blx, 1)
\notag
&\mathop{\geq}^{} \tfrac{1}{2} (\PCX{\out^{\blx}=\tilde{\dout}\tilde{\dout} \ldots \tilde{\dout}}{{\mes} = 1}+\PCX{\out^{\blx}=\tilde{\dout}\tilde{\dout} \ldots \tilde{\dout}}{{\mes} = 2})\\
\notag
&\mathop{\geq}^{(a)} \sqrt{\PCX{\out^{\blx}=\tilde{\dout}\tilde{\dout} \ldots \tilde{\dout}}{{\mes} = 1}\PCX{\out^{\blx}=\tilde{\dout}\tilde{\dout} \ldots \tilde{\dout}}{{\mes} = 2}}\\
\label{eq:boundtwo-i}
&\mathop{=}^{(b)} q^{\frac{\blx}{2}}
\end{align}
where $(a)$ hods because arithmetic mean is larger than the geometric mean and $(b)$  follows from  the equation (\ref{eq:optenccon}).

For the $\CTM$  given in (\ref{eq:exzchan}) the bound given in (\ref{eq:boundtwo-i})   is very tight. If the  encoder assigns the first message to the input letter that always leads to $\tilde{\dout}$ and the second message to the other input letter  in first  $\lfloor \tfrac{\blx}{2} \rfloor$ time instances, and  does the flipped  assignment in the last   $\lceil \tfrac{\blx}{2} \rceil$ time instances, then an erasure happens with a probability less than $q^{\lfloor \tfrac{\blx}{2} \rfloor}$, i.e.  $\PERA(2, \blx, 1)\leq q^{\lfloor \tfrac{\blx}{2} \rfloor}$.

On the other hand for the $\CTM$ given in (\ref{eq:exzchan}), bound given in  equation (\ref{eq:boundone}) ensures only  $\PERA(2, \blx, 1)\geq q^{\blx}$, rather than $\PERA(2, \blx, 1)\geq q^{\lfloor \tfrac{\blx}{2} \rfloor}$.  Thus for the channel given (\ref{eq:exzchan}) the bound given in equation (\ref{eq:boundtwo-i}) is tighter than the one in  equation (\ref{eq:boundone}).

The  idea used in deriving the bound given in  equation (\ref{eq:boundtwo-i}) for this particular $\CTM$ can be applied to a general DMC to prove the following lower bound,
\begin{equation}
  \label{eq:boundtwo}
 \PERA(2, \blx, 1)   \geq  \left(\min\nolimits_{\dinp, \tilde{\dinp}} \sum\nolimits_{\dout}  \sqrt{ \CT{\dinp}{\dout} \CT{\tilde{\dinp}}{\dout}} \right)^{\blx}.
  \end{equation}
The  bound given in equation (\ref{eq:boundtwo}) is decaying exponentially in $\blx$, even when all entries of the $W$ are positive, however for those channels  the bound given in (\ref{eq:boundone}) implies $\PERA(2, \blx, 1)\geq 1$. Thus the bound given in (\ref{eq:boundtwo})  can not be  superior to the bound given in equation (\ref{eq:boundone}) in general. The  following bound implies  bounds given in both equation (\ref{eq:boundone}) and equation (\ref{eq:boundtwo}). Furthermore for certain channels it is strictly better than both.
\begin{lemma}
\label{lem:zerobound}
Erasure probability of all  error free block codes with  two messages is lower bounded as
  \begin{equation}
\label{eq:zerobound}
\PERA(2, \blx, 1)  \geq  (\sup_{s \in (0,.5)}  \beta(s))^{\blx}  \quad \mbox{where}  \quad \beta(s)=\min\nolimits_{\dinp, \tilde{\dinp}}  \sum\nolimits_{\dout}  \CT{\dinp}{\dout}^{(1-s)} \CT{\tilde{\dinp}}{\dout}^{s}
  \end{equation}
\end{lemma}
Note that  bounds given  in equation (\ref{eq:boundone})  and (\ref{eq:boundtwo}) are  implied by $\lim_{s\rightarrow 0^+}\beta(s) $ and $\lim_{s\rightarrow 0.5^-}\beta(s)$  respectively. 

Although  {\small $\sum_{\dout}  \CT{\dinp}{\dout}^s \CT{\tilde{\dinp}}{\dout}^{1-s}$} is convex in $s$ on $(0,0.5)$ for all $(\dinp,\tilde{\dinp})$ pairs,  $\beta(s)$ is not convex in $s$ because of the minimization in its definition. Thus the supremum over $s$ does not necessarily occur on the boundaries. Indeed there are channels for which bound given in Lemma \ref{lem:zerobound} is strictly better than the bounds given in (\ref{eq:boundone}) and (\ref{eq:boundtwo}). Following is the transition probability matrix of  one such channel.

\begin{equation*}
\CTM=\left[
 \begin{matrix}
    0.1600  &  0.0200 &   0.2200  &  0.3000   & 0.3000\\
    0.0900  &  0.4000 &   0.2700  &  0.0002   & 0.2398\\
    0.1800  &  0.2000 &   0.3000  &  0.3200   &   0
   \end{matrix} \right]   
 \qquad  \qquad  \qquad 
 \begin{matrix}
\lim\limits_{s\rightarrow 0~}      \beta(s)=0.7000 \\
\lim\limits_{s\rightarrow 0.5}   \beta(s)=0.7027\\
 ~ ~~ \beta(0.18)=0.7299 . 
 \end{matrix}
\end{equation*}

\begin{proof} 
Let $\ers{t}$ and $\erst{t}{\dfsy^{t-1}}$ be,
\begin{align*}
\ers{t}   &=\{\dfsy^t: \PCX{\mes=1}{\dfsy^t} \PCX{\mes=2}{\dfsy^t}>0\}\\
\erst{t}{\dfsy^{t-1}} &=\{\dout_t: \PCX{\dout_t}{\mes=1,\dfsy^{t-1}} \PCX{\dout_t}{\mes=2,\dfsy^{t-1}}>0\}\\
\end{align*}
Then for any  error free code and for any $s\in (0,0.5)$ we have
\begin{align}
\Per
\notag
&=\EX{\IND{\ers{\blx}}}\\
\notag
&=\EX{\IND{\ers{\blx}}( \PCX{\mes=1}{\fsy^\blx} +\PCX{\mes=2}{\fsy^\blx})}\\
\notag
&=
\EX{\IND{\ers{\blx}}( (1-s)\PCX{\mes=1}{\fsy^\blx} +s\PCX{\mes=2}{\fsy^\blx})}
+
\EX{\IND{\ers{\blx}}( s\PCX{\mes=1}{\fsy^\blx} +(1-s)\PCX{\mes=2}{\fsy^\blx})}
\\
\label{eq:inteqs1}
&\geq\EX{\IND{\ers{\blx}}{\PCX{\mes=1}{\fsy^{\blx}}}^{1-s} {\PCX{\mes=2}{\fsy^{\blx}}}^s}
+\EX{\IND{\ers{\blx}}{\PCX{\mes=1}{\fsy^{\blx}}}^{s} {\PCX{\mes=2}{\fsy^{\blx}}}^{1-s}}
\end{align}
where the last inequality follows from the fact that arithmetic mean is lower bounded by the geometric mean. Furthermore,
\begin{align}
\notag
&\hspace{-.5cm}\EX{\IND{\ers{\blx}}{\PCX{\mes=1}{\fsy^{\blx}}}^{1-s} {\PCX{\mes=2}{\fsy^{\blx}}}^s}\\
\label{eq:inteqs2}
&=\EX{\ECX{\IND{\erst{\blx}{\fsy^{\blx-1}} } 
\left(\tfrac{ {\PCX{\mes=1}{\fsy^{\blx}}}}{{\PCX{\mes=1}{\fsy^{\blx-1}}}}\right)^{1-s}\!
\left(\tfrac{ {\PCX{\mes=2}{\fsy^{\blx}}}}{{\PCX{\mes=2}{\fsy^{\blx-1}}}}\right)^{s}  }{\fsy^{\blx-1}}
\!\IND{\ers{(\blx-1)}}{\PCX{\mes=1}{\fsy^{\blx-1}}}^{1-s} {\PCX{\mes=2}{\fsy^{\blx-1}}}^s}.
\end{align}
Note that,
\begin{align} 
\notag
\tfrac{ {\PCX{\mes=1}{\fsy^{\blx}}}}{{\PCX{\mes=1}{\fsy^{\blx-1}}}}
=\tfrac{ {\PCX{\mes=1}{\fsy^{\blx-1},\out_{\blx}}}}{{\PCX{\mes=1}{\fsy^{\blx-1}}}}\\
\label{eq:inteqs3}
=\tfrac{ {\PCX{\out_{\blx}}{\mes=1,\fsy^{\blx-1}}}}{{\PCX{\out_{\blx}}{\fsy^{\blx-1}}}}.
\end{align}
Similarly, 
\begin{align} 
\tfrac{\PCX{\mes=2}{\fsy^{\blx}}}{\PCX{\mes=2}{\fsy^{\blx-1}}}
\label{eq:inteqs4}
=\tfrac{ {\PCX{\out_{\blx}}{\mes=2,\fsy^{\blx-1}}}}{{\PCX{\out_{\blx}}{\fsy^{\blx-1}}}}.
\end{align}
Thus using equations  (\ref{eq:inteqs3}) and (\ref{eq:inteqs4}) we have
\begin{align} 
\ECX{\IND{\erst{\blx}{\fsy^{\blx-1}} } 
\left(\tfrac{ {\PCX{\mes=1}{\fsy^{\blx}}}}{{\PCX{\mes=1}{\fsy^{\blx-1}}}}\right)^{1-s}
\left(\tfrac{ {\PCX{\mes=2}{\fsy^{\blx}}}}{{\PCX{\mes=2}{\fsy^{\blx-1}}}}\right)^{s}  }{\fsy^{\blx-1}}
\notag
&=\ECX{\tfrac{{\PCX{\out_{\blx}}{\mes=1,\fsy^{\blx-1}}}^{1-s} {\PCX{\out_{\blx}}{\mes=2,\fsy^{\blx-1}}}^{s}}{{\PCX{\out_{\blx}}{\fsy^{\blx-1}}}}}{\fsy^{\blx-1}}\\
\notag
&= \sum_{\dout_{\blx}}  {\PCX{\dout_\blx}{\mes=1,\fsy^{\blx-1}}}^{1-s} {\PCX{\dout_\blx}{\mes=2,\fsy^{\blx-1}}}^s\\
\label{eq:inteqs5} 
&\geq \beta(s)  
\end{align}
where the last inequality follows from the definition of  $\beta(s)$ given in equation (\ref{eq:zerobound}).

Using equations (\ref{eq:inteqs2}) and (\ref{eq:inteqs5}) we get
\begin{align}
\EX{\IND{\ers{\blx}}{\PCX{\mes=1}{\fsy^{\blx}}}^{1-s} {\PCX{\mes=2}{\fsy^{\blx}}}^s}
\notag
&\geq\EX{\PCX{\mes=1}{\fsy^0}^{1-s}  \PCX{\mes=2}{\fsy^{0}}^s} \beta(s)^{\blx}\\
\label{eq:inteqs6} 
&\geq\tfrac{1}{2} \beta(s)^{\blx}.
\end{align}
If we follow a similar line of reasoning for the second term in (\ref{eq:inteqs1}) we get
\begin{align}
\EX{\IND{\ers{\blx}}{\PCX{\mes=1}{\fsy^{\blx}}}^{s} {\PCX{\mes=2}{\fsy^{\blx}}}^{1-s}}
\notag
&\geq\tfrac{1}{2} \beta(1-s)^{\blx}\\
\label{eq:inteqs7}
&=\tfrac{1}{2} \beta(s)^{\blx}.
\end{align}
Lemma follows from equations  (\ref{eq:inteqs1}),  (\ref{eq:inteqs6}) and  (\ref{eq:inteqs7}) by taking the supremum over $s\in (0,0.5)$.
\end{proof} 

\section{Discussion}\label{sec:disc}
The value of error exponent is not known for erasure free fixed length block codes with feedback on a general DMC. We do not even know if it is still upper bounded by sphere packing exponent for non-symmetric DMCs.  Yet the value of error exponent for fixed length block codes with feedback and errors-and-erasures decoding can be deduced, for the zero-erasure exponent case, from the results on the variable length block  codes \cite{bur}, \cite{ito}.  Our main  aim in this paper was establishing  upper and lower bounds  that  extend the  bounds at the  zero erasure exponent case gracefully and non-trivially to the positive erasure exponents values.  Our results are best understood in this framework and  should be interpreted accordingly.

By finding the optimal error exponent erasure exponent trade off, one solves the open problem of finding the optimal  error exponent of erasure free fixed length block codes with feedback. This is an important and difficult problem on its own right. We did not attempted to solve that problem, yet the inner and outer bounds we have derived for the case with erasure quantify how much we loose from the optimal performance by using the encoding schemes inspired by the optimal encoding schemes for variable length block codes.

We derived  inner bounds  using  two phase encoding schemes, which are known to be optimal at zero-erasure exponent case. We have improved the performance  of these two phase schemes at positive erasure exponent values by choosing relative durations of the phases considering the desired values of  rate and erasure exponent,  and by using a decoder that takes into account the outputs of both phases while deciding between decoding to a message and declaring an erasure. However within each phase the assignment  of messages to input letters is fixed.  In a general  feedback encoder, on the other hand, assignment of the messages to input symbols at each time  can depend on  the previous channel outputs and such encoding schemes have proven to improve the error  exponent at low rates,  \cite{zig}, \cite{yak}, \cite{burbsc}, \cite{NZ3}, \cite{NT} for some DMCs.  Using such an encoding in the communication phase will improve the performance at low rates.  In addition instead of committing to a fixed duration for the communication phase one might consider using a stopping time to switch from communication phase to the control phase.  However in order to apply those ideas effectively  for a general DMC, it seems one first needs to solve the problem for the erasure free block codes for a general DMC.  

We derived the outer bounds without making any assumption about the feedback encoding scheme.  Thus they are valid for any fixed length block code with  feedback and erasures. The principal idea of the straight line bound is making use of the bounds derived for different rate,  erasure exponent pairs by taking their convex combinations. This approach  can be interpreted as a generalization of the outer bounds used for variable length block codes,  \cite{bur}, \cite{pet}.  As it was the case for the inner bounds, it seems in order to improve the outer bounds one needs establish outer bounds on two  related problems, i.e.  on the error exponents of erasure free block codes with feedback and on the error exponent erasure exponent trade off at zero rate.

The inner and outer bounds we have derived do not coincide for arbitrary values of erasure exponent. But they do coincide for all channels at all rates at zero erasure exponent.
\begin{itemize}
 \item If the channel does not have a zero probability transition, both the inner bound and the outer bound are equal to $(1-\tfrac{R}{\CX})\DX$. 
 \item If the channel does have a zero probability transition, the inner bound is equal to infinity and there are  fixed length block codes with zero error probability for all large enough block lengths.
  \end{itemize}
Furthermore on the plane where erasure exponent is equal to the error exponent, the outer bound we have derived is loose only as much as the best outer bound we know for the error exponent of the erasure free block codes with feedback is loose. Thus the proximity we have observed between inner and outer bounds in Figure \ref{fig:exp3} is not peculiar to the particular channel we have chosen for Figure \ref{fig:exp3}. For all channels inner and outer bounds we have derived coincide on the upper left corner like they do in  Figure \ref{fig:exp3}. If the channel is symmetric and if we are considering a rate over critical rate they will also coincide in lower right corner.  Furthermore if the sphere packing exponent is shown to be an upper bound for the error exponent of erasure free fixed length block codes this behavior will extend to non-symmetric channels.

\section*{Acknowledgment} 
Authors are grateful to Emre Telatar for his encouragement on the problem and for  numerous  discussions on error-free codes. In particular the observations presented about $z$-channels are his and Lemma \ref{lem:epfl} was  proved in 2006 summer at Ecole Polytechnique Federale de Lausanne (EPFL).  
Authors are thankful to Tsachy Weissman and  Amos Lapidoth for bringing the  Shannon-Gallager result mentioned in
Elwyn Berlekamp's thesis to their attention, to Anant Sahai for various discussions on communication problems with feedback and to Robert G. Gallager for various discussions on the encoding scheme presented in Section \ref{sec:ach} and the two message error exponent trade off.  Authors would like to acknowledge the thorough review provided by the anonymous Reviewer B, which has helped them to  improve the presentation of the paper in general. In addition  the $(1-s)$ factor on the left hand side of Theorem \ref{thm:con1} was pointed out to the authors by  Reviewer B.

\appendix

\subsection{The Error Exponent Trade Off for Feedback Encoding Schemes with   Two Message  and  Erasure Free Decoders  :}\label{app:E2=F2}

In this section we will first establish an alternative expression for the $\bhexp{T}{\XSID}$ function  defined  in equation (\ref{eq:bhtexpd}) in Lemma \ref{lem:altexpfexpt}. After that we will prove
 that in a two message code with feedback on a DMC, if the error exponent of one of the messages is greater than   some $ T\geq T_0$ then the  error exponent of the other message cannot be greater than $ \obhexp{T}$, where $T_0$ and $ \obhexp{T}$ are defined in    (\ref{eq:tzero}) and  (\ref{eq:ber2}) respectively.  Furthermore we will prove that if the error probability of the one of the message is zero than the error probability of the other message cannot be lower than $e^{-\blx T_0}$; we will also prove that it can be as low as $e^{-\blx T_0}$, see Lemma \ref{lem:zeroerxlem}.   These results   will imply that the error performance of a two message code, does not improve with feedback. This result is attributed to Shannon and Gallager by Berlekamp in  \cite{ber}. 

\begin{lemma}\label{lem:altexpfexpt}
$\bhexp{T}{\XSID}$ defined  in equation (\ref{eq:bhtexpd}) is  equal to
\begin{equation*}
\bhexp{T}{\XSID}=
\left\{\begin{array}{cll}
\infty
& \mbox{if~} &T< \CDIV{\Xq{0}}{W_a}{\XSID}\\
\CDIV{\Xq{s}}{W_r}{\XSID}
& \mbox{if~} &T=\CDIV{\Xq{s}}{W_a}{\XSID} \quad \mbox{for~some~}  s \in [0,1]\\
\CDIV{\Xq{1}}{W_r}{\XSID}
& \mbox{if~} &T> \CDIV{\Xq{1}}{W_a}{\XSID} 
\end{array}\right\}
\end{equation*} 
where
\begin{equation*}
\XQ{s}{\dout}{\dinp,\tilde{\dinp}}=
\left\{
\begin{array}{lcl}
\tfrac{\IND{\CT{\tilde{\dinp}}{\dout}>0}}{\sum_{\tilde{\dout}:\CT{\tilde{\dinp}}{\tilde{\dout}}>0} \CT{\dinp}{\tilde{\dout}}} \CT{\dinp}{\dout}
&\mbox{ if }&s=0\\
\tfrac{\CT{\dinp}{\dout}^{1-s} \CT{\tilde{\dinp}}{\dout}^s }{\sum_{\tilde{\dout}}\CT{\dinp}{\tilde{\dout}}^{1-s} \CT{\tilde{\dinp}}{\tilde{\dout}}^s}
&\mbox{ if }&s\in (0,1)\\
\tfrac{\IND{\CT{\dinp}{\dout}>0}}{\sum_{\tilde{\dout}:\CT{\dinp}{\tilde{\dout}}>0} \CT{ \tilde{\dinp}}{\tilde{\dout}}} \CT{\tilde{\dinp}}{\dout}
&\mbox{ if} &s=1\\
\end{array}\right\}
\end{equation*}
\end{lemma}

\begin{proof}
\begin{align}
\bhexp{T}{\XSID} 
\notag
&=\min_{\XSOD: \CDIV{\XSOD}{W_a}{\XSID}\leq T  } \CDIV{\XSOD}{W_r}{\XSID}\\ \
\notag
&\mathop{=}\min_{\XSOD} \sup_{\lambda>0} \CDIV{\XSOD}{W_r}{\XSID}+\lambda(\CDIV{\XSOD}{W_a}{\XSID}-T)\\ 
\notag
&\mathop{=}^{(a)}\sup_{\lambda>0}\min_{\XSOD} \CDIV{\XSOD}{W_r}{\XSID}+\lambda(\CDIV{\XSOD}{W_a}{\XSID}-T)\\ 
\notag
&=\sup_{\lambda>0} \min_{\XSOD}-\lambda T+ (1+\lambda)  \sum\nolimits_{\dinp,\tilde{\dinp},\dout} \XSID(\dinp,\tilde{\dinp}) U(\dout|\dinp,\tilde{\dinp})\ln \tfrac{ U(\dout|\dinp,\tilde{\dinp})}{ \CT{\dinp}{\dout}^{\frac{\lambda}{1+\lambda}} \CT{\tilde{\dinp}}{\dout}^{\frac{1}{1+\lambda}}}\\
\label{eq:etftp}
&\mathop{=}^{(b)}\sup_{\lambda>0} -\lambda T - (1+\lambda) \sum\nolimits_{\dinp,\tilde{\dinp}} \XSID(\dinp,\tilde{\dinp})  \ln \sum_{\dout} \CT{\dinp}{\dout}^{\frac{\lambda}{1+\lambda}} \CT{\tilde{\dinp}}{\dout}^{\frac{1}{1+\lambda}} 
\end{align}
where $(a)$ follows from convexity of $\CDIV{\XSOD}{W_r}{\XSID}+\lambda(\CDIV{\XSOD}{W_a}{\XSID}-T)$  in $U$ and linearity (concavity) of it in $\lambda$; $(b)$ holds because minimizing $U$ is $\Xq{s}$ for $s=\tfrac{1}{1+\lambda}$. The function on the right hand side of (\ref{eq:etftp}) is maximized at a positive and finite $\lambda$ iff there is a $\lambda$  such that $\CDIV{\Xq{\frac{1}{1+\lambda}}}{W_a}{\XSID}=T$. Thus by substituting $\lambda=\tfrac{1-s}{s}$ we get
\begin{equation}
\label{eq:SS0}
\bhexp{T}{\XSID}=
\left\{\begin{array}{cll}
\infty
& \mbox{if~} & T<  \lim\nolimits_{s \rightarrow 0^+} \CDIV{\Xq{s}}{W_a}{\XSID}\\
\lim\nolimits_{s \rightarrow 0^+} \CDIV{\Xq{s}}{W_r}{\XSID}
& \mbox{if~} &T=\lim\nolimits_{s \rightarrow 0^+}\CDIV{\Xq{s}}{W_a}{\XSID}\\
\CDIV{\Xq{s}}{W_r}{\XSID}
& \mbox{if~} &T=\CDIV{\Xq{s}}{W_a}{\XSID} \quad \mbox{for~some~}  s \in (0,1)\\
\lim\nolimits_{s \rightarrow 1^-} \CDIV{\Xq{s}}{W_r}{\XSID}
& \mbox{if~} &T=\lim\nolimits_{s \rightarrow 1^-}\CDIV{\Xq{s}}{W_a}{\XSID}\\
\lim\nolimits_{s \rightarrow 1^-} \CDIV{\Xq{s}}{W_r}{\XSID}
& \mbox{if~} &T>   \lim\nolimits_{s \rightarrow 1^-} \CDIV{\Xq{s}}{W_a}{\XSID}\\
\end{array}\right\}
\end{equation}  
Lemma follows from  the definition $\Xq{s}$ at $s=0,1$   and equation (\ref{eq:SS0}). 
\end{proof}

Now we are ready to present the proof of Lemma \ref{lem:E2=F2}

\begin{proofs}{Lemma \ref{lem:E2=F2}}
Our proof is very much like the one for the converse part of \cite[Theorem 5]{SGB}, except few modifications that allow us to handle the fact that encoding schemes we are considering are feedback encoding schemes. Like \cite[Theorem 5]{SGB} we construct a probability measure $\XPT{T}{\cdot}$ on ${\fsyS}^{\blx}$ as a function of $T$ and the encoding scheme. Then we bound the error probability of each  message from below using the probability of the decoding region of the other message under $\XPT{T}{\cdot}$. We consider probability measures on  ${\fsyS}^{\blx}$ rather than $\out^{\blx}$ to include the possible randomization in the encoding and decoding schemes.

For any  $T\geq T_0$ and $\XSID$, let $\WS{T}{\XSID}$ be
\begin{equation}
\label{eq:SS}
\WS{T}{\XSID}=
\left\{\begin{array}{cll}
0
& \mbox{if~} 
&T< \CDIV{\Xq{0}}{W_a}{\XSID}\\
s
&\mbox{if~} 
&\exists s \in [0,1]  \mbox{s.t. } \CDIV{\Xq{s}}{W_a}{\XSID}= T\\
1
& \mbox{if~} &T> \CDIV{\Xq{1}}{W_a}{\XSID} 
\end{array}\right\}.
\end{equation}
Recall that  
\begin{align*}
T_0 &=\max_{\dinp,\tilde{\dinp}} -\ln \sum\nolimits_{\dout: \CT{\tilde{\dinp}}{\dout}>0} \CT{\dinp}{\dout}
&&\mbox{and}&
\CDIV{\Xq{0}}{W_a}{\XSID}
&=-\sum\nolimits_{\dinp,\tilde{\dinp}} \XSID(\dinp,\tilde{\dinp}) \ln \sum\nolimits_{\dout: \CT{\tilde{\dinp}}{\dout}>0} \CT{\dinp}{\dout}.
\end{align*}
Then  for all $\XSID$ we have
\begin{equation}
\label{eq:SS1c}
 T_0\geq \CDIV{\Xq{0}}{W_a}{\XSID}.
\end{equation}
Thus as a result of definition of $\WS{T}{\XSID}$ and equation (\ref{eq:SS1c})  we have
\begin{equation}
\label{eq:SS1b}
  \CDIV{\Xq{\WS{T}{\XSID}}}{W_a}{\XSID}\leq T \qquad \forall T\geq T_0. 
\end{equation}
Using Lemma \ref{lem:altexpfexpt},  definition of $\WS{T}{\XSID}$ and equation (\ref{eq:SS1c}) we can also conclude that
\begin{equation}
\label{eq:SS1a}
 \CDIV{\Xq{\WS{T}{\XSID}}}{W_r}{\XSID}= \bhexp{T}{\XSID}\leq \obhexp{T} \qquad \forall T\geq T_0.
\end{equation}

Note that given  $\fsy^{t-1}=\dfsy^{t-1}$  channel input letters assigned to each message at time $t$, $\ENC_t(\dmes_1,\dfsy^{t-1})$ and $\ENC_t(\dmes_2,\dfsy^{t-1})$, are fixed  for any feedback encoding schemes,     $\ENC_t(\cdot): \{\dmes_1,\dmes_2\}  \times \fsyS^{t-1}$. Thus the corresponding $\XSID$ is given by:
\begin{equation}
\label{eq:SS2}
  \XSID(\dinp, \tilde{\dinp})=\left\{
    \begin{array}{lcc}
    0 & \mbox{if~} & (\dinp, \tilde{\dinp}) \neq (\ENC_t(\dmes_1,\dfsy^{t-1}), \ENC_t(\dmes_2,\dfsy^{t-1}))\\
    1 & \mbox{if~} & (\dinp, \tilde{\dinp})   =   (\ENC_t(\dmes_1,\dfsy^{t-1}), \ENC_t(\dmes_2,\dfsy^{t-1}))
   \end{array} \right\}.
\end{equation}
Then for any $T\geq T_0$  let $  \XPTC{T}{\dout_{t}}{\dfsy^{t-1}}$ be 
\begin{equation}
\label{eq:SS3}
  \XPTC{T}{\dout_{t}}{\dfsy^{t-1}}=\XQ{\WS{T}{\XSID}}{\dout_t}{\ENC_t(\dmes_1,\dfsy^{t-1}),\ENC_t(\dmes_2,\dfsy^{t-1})}.
\end{equation}
Furthermore let us assume that the conditional distribution of $\rsy_{t}$   given $(\mes, \fsy^{t-1},\out_t)$ under $\XPT{T}{\cdot}$ be identical to the conditional distribution of $\rsy_{t}$  given $(\mes, \fsy^{t-1},\out_t)$  under $\PX{\cdot}$, i.e. the original conditional distribution. 

 Note that as a result of equation (\ref{eq:SS1b})  and equation (\ref{eq:SS1a}) we have
\begin{align*}
\XETC{T}{ \ln \frac{\XPTC{T}{\out_t}{\fsy^{t-1}}}{\PCX{\out_t}{\mes=\dmes_1,\fsy^{t-1}}}}{\fsy^{t-1}} &\leq T
&&\mbox{and}&
\XETC{T}{ \ln \frac{\XPTC{T}{\out_t}{\fsy^{t-1}}}{\PCX{\out_t}{\mes=\dmes_2,\fsy^{t-1}}}}{\fsy^{t-1}} &\leq \obhexp{T}
&& \mbox{w.p.}1
\end{align*}
Now we make a standard measure change argument,
\begin{align}
\PCX{\out_t}{\mes=\dmes_1,\fsy^{t-1}}
\notag
&= e^{-\ln \frac{\XPTC{T}{\out_t}{\fsy^{t-1}}}{\PCX{\out_t}{\mes=\dmes_1,\fsy^{t-1}}}}   \XPTC{T}{\out_t}{\fsy^{t-1}}\\
\notag
&= e^{-\XETC{T}{\ln \frac{   \XPTC{T}{\out_t}{\fsy^{t-1}}}{\PCX{\out_t}{\mes=\dmes_1,\fsy^{t-1}}} }{\fsy^{t-1}}} e^{\chi_{t,\dmes_1}(\out_t|\fsy^{t-1}) }  \XPTC{T}{\out_t}{\fsy^{t-1}}\\
\label{eq:SS5b} 
&\geq e^{-T} e^{\chi_{t,\dmes_1}(\out_t|\fsy^{t-1}) }  \XPTC{T}{\out_t}{\fsy^{t-1}}  
\end{align}
where
\begin{equation}
\label{eq:SS5a}
  \chi_{t,\dmes_1}(\out_t|\fsy^{t-1})=\XETC{T}{\ln \tfrac{\XPTC{T}{\out_t}{\fsy^{t-1}}}{\PCX{\out_t}{\mes=\dmes_1,\fsy^{t-1}}}}{\fsy^{t-1}}  -\ln \tfrac{\XPTC{T}{\out_t}{\fsy^{t-1}}}{\PCX{\out_t}{\mes=\dmes_1,\fsy^{t-1}}}
 \end{equation} 
For  $\dmes=\dmes_1,\dmes_2$ let $\chi(\dmes)$ be
\begin{equation}
\label{eq:defchi}
\chi(\dmes)=\left\{\dfsy^{\blx}:|\sum\nolimits_{t=1}^{\blx}  \chi_{t,\dmes}(\out_t|\fsy^{t-1})|
 \leq 4  \sqrt{\blx} \ln\tfrac{1}{P_{min}}\right\}  
\end{equation}
For any event ${\sf B}$ measurable in the sigma field generated by ${\fsy}^{\blx}$ as a result of equation equations  (\ref{eq:SS5b}) we have
\begin{align}
\PX{\sf B}
\notag
&\geq \EX{ \IND{{\sf B}} \IND{\chi(\dmes_1)}}\\
\notag
&\geq  e^{-\blx T} e^{- 4  \sqrt{\blx} \ln\tfrac{1}{P_{min}} } \XET{T}{ \IND{{\sf B}} \IND{\chi(\dmes_1)}}\\
\label{eq:SS5c}
&\geq  e^{-\blx T} e^{- 4  \sqrt{\blx} \ln\tfrac{1}{P_{min}} } \XPT{T}{\{ {\sf B}  \mbox{~and~} \chi(\dmes_1)\}}
\end{align}
Following a similar line of reasoning we get,
\begin{equation}
\label{eq:SS6b}
\PCX{\out_t}{\mes=\dmes_2,\fsy^{t-1}}
\geq e^{-\blx \obhexp{T}}  e^{\chi_{t,\dmes_2}(\out_t|\fsy^{t-1}) }  \XPTC{T}{\out_t}{\fsy^{t-1}}  
\end{equation}
where 
\begin{equation}
\label{eq:SS6a}
\chi_{t,\dmes_2}(\out_t|\fsy^{t-1})=\XETC{T}{\ln \tfrac{\XPTC{T}{\out_t}{\fsy^{t-1}}}{\PCX{\out_t}{\mes=\dmes_2,\fsy^{t-1}}}}{\fsy^{t-1}}  -\ln \tfrac{\XPTC{T}{\out_t}{\fsy^{t-1}}}{\PCX{\out_t}{\mes=\dmes_2,\fsy^{t-1}}}
\end{equation}
and for any event ${\sf B}$  measurable in the sigma field generated by ${\fsy}^{\blx}$  we have
\begin{align}
\PX{\sf B}
\label{eq:SS6c}
&\geq  e^{-\blx \obhexp{T}} e^{- 4  \sqrt{\blx} \ln\tfrac{1}{P_{min}} } \XPT{T}{\{ {\sf B}  \mbox{~and~} \chi(\dmes_2)\}}.
\end{align}
Note that for  $\dmes=\{\dmes_1,\dmes_2\}$ and  $t\in \{1,2,\ldots,\blx\}$,
\begin{subequations}
 \label{eq:SS7}
\begin{align}
 \XETC{T}{\chi_{t,\dmes}(\out_t|\fsy^{t-1})}{\fsy^{t-1}}
&=0   
&&\forall \dfsy^{t-1}\in \fsyS^{t-1}\\
 \XETC{T}{(\chi_{t,\dmes}(\out_t|\fsy^{t-1}))^2}{\fsy^{t-1}}
&\leq 4 (\ln P_{min})^2
&&\forall \dfsy^{t}\in \fsyS^{t}\\
 \XETC{T}{ \chi_{t,\dmes}(\out_t|\fsy^{t-1})  \chi_{t-k,\dmes}(\out_{t-k}|\fsy^{t-k-1}) }{\fsy^{t-1}}
&=0   
&&\forall \dfsy^{t-1}\in \fsyS^{t-1}\quad \forall k\{1,2,\ldots,t-1\}
\end{align}
\end{subequations}
Thus as a result of equation (\ref{eq:SS7}), for  $\dmes \in \{\dmes_1,\dmes_2\}$
\begin{subequations}
 \label{eq:SS8}
\begin{align}
 \XET{T}{ \sum_{t=1}^{\blx} \chi_{t,\dmes}(\out_t|\fsy^{t-1})}
&=0\\
 \XET{T}{\left (\sum_{t=1}^{\blx} \chi_{t,\dmes}(\out_t|\fsy^{t-1}) \right)^2}
&\leq 4\blx (\ln P_{min})^2.
\end{align}
\end{subequations}
Using equation (\ref{eq:SS8}) and Chebychev's inequality we conclude that,
\begin{equation*}
  \XPT{T}{\chi(\dmes)}\geq 3/4  \qquad \dmes=\dmes_1,\dmes_2 
\end{equation*}
Hence,
\begin{equation*}
 \XPT{T}{  \chi(\dmes_1) \cap    \chi(\dmes_2)}\geq 1/2
\end{equation*}
Thus either the total probability of intersection of $\chi(\dmes_1) \cap    \chi(\dmes_2)$ with the decoding region of the second message is equal to or larger than  $1/4$ or the total probability of intersection of $\chi(\dmes_1) \cap    \chi(\dmes_2)$ with the decoding region of the first message is  strictly larger than  $1/4$. Then the lemma follows from equations  (\ref{eq:SS5c}) and (\ref{eq:SS6c}).

\end{proofs}
As we have noted previously $T_0$ does have an operational meaning it is the maximum error exponent first message can have, when the error probability of the second message is zero. 
\begin{lemma}\label{lem:zeroerxlem}
For any feedback encoding scheme with two messages, if $\Pem{\dmes_2}=0$  then $\Pem{\dmes_1} \geq e^{-\blx T_0}$. Furthermore there does exist  an encoding scheme such that  $\Pem{\dmes_2}=0$  then $\Pem{\dmes_1} = e^{-\blx T_0}$.
\end{lemma}
\begin{proof}
Let us use a construction similar to the one used in the proof of Lemma \ref{lem:E2=F2}
\begin{equation*}
  \XPTC{T}{\out_{t}}{\fsy^{t-1}}=\XQ{0}{\out_t}{\ENC_t(\dmes_1,\fsy^{t-1}),\ENC_t(\dmes_2,\fsy^{t-1})}.
\end{equation*}
Recall that
\begin{equation*}
\XQ{0}{\dout_t}{\dinp,\tilde{\dinp}}=\tfrac{\IND{\CT{\tilde{\dinp}}{\dout}>0}}{\sum_{\tilde{\dout}:\CT{\tilde{\dinp}}{\tilde{\dout}}>0}\CT{\dinp}{\tilde{\dout}}} \CT{\dinp}{\dout} 
\end{equation*}
Thus
\begin{align*}
\XPTC{T}{\out_{t}}{\fsy^{t-1}}&\leq e^{T_0}   \PCX{\out_{t}}{\mes=\dmes_1,\fsy^{t-1}}\\
\XPTC{T}{\out_{t}}{\fsy^{t-1}}&\leq \IND{ \PCX{\out_{t}}{\mes=\dmes_2,\fsy^{t-1}}>0}
\end{align*}
As we did in the proof of Lemma \ref{lem:E2=F2} we will assume that conditional distribution of $\rsy_t$ given $(\mes, \fsy^{t-1},\out_t)$ under $\XPT{T}{\cdot}$ is identical to the conditional distribution of $\rsy_{t}$  given $(\mes, \fsy^{t-1},\out_t)$  under $\PX{\cdot}$, i.e. the original conditional distribution. 

Then for any event ${\sf B}$ measurable in the sigma field generated by ${\fsy}^{\blx}$ we  have
 \begin{align}
\label{eq:SS10a}
\PCX{{\sf B}}{\mes=\dmes_1} 
&\geq e^{-\blx T_0}  \XPT{T}{{\sf B}}\\  
\label{eq:SS10b}
\PCX{{\sf B}}{\mes=\dmes_2}
&\geq  e^{\blx \ln P_{min}} \XPT{T}{{\sf B}}
 \end{align}
where  $P_{min}$ is the minimum non-zero element of $W$. 

  Since $\Pem{\dmes_2}\!=\!0$ equation (\ref{eq:SS10b}) implies that  $\XPT{T}{\est\!\neq \!\dmes_2}\!=\!0$ and $\XPT{T}{\est \neq \dmes_1}\!=1$. Using  this fact together with  equation (\ref{eq:SS10a}) we conclude that
\begin{equation}
  \label{eq:E2=F2N}
\Pem{\dmes_1} \geq e^{-\blx T_0}.
\end{equation}
 Let us assume that maximizing x-pair  in (\ref {eq:tzero})   is $(\dinp_1^*,\dinp_2^*)$ i.e. $T_0= -\ln \sum_{\dout:\CT{\dinp_2^*}{\dout}>0} \CT{\dinp_1^*}{\dout}$. If the the encoding scheme sends  $\dinp_1^*$ for      the first message and $\dinp_2^*$ for the second message, and  the decoder  decodes to second message unless $\out_t=\dout^*$ for some  $t\in \{1,2,\ldots,\blx\}$ and for some  $\dout^*$ such that $\CT{\dinp_2^*}{\dout^*}=0$. Then ${\Pe}_{\dmes_2}=0$ and ${\Pe}_{\dmes_1}=e^{-\blx T_0}$.
\end{proof}

\subsection{Convexity of  $\EXo (R,\EXa, \ts,P, \XSID)$ in $\ts$:}\label{app:convexityints}
\begin{lemma}
\label{lem:precon}
  For any  probability distribution $P$ on input alphabet $\inpS$,  $\cesp(P,Q,R)$ is convex in $(Q,R)$  pair. 
\end{lemma}
\begin{proof}
Note that
\begin{equation*}
\gamma \cesp(R_a,P,Q_a)+(1-\gamma) \cesp(R_b,P,Q_b)
=\min_{V_a,V_b:\substack{ \MI{P}{V_a}\leq R_a~~\MI{P}{V_b}\leq R_b\\ \Ymar{P V_a}=Q_a~~\Ymar{P V_b}=Q_b}}  \gamma \CDIV{V_a}{W}{P}+(1-\gamma) \CDIV{V_b}{W}{P}
\end{equation*}
Using the convexity of $\CDIV{V}{W}{P}$  in $V$ and Jensen's inequality we get,
\begin{equation*}
\gamma \cesp(R_a,P,Q_a)+(1-\gamma) \cesp(R_b,P,Q_b)
\geq \min_{V_a,V_b:\substack{ \MI{P}{V_a}\leq R_a~~\MI{P}{V_b}\leq R_b\\ \Ymar{P V_a}=Q_a~~\Ymar{P V_b}=Q_b}} \CDIV{V_\gamma}{W}{P}
\end{equation*}
where $V_{\gamma}=\gamma V_a+(1-\gamma)V_b$.

If the set that a minimization is done over is enlarged, then the resulting minimum does not increase. Using this fact together with the convexity of $\MI{P}{V}$ in $V$  and  Jensen's inequality we get,
 \begin{align*}
\gamma \cesp(R_a,P,Q_a)+(1-\gamma) \cesp(R_b,P,Q_b)
&\geq\min_{V_\gamma :\substack{ \MI{P}{V_\gamma}\leq R_\gamma \\ \Ymar{P V_\gamma}=Q_\gamma}}  \CDIV{V_\gamma}{W}{P}\\
&=\cesp(R_{\gamma},P,Q_{\gamma})
\end{align*}
where $R_{\gamma}=\gamma R_a+(1-\gamma)R_b$, $Q_{\gamma}=\gamma Q_a+(1-\gamma)Q_b$.
\end{proof}

\begin{lemma}
\label{lem:conveityints}
  For all $(R,\EXa,P,\XSID)$ quadruples such that $\rexp(R,P)\geq\EXa$,  $\EXo (R,\EXa, \ts,P, \XSID)$ is a convex function of $\ts$ on the interval $[\ts^*(R, \EXa, P),1]$ where $\ts^*(R, \EXa, P)$ is the unique solution\footnote{The equation $\ts\rexp(\tfrac{R}{\ts},P)=0$ has multiple solutions; we choose the minimum of those to be the $\ts^*$ i.e., $\ts^*(R, 0, P)=\tfrac{R}{\MI{P}{W}}$.} of $\ts\rexp(\tfrac{R}{\ts},P)=\EXa$.
\end{lemma}
\begin{proof}
For any $P$ such that $\rexp(R,P)$ is non-negative, convex and decreasing function of $R$ in the interval $[0,\MI{P}{W}]$. Thus $\ts E_r(\tfrac{R}{\ts},P)$ is strictly increasing continuous function of $\ts \in [\tfrac{R}{\MI{P}{W}},1]$.  Furthermore for $\ts=\tfrac{R}{\MI{P}{W}}$, $\ts E_r(\tfrac{R}{\ts},P)=0$ and  for  $\ts=1$, $\ts E_r(\tfrac{R}{\ts},P)  \geq \EXa$. Thus $\ts\rexp(\tfrac{R}{\ts},P)=\EXa$ has a unique solution.

Note that for any $\gamma \in [0,1]$
\begin{align*}
\notag
\gamma  \EXo (R,\EXa, \ts_a, P, \XSID) &+ (1-\gamma) \EXo (R,\EXa, \ts_b, P, \XSID) \\
&=\hspace{-.5cm} \min_{ 
\substack{ 
\hspace{.7cm}Q_a, R_{1a}, R_{2a}, T_a, Q_b, R_{1b}, R_{2b}, T_b:\\
\hspace{.3cm}R_{1a} \geq R_{2a}\geq R~~T_{a} \geq 0\\ 
\hspace{.3cm}R_{1b} \geq R_{2b}\geq R~~T_{b} \geq 0\\ 
\hspace{.7cm}\ts_a \cesp(\frac{R_{1a}}{\ts_a},P,Q_a)+R_{2a}-R+T_a\leq \EXa\\
\hspace{.7cm}\ts_b \cesp(\frac{R_{1b}}{\ts_b},P,Q_b)+R_{2b}-R+T_b\leq \EXa}}
\hspace{-0.5cm}
\begin{array}{l}
\gamma\left[\ts_a\cesp(\tfrac{R_{2a}}{\ts_a},P,Q_a)+R_{1a}-R+ (1-\ts_a)\bhexp{\tfrac{T_a}{1-\ts_a}}{\XSID} \right]\\
\qquad +(1-\gamma)\left[\ts_b\cesp(\tfrac{R_{2b}}{\ts_b},P,Q_b)+ R_{1b}-R+(1-\ts_b)\bhexp{\tfrac{T_b}{1-\ts_b}}{\XSID} \right]
\end{array}
\\
&\geq \hspace{-1.5cm} \min_{ 
\substack{
\hspace{-.1cm}Q_\gamma, R_{1\gamma}, R_{2\gamma} T_{\gamma}:\\ 
\hspace{.1cm}R_{1\gamma} \geq R_{2\gamma} \geq R~~T_{\gamma}\geq 0 \\
\hspace{1.7cm}\ts_{\gamma}\cesp(\frac{R_{1\gamma}}{\ts_{\gamma}},P,Q_{\gamma})+ R_{2\gamma}-R +T_{\gamma}  \leq \EXa}}
\hspace{-1cm}
\ts_{\gamma}\cesp(\tfrac{R_{2\gamma}}{\ts_{\gamma}},P,Q_{\gamma})+ R_{1\gamma}-R + (1-\ts_{\gamma})\bhexp{\tfrac{T}{1-\ts_{\gamma}}}{\XSID}\\
&= \EXo (R,\EXa, \ts_{\gamma}, P, \XSID). 
\end{align*}
where $\ts_{\gamma}$, $T_{\gamma}$, $Q_{\gamma}$, $R_{1\gamma}$ and $R_{2\gamma}$  are given by,
\begin{align*}
\ts_{\gamma}&=\gamma \ts_a+(1-\gamma) \ts_b
&T_{\gamma}&=\gamma T_a+(1-\gamma) T_b
&Q_{\gamma}&=\tfrac{\gamma\ts_a}{\ts_{\gamma}}Q_a+\tfrac{(1-\gamma)\ts_b}{\ts_{\gamma}} Q_b\\
R_{1\gamma}&=\gamma R_{1a} +(1-\gamma) R_{1b}
&R_{2\gamma}&=\gamma R_{2a} +(1-\gamma) R_{2b}
& ~&
\end{align*}
The inequality follows from convexity arguments analogous to the ones used in the  proof of Lemma \ref{lem:precon}.
\end{proof}

\subsection{  $\max_{\XSID}\EXo (R,\EXa,\ts, P,\XSID) > \max_{\XSID}\EXo(R,\EXa,1, P,\XSID)$,~~ $\forall P\in \PSet{R}{\EXa}{\ts}$}\label{app:specialcase}
Let us first consider a control phase type $\XSID_{P}(\dinp_1,\dinp_2)=\frac{P(\dinp_1)P(\dinp_2) \IND{\dinp_1 \neq \dinp_2}}{1-\sum_{\dinp} (P(\dinp))^2}$ and establish,
\begin{equation}
\label{eq:specopta}
  \EXo (R,\EXa,\ts, P, \XSID_{P}) >  \EXo (R,\EXa,1, P, \XSID_{P}) \qquad \forall P \in   \PSet{R}{\EXa}{\ts}
\end{equation}
First consider 
\begin{align}
\CDIV{\XSOD}{W_a}{\XSID_{P}}
\notag
&=\tfrac{1}{1-\sum_{\dinp} (P(\dinp))^2} \sum_{\dinp_1,\dinp_2:\dinp_1\neq  \dinp_2}  P(\dinp_1)P(\dinp_2) \sum\nolimits_{\dout} \XSOD(\dout|\dinp_1,\dinp_2) \log \tfrac{\XSOD(\dout|\dinp_1,\dinp_2)}{\CT{\dinp_1}{\dout}}\\
\notag
&=\tfrac{1}{1-\sum_{\dinp} (P(\dinp))^2} \sum_{\dinp_1,\dinp_2:\dinp_1\neq  \dinp_2}  P(\dinp_1)P(\dinp_2) \sum\nolimits_{\dout} \XSOD(\dout|\dinp_1,\dinp_2) 
 \left[\log \tfrac{\XSOD(\dout|\dinp_1,\dinp_2)}{V_{\XSOD}(\dout|\dinp_1)}-
\log \tfrac{V_{\XSOD}(\dout|\dinp_1)}{\CT{\dinp_1}{\dout}} \right]\\
\label{eq:appb1}
&\geq \tfrac{1}{1-\sum_{\dinp} (P(\dinp))^2} \left[\MI{P}{\hat{V}_{\XSOD}}+\CDIV{V_{\XSOD}}{W}{P} \right]
\end{align}
where the last step follows from the log sum inequality and transition probability matrices $V_{\XSOD}$ and  $\hat{V}_{\XSOD}$ are given by 
\begin{align*}
V_{\XSOD}(\dout|\dinp_1)&=\CT{\dinp_1}{\dout} P(\dinp_1)  +\sum\nolimits_{\dinp_2:\dinp_2\neq \dinp_1} \XSOD(\dout|\dinp_1,\dinp_2) P(\dinp_2)\\
\hat{V}_{\XSOD}(\dout|\dinp_2)&=\CT{\dinp_2}{\dout} P(\dinp_2) +\sum\nolimits_{\dinp_1:\dinp_1\neq \dinp_2} \XSOD(\dout|\dinp_1,\dinp_2) P(\dinp_1).  
\end{align*}
Using a similar line of reasoning  we get, 
\begin{equation}
\label{eq:appb2}
 \CDIV{\XSOD}{W_r}{\XSID_{P}} \geq \tfrac{1}{1-\sum_{\dinp} (P(\dinp))^2} \left[ \CDIV{\hat{V}_{\XSOD}}{W}{P}+\MI{P}{V_{\XSOD}} \right]
\end{equation}
Note that for all $P\in \PSet{R}{\EXa}{\ts}$ if use the inequalities (\ref{eq:appb1}) and (\ref{eq:appb2}) together the definition of $ \EXo$ given in equation (\ref{eq:prand}) and  (\ref{eq:thm1}) we get, 
\begin{equation*}
  \EXo (R,\EXa,\ts(R,\EXa), P, \XSID_{P}) \geq  \EXo (R,\EXa,1, P, \XSID_{P})  +\delta_{P}
\end{equation*}
for some $\delta_P>0$. Consequently for all $P\in \PSet{R}{\EXa}{\ts}$, equation (\ref{eq:specopta}) holds. 

Note that  for all $\XSID$ and for all $P\in \PSet{R}{\EXa}{\ts}$
\begin{equation*}
  \EXo (R,\EXa,1, P, \XSID_{P})=  \EXo (R,\EXa,1, P, \XSID). 
\end{equation*}
Thus we have:
\begin{equation}
  \label{eq:specopt}
\max_{\XSID}\EXo (R,\EXa,\ts, P,\XSID) > \max_{\XSID}\EXo(R,\EXa,1, P,\XSID) \qquad  \forall P\in \PSet{R}{\EXa}{\ts}.
\end{equation}

\input{main.bbl}

\end{document}

%% file: main.bbl
\newcommand{\noopsort}[1]{} \newcommand{\printfirst}[2]{\#1}
  \newcommand{\singleletter}[1]{\#1} \newcommand{\switchargs}[2]{\#2\#1}